\documentclass[a4paper]{article}

\usepackage[margin=1in]{geometry}
\usepackage[affil-it]{authblk}

\usepackage{microtype}

\usepackage[dvips]{color}
\usepackage[pdftex]{graphicx}
\usepackage{paralist,tabularx}
\usepackage{amsmath,amsfonts,amssymb,amsthm}
\usepackage{stmaryrd}
\usepackage{latexsym}
\usepackage{enumitem}
\setitemize{itemsep=0pt,topsep=2pt,parsep=0pt,partopsep=0pt}
\setenumerate{itemsep=0pt,topsep=2pt,parsep=0pt,partopsep=0pt}
\usepackage{mathtools}
\usepackage{wrapfig}
\usepackage{thmtools,thm-restate}
\usepackage{mdframed}
\usepackage{tikz}

\usepackage{algorithm, algorithmicx, algpseudocode} 
\usepackage[algo2e]{algorithm2e}

\usepackage[colorlinks=true]{hyperref}
\usepackage[capitalize]{cleveref}

\newtheorem{theorem}{Theorem}[section]

\newtheorem{lemma}[theorem]{Lemma}
\newtheorem{observation}[theorem]{Observation}

\newtheorem{proposition}[theorem]{Proposition}
\newtheorem{corollary}[theorem]{Corollary}

\theoremstyle{definition}
\newtheorem{definition}[theorem]{Definition}

\newenvironment{reminder}[1]{\smallskip
	\noindent {\scshape \textbf{Reminder of #1.}}\em}{
}

\DeclareMathOperator*{\polylog}{polylog}

\newcommand{\eps}{\epsilon}

\newcommand{\R}{\mathbb R}

\renewcommand{\P}{\mathbf{P}}
\newcommand{\B}{\mathbf{B}}
\newcommand{\F}{\mathbf{F}}

\newcommand{\ComBkAlig}{$k$-Shift  Distance}

\newcommand{\CBKA}{k\text{-}SD}
\newcommand{\CBtA}{2\text{-}SD}

\newcommand{\klcs}[1][]{\ifthenelse{\equal{#1}{}}{$k$-LCS}{${#1}$-LCS}}

\newcommand{\kwlcs}[1][]{\ifthenelse{\equal{#1}{}}{$k$-WLCS}{${#1}$-WLCS}}

\newcommand{\kNLstC}[1][]{\ifthenelse{\equal{#1}{}}{$k$-NLstC}{${#1}$-NLstC}}

\newcommand{\kELstC}[1][]{\ifthenelse{\equal{#1}{}}{$k$-ELstC}{${#1}$-ELstC}}

\newcommand{\czkc}[1][]{\ifthenelse{\equal{#1}{}}{FZ$k$C}{FZ${#1}$C}}

\newcommand{\czkch}[1][]{\ifthenelse{\equal{#1}{}}{FZ$k$CH}{FZ${#1}$CH}}

\newcommand{\cfkc}[1][]{\ifthenelse{\equal{#1}{}}{F$\mathfrak{f}k$C}{F$\mathfrak{f}{#1}$C}}

\newcommand{\cfkch}[1][]{\ifthenelse{\equal{#1}{}}{F$\mathfrak{f}k$CH}{F$\mathfrak{f}{#1}$CH}}

\newcommand{\ckov}[1][]{\ifthenelse{\equal{#1}{}}{F$k$-OV}{F${#1}$-OV}}

\newcommand{\ckovh}[1][]{\ifthenelse{\equal{#1}{}}{F$k$-OVH}{F${#1}$-OVH}}

\newcommand{\cksum}[1][]{
\ifthenelse{\equal{#1}{}}
{F$k$-SUM}
{F${#1}$-SUM}
}

\newcommand{\cksumh}[1][]{\ifthenelse{\equal{#1}{}}{F$k$-SUMH}{F${#1}$-SUMH}}

\newcommand{\fzkc}[1][]{\ifthenelse{\equal{#1}{}}{FZ$k$C}{FZ${#1}$C}}

\newcommand{\ckxor}[1][]{\ifthenelse{\equal{#1}{}}{F$k$-XOR}{F${#1}$-XOR}}

\newcommand{\ckxorh}[1][]{\ifthenelse{\equal{#1}{}}{F$k$-XORH}{F${#1}$-XORH}}

\newcommand{\ckfunc}[1][]{\ifthenelse{\equal{#1}{}}{F$k$-$\mathfrak{f}$}{F${#1}$-$\mathfrak{f}$}}

\newcommand{\ckfunch}[1][]{\ifthenelse{\equal{#1}{}}{F$k$-$\mathfrak{f}$H}{F${#1}$-$\mathfrak{f}$H}}

\newcommand{\ksum}[1][]{\ifthenelse{\equal{#1}{}}{$k$-SUM}{${#1}$-SUM}}

\newcommand{\Flist}{\mathbf{List}}

\newcommand{\Oh}{O}
\newcommand*{\Ohtilde}{
  \setbox0=\hbox{$\Oh$}
  \smash{\Oh}\mathllap{\widetilde{\phantom{\rule{6pt}{\ht0+{-0.5pt}}}}}%
}
\newcommand{\dd}{\mathinner{.\,.}}
\newcommand{\DIST}{DIST}
\newcommand{\ID}{\delta_{D}}
\newcommand{\ED}{\delta_E}
\newcommand{\HD}{\delta_H}
\newcommand{\CED}{\delta_{CE}}
\newcommand{\LCS}{\ensuremath{\mathrm{LCS}}}
\newcommand{\lcsf}{\LCS}
\newcommand{\flcs}{\lcsf}
\newcommand{\inv}{\mathrm{in}}
\newcommand{\outv}{\mathrm{out}}
\newcommand{\emptystring}{\gamma}
\newcommand{\conv}{\circledast}
\newcommand{\floor}[1]{\lfloor #1 \rfloor}
\DeclareMathOperator*{\argmin}{arg\,min}

\newcommand{\G}{\mathbf{G}}
\renewcommand{\S}{\mathbf{S}}
\newcommand{\N}{\mathbf{N}}
\newcommand{\rhs}{\mathsf{rhs}}
\newcommand{\sub}{\subseteq}
\newcommand{\sm}{\setminus}

\begin{document}

\title{\Large How Compression and Approximation Affect Efficiency\\ in String Distance Measures}
\author{Arun Ganesh\thanks{Supported in part by an NSF 1816861 grant.}}
\author{Tomasz Kociumaka\thanks{Supported in part by NSF 1652303, 1909046, and HDR TRIPODS 1934846 grants, and an Alfred P. Sloan Fellowship.}}
\author{Andrea Lincoln\protect\footnotemark[2]\;\thanks{Supported in part by a Simons NTT Research Fellowship.}}
\author{Barna Saha\protect\footnotemark[2]}

\affil{University of California, Berkeley}
\date{}

\maketitle

\begin{abstract}
Real-world data often comes in compressed form.
Analyzing compressed data directly (without first decompressing it) can save space and time by orders of magnitude. 
In this work, we focus on fundamental sequence comparison problems and try to quantify the gain in time complexity when the underlying data is highly compressible. 
We consider grammar compression, which unifies many practically relevant compression schemes such as the Lempel--Ziv family, dictionary methods, and others.
For two strings of total length $N$ and total compressed size $n$, it is known that the edit distance and a longest common subsequence (LCS) can be computed exactly in time $\Ohtilde(nN)$, as opposed to $O(N^2)$ for the uncompressed setting.
Many real-world applications need to align multiple sequences simultaneously, and the fastest known exact algorithms for median edit distance and LCS of $k$ strings run in $\Oh(N^k)$ time, whereas the one for center edit distance has a time complexity of $\Oh(N^{2k})$. This naturally raises the question if compression can help to reduce the running time significantly for $k \geq 3$, perhaps to $O(N^{k/2}n^{k/2})$ or, more optimistically, to $O(Nn^{k-1})$.\footnote{In this paper, we assume that $k$ is a constant; thus, the $\Oh(\cdot)$ and $\Omega(\cdot)$ notation may hide factors with exponential dependence on $k$.}

Unfortunately, we show new lower bounds that rule out any improvement beyond $\Omega(N^{k-1}n)$ time for any of these problems assuming the Strong Exponential Time Hypothesis (SETH), where again $N$ and $n$ represent the total length and the total compressed size, respectively. This answers an open question of Abboud, Backurs, Bringmann, and K{\"{u}nnemann} (FOCS'17).

In presence of such negative results, we ask if allowing approximation can help, and we show that approximation and compression together can be surprisingly effective for both multiple and two strings.

We develop an $\Ohtilde(N^{k/2}n^{k/2})$-time FPTAS for the median edit distance of $k$ sequences, leading to a saving of nearly half the dimensions for highly-compressible sequences.  
In comparison, no $O(N^{k-\Omega(1)})$-time PTAS is known for the median edit distance problem in the uncompressed setting. We obtain an improvement from $\Ohtilde(N^{2k})$ to $\Ohtilde(N^{k/2+o(k)}n^{k/2})$ for the center edit distance problem.
For two strings, we get an $\Ohtilde(N^{2/3}n^{4/3})$-time FPTAS for both edit distance and LCS; note that this running time is $o(N)$ whenever $n \ll N^{1/4}$. In contrast, for uncompressed strings, there is not even a subquadratic algorithm for LCS that has less than polynomial gap in the approximation factor. Building on the insight from our approximation algorithms, we also obtain several new and improved results for many fundamental distance measures including the edit, Hamming, and shift distances.
\end{abstract}

\section{Introduction}\label{sec:intro}

With the information explosion, almost all real-world data comes in a compressed form.
While compression is primarily intended to save storage space and transmission bandwidth,
processing compressed data directly often provides an opportunity to reduce computation time and energy by several orders of magnitude.
In this work, we focus on sequential data such as natural-language texts, biological sequences (nucleic acid sequences, including DNA, and amino acid sequences, including proteins), and computer codes.
Sequential data often contains highly repetitive pattern. Modern technology (e.g., high-throughput sequencing) has led to an astonishingly rapid accumulation of such data, so much so that without proper data compression and algorithms over compressed data, it is not possible to utilize the wealth of information in them~\cite{berger2013computational, berger2016computational, pavlichin2018desperate,greenfield2019, hernaez2019genomic}. 

Grammar compression represents strings as \emph{straight-line programs} (SLPs), and provides a mathematically elegant way to unify algorithm design principles for processing compressed data~\cite{lohrey2012}. It is equivalent to many well-known compression schemes up to logarithmic factors and moderate constants~\cite{Rytter2003,KP18,KK20} such as the celebrated LZ77~\cite{ziv1977} and RLBWT~\cite{BWT} schemes, and at least as strong as
byte-pair encoding~\cite{gage1994}, Re-Pair~\cite{Larsson_2000}, Sequitur~\cite{nevill1997}, further members of the Lempel--Ziv family~\cite{lz78,welch1984}, and many more popular schemes (the list keeps growing).
Therefore, following the lead of a large body of previous work (including~\cite{Tiskin15,Jez15,compressedLCSSETH,bringmann2019,CKW20}), we work with grammar-compressed data.

In this work, we ask whether fundamental sequence similarity measures can be computed faster for compressed data.
This research is motivated in part by the success of computing edit distance and longest common subsequence (LCS) of two strings~\cite{Gawrychowski12,HermelinLLW13,Tiskin15} much faster than the ``decompress-and-solve'' approach: 
If we let $N$ denote the total length and $n$ denote the total compressed size of the input strings, then the edit distance and the LCS length can be computed exactly in time $\Ohtilde(nN)$ in contrast to $O(N^2)$ time for the uncompressed setting.
Therefore, for highly compressible sequences where, say, $n=\polylog N$, the running time reduces to $\Ohtilde(N)$. 
Abboud, Backurs, Bringmann, and K\"{u}nnemann~\cite{compressedLCSSETH} asked whether it is possible to improve upon $\Ohtilde(nN)$, noting that: ``For example, an $O(n^2 N^{0.1})$ bound could lead to major real-world improvements.''
In general, any sublinear dependency on $N$ would be preferable;
unfortunately,~\cite{compressedLCSSETH} shows that $\Ohtilde(Nn)$ is essentially optimal under the Strong Exponential Time Hypothesis (SETH). 

There are many real-world applications which deal with multiple sequences. 
A survey by Nature~\cite{VMN:14} reports \emph{multiple sequence alignment} as one of the most used modeling methods in biology, with~\cite{THG:94} among the top-10 papers cited of all time (citation count 63105).
Some of the basic measures for multiple sequence similarity include the LCS length and the cost of the median and center strings under edit distance.
Abboud, Backurs, and V.-Williams~\cite{LCSisHard} showed that exact computation of $k$-LCS requires $\Omega(N^{k-o(1)})$ time (under SETH), and a similar result has been recently shown for both median and center $k$-edit distance~\cite{editDistLBSeth}.
A simple extension of the basic dynamic programming for two strings
solves the median $k$-edit distance problems in $\Oh(N^k)$ time whereas the best bound known for the center $k$-edit distance is
$O(N^{2k})$~\cite{NicolasRivals05}.
The two-string lower bound in the compressed setting leaves open the possibility of reducing the running times of the $k$-LCS and the median $k$-edit distance problems for compressed strings: 
It might be feasible to achieve runtimes of $O(N^{k/2}n^{k/2})$ or even $O(Nn^{k-1})$, and a substantial reduction of the exponent at $N$ could lead to significant savings. 
This raises~the~following~questions:%dirty squeezing tricks
\begin{enumerate}
\item\emph{Does compression allow for significantly reducing the running time for multi-sequence similarity problems?}
\item\emph{For the case of two highly compressible strings, what relaxations of the LCS and the edit distance problems could allow circumventing the lower bounds and achieving sublinear dependency on $N$?}
\end{enumerate}

\subsubsection*{Lower Bounds: Compression does not help with exact bounds much!} Unfortunately, we show that computing the $k$-LCS, median $k$-edit distance, and center $k$-edit distance all require $\Omega((N^{k-1}n)^{1-o(1)})$ time under SETH. Therefore, the potential gain from compression becomes insignificant as $k$ grows. 
Intuitively, SETH states that CNF-satisfiability requires $2^{n-o(n)}$ time~\cite{cseth}. Even more specifically, we use the $k$-Orthogonal Vectors problem ($k$-OV)~\cite{virgiSurvey}. 
At a high level, $k$-OV takes as input a list $L$ with $n$ zero-one vectors of dimension $d$. 
We must return YES if there exist $k$ vectors that, when multiplied element-wise, form the all zeros vector.
The $k$-OV conjecture, which is implied by SETH, states that $k$-OV cannot be solved in $O(n^{k-\Omega(1)})$ time.

\begin{theorem}\label{thm:compkLCSHardFromseth}

If the $k'$-OV hypothesis is true for all constants $k'$, then for any constant $\epsilon \in (0,1]$ grammar-compressed $k$-LCS requires $\left(N^{k-1}n\right)^{1-o(1)}$ time when the alphabet size is $|\Sigma|=\Theta(k)$ and $n=N^{\epsilon \pm o(1)}$. Here, $N$ denotes the total length of the $k$ input strings and $n$ is their total compressed size.
\end{theorem}

We prove similar lower bounds for median and center $k$-edit distance (Theorem~\ref{thm:editDistanceLowerBoundFromSETH} and Theorem~\ref{thm:centerEditDistanceAllSameLength}). Sections~\ref{sec:kLCSLB},~\ref{sec:EditDistLB}, and~\ref{sec:SETHEditDistLB} contain our lower bound results.

Abboud, Backurs, Bringmann, and K\"{u}nnemann~\cite{compressedLCSSETH} left an open question whether their $\Omega((Nn)^{1-o(1)})$ lower bound for LCS also holds for computing the edit distance of two strings.
We answer this question affirmatively and extend the argument to the $k$-string setting. Moreover, we note that for a seemingly simpler problem of computing the shift distance~\cite{AndoniIK07,AndoniIKH13,andoni2013,golan2020}, we show that compression does not help to reduce even a single dimension (Section~\ref{sec:BestAlign}).~\\

\subsubsection*{Algorithms: Effectiveness of Approximation \& Compression.} In presence of such negative results, relaxing the median and center $k$-edit distance to circumvent the $\Omega((N^{k-1}n)^{1-o(1)})$ lower bound becomes even more important. 

\begin{center}
  \emph{Can we use compression and approximation together to achieve much better approximation guarantees and, simultaneously, circumvent the exact computation lower bounds?} 
\end{center}

To the best of our knowledge, even for two strings, there is no previous work on approximating
the edit distance of grammar-compressed strings. On the other hand, even after a long line of research in developing fast algorithms for approximate edit distance for the uncompressed setting (see e.g.~\cite{BEKMRRS03,BJKK04,BES06,AKO10,AO12,CDGKS18,GRS:20,BR20,KS20,AN:20}),
the best approximation ratio achievable in truly subquadratic time is currently $3+\eps$~\cite{GRS:20},
and the fastest constant-factor approximation algorithm runs in $\Oh(n^{1+\eps})$ time~\cite{AN:20} with an approximation factor that has doubly-exponential dependence on $\frac{1}{\epsilon}$.
%For two strings case, despite significant progress in developing fast algorithms for approximate edit distance~\cite{BEKMRRS03,BJKK04,BES06,AKO10,AO12,BEGHS18,CDGKS18,GRS:20,KS20,BR20,AN:20}, so far there does not exist an $(1+\epsilon)$ approximation in truly subquadratic time where $\epsilon >0$ is any arbitrary constant. The best known bound stands at $3+\epsilon$ for subquadratic algorithms~\cite{GRS:20}, and is doubly-exponential in $\frac{1}{\epsilon}$ when running time improves to  $O(N^{1+\epsilon})$~\cite{AN:20}. 
The situation is even worse for LCS approximation, where we do not know how to design a subquadratic algorithm with sub-polynomial approximation gap~\cite{HSSS:19,RSSS:19}.
We are also unaware of any previous research on approximating $\LCS$ of two compressed strings.
In the case of multiple strings, there is a classic $\Oh(N^2)$-time $(2-2/k)$-approximation for median edit distance and an $\Oh(N^2)$-time $2$-approximation for center edit distance~\cite{Gusfield1997}. Combined with the results of~\cite{AN:20}, this yields an $\Oh(N^{1+\eps})$-time constant-factor approximation for both versions.
Nevertheless, a PTAS, that is, a $(1+\epsilon)$-approximation algorithm for every constant $\epsilon >0$, would be much more desirable for practical applications.
%In general, the following question arises:

Surprisingly, we show that already when an $(1+\epsilon)$-approximation is allowed for an arbitrary constant $\epsilon >0$, the median $k$-edit distance computation time reduces to $\Ohtilde(N^{k/2}n^{k/2})$ compared to the $\Omega((N^{k-1}n)^{1-o(1)})$ lower bound for exact algorithms.
In other words, we can save $k/2$ dimensions by allowing approximation and compression. 
For $\eps = o(1)$, the running time of our algorithm increases by an $\eps^{-\Oh(k)}$ factor,
so we even obtain an FPTAS whereas no prior work in the uncompressed setting gives a $(1+\epsilon)$-approximation in $O(N^{k-\Omega(1)})$ time. The reduction in time for center $k$-edit distance is even more dramatic (and technically more complex) from exact $O(N^{2k})$ to $\Ohtilde(N^{k/2+o(k)}n^{k/2})$ for a $(1+\epsilon)$ approximation.

For edit distance between two strings, we develop a more efficient FPTAS whose running time is $\Ohtilde(N^{2/3}n^{4/3}\eps^{-1/3})$, which is  sublinear in $N$ as long as $n \ll N^{1/4}$. A slightly more sophisticated $\Ohtilde(N^{2/3}n^{4/3}\eps^{-1/3})$-time algorithm also provides a $(1+\eps)$-approximation of the LCS length.
In contrast, a comparable result for the uncompressed setting is an $O(N^{1.95})$-time algorithm of~\cite{RSSS:19},
which returns a common subsequence of length $\Omega(N/\lambda^4)$, providing an $\Oh(\lambda^3)$-factor approximation. Even when the alphabet size is $2$, so far, there does not exist any $(1+\eps)$ approximation in subquadratic time~\cite{RS:20}.

%As we show in the following theorem, the compressed setting allows achieving a significantly
%better approximation ratio with a polynomial-factor speed-up compared to computing $\LCS(X,Y)$ exactly.
%\tknote{End: Moved from Overview}

\subsubsection*{Improved Exact Algorithms in Compressed Setting.} Interestingly, the insights behind our approximation algorithms also lead to new \emph{exact algorithms}.
In particular, we show that the edit distance can be computed in time $\Ohtilde(n\sqrt{ND})$, where $D$ is an upper bound on the edit distance. 
This improves upon the state-of-the-art bound of $\Ohtilde(\min(nN, n+D^2))$~\cite{Tiskin15,LV88,Mehlhorn1997} whenever $D \gg N^{1/3}n^{2/3}$.

For this problem, the first improvements compared to the uncompressed settings were given in~\cite{Tiskin09,HermelinLLW09}. Then, Tiskin~\cite{Tiskin15} obtained an $\Oh(nN\log N)$-time algorithm and
subsequent works~\cite{HermelinLLW13,Gawrychowski12} reduced the $\Oh(\log N)$ factor.
However, when the distance $D$ is small, the edit distance can be computed in $\Oh(N+D^2)$ time~\cite{LV88} in the uncompressed setting.
The $\Oh(N)$ term in the running time of the Landau--Vishkin algorithm~\cite{LV88} is solely needed to construct a data structure efficiently answering the Longest Common Extension (LCE) queries. 
However, already the results of Mehlhorn, Sundar, and Uhrig~\cite{Mehlhorn1997} yield $\Ohtilde(1)$-time LCE queries after $\Ohtilde(n)$-time preprocessing of the grammars representing $X$ and $Y$.
This gives rise to an $\Ohtilde(n+D^2)$-time algorithm computing the edit distance. With a more modern implementation of LCE queries in compressed strings by I~\cite{I17},
the factor hidden within the $\Ohtilde(\cdot)$ notation can be reduced to $\Oh(\log N)$. 

While the $\Ohtilde(n+D^2)$-time algorithm is very fast if $D$ is small, its efficiency quickly
degrades with increasing $D$ and  the $\Ohtilde(nN)$-time algorithm becomes
more suitable already for $D \gg \sqrt{nN}$.
With a time complexity of $\Ohtilde(n\sqrt{ND})$, our algorithm improves upon the previous algorithms
whenever $\sqrt[3]{n^2N} \ll \ED(X,Y) \ll N$.
%\tknote{End: Moved from Overview}
Nevertheless, the current lower bounds allow for a hypothetical holy-grail algorithm achieving the running time of $\Ohtilde(\min{(nD, n+D^2)})$ which we leave as an interesting open question.

We also get improved results for the Hamming distance, which is a more basic measure trivially computable in $\Oh(N)$ time.
Here, we present an $\Oh(n\sqrt{N})$-time algorithm which improves upon the $O(n^{1.41}N^{0.593})$ bound of~\cite{compressedLCSSETH}.
Additionally, we note that natural~generalizations to multiple strings (including the median Hamming distance) can be computed in $O(n N^{1-1/k})$ time.

\section{Preliminaries}\label{sec:prelims}

For two integers $i\le j$, we write $[i\dd j]$ to denote the set $\{i,\ldots,j\}$
and $[i\dd j)$ to denote $\{i,\ldots,j-1\}$. The notions $(i\dd j]$ and $(i\dd j)$
are defined analogously.

A \emph{string} is a sequence of characters from a fixed alphabet $\Sigma$.
We write $\Sigma^*$ to denote the set of all strings over $\Sigma$,
and we define $\Sigma^+=\Sigma^* \setminus \{\emptystring\}$, where $\emptystring$ denotes the empty string.
The length of a string $X$ is denoted by $|X|$ and, for a position $i\in [1\dd |X|]$ in $X$,
the character of $X$ at position $i$ is denoted by $X[i]$.
For an integer $N\ge 0$, the set of length-$N$ strings over $\Sigma$ is denoted by $\Sigma^N$.

For two positions $i\le j$ in $X$, we write $X[i\dd j]$ to denote the fragment of $X$
starting at positions $i$ and ending at position $j$; this fragment is an occurrence
of $X[i]\cdots X[j]$ as a substring of $X$. 
The fragments $X[i\dd j)$, $X(i\dd j]$, and $X(i\dd j)$ are defined similarly.

A morphism is a function $f: \Sigma_1^* \to \Sigma_2^*$ such that $f(X)=\bigcirc_{i=1}^{|X|}f(X[i])$,
where $\bigcirc$ denotes the concatenation operator.
Note that every function mapping $\Sigma_1$ to $\Sigma_2^*$ can be uniquely extended to a morphism.

\subsection{Straight-Line Programs}

A straight-line program is a tuple $\G=(\S, \Sigma, \rhs, S)$,
where $\S$ is a finite sequence of \emph{symbols}, $\Sigma \sub \S$ is a set of \emph{terminal symbols},
$\rhs : (\S\sm \Sigma) \to \S^*$ is the \emph{production} (or \emph{right-hand side}) function, and $S\in\S$ 
is the start symbol, and the symbols in $\S$ are ordered so that $B$ precedes $A$ if $B$ occurs in $\rhs(A)$. 
We also write $A\to B_1\cdots B_k$ instead of $\rhs(A)=B_1\cdots B_k$.

The set $\S \sm \Sigma$ of \emph{non-terminals} is denoted by $\N$. 
The size of $\G$ is $|\G|:=|\S|+\sum_{A\in \N}|\rhs(A)|$: the number of symbols plus the total length of productions.
The \emph{expansion} function $\exp:\S \to \Sigma^+$ is defined recursively:
\[\exp(A) = \begin{cases}
    A & \text{if }A\in \Sigma,\\
    \bigcirc_{i=1}^k \exp(B_i) & \text{if }A\to \bigcirc_{i=1}^k B_i.
\end{cases}
\]
We say that $\G$ is a \emph{grammar-compressed representation} of $\exp(S)$. 
The $\exp$ function naturally extends to a morphism $\exp:\S^* \to \Sigma^*$
with $\exp(\bigcirc_{i=1}^m A_i) = \bigcirc_{i=1}^m \exp(A_i)$.

For a symbol $A\in \S$, we denote $|A|=|\exp(A)|$. 
In this work, we assume a word RAM machine with machine words of $\Omega(\log |S|)$ bits.
In this setting, one can compute $|A|$ for all $A\in \S$ in $\Oh(|\G|)$ time.
Consequently, we assume that $|A|$ is stored along with $A$
in the straight-line programs given to our algorithms.

A straight-line program $\G$ is in \emph{Chomsky normal form} if $|\rhs(A)|=2$ for all $A\in \N$. Given an arbitrary straight-line program $\G$, an equivalent straight-line program $\G'$ in Chomsky normal form can be constructed in $\Oh(|\G|)$ time; moreover, $|\G|=\Oh(|\G'|)$.
Thus, without loss of generality, we assume that all straight-line programs given to our algorithms
are already in the Chomsky normal form.

\section{FPTAS for Compressed Edit Distance of Two Strings}

The edit distance $\ED(X,Y)$ of two strings $X,Y\in \Sigma^*$ is defined as the minimum number of character insertions, deletions, and substitutions needed to transform $X$ into $Y$.

In this section, we prove the following result.

\begin{theorem}\label{thm:edap}
    Given a straight-line program $\G_X$ of size $n$ generating a string $X$ of length $N>0$,
    a straight-line program $\G_Y$ of size $m$ generating a string $Y$ of length $M>0$,
    and a parameter $\eps\in (0,1]$,
    an integer between $\ED(X,Y)$ and $(1+\eps)\ED(X,Y)$ can be computed in $\Ohtilde\left((nm(N+M))^{2/3}\eps^{-1/3}\right)$ time.
\end{theorem}

Let $\$\notin \Sigma$ and let $\cdot^\$ : \Sigma^* \to (\Sigma \cup \{\$\})^*$ be a morphism defined with $a^\$ = a\$$
for $a\in \Sigma$.
Then, $\ED(X,Y)=\frac12 \ID(X^\$, Y^\$)$~\cite{Tiskin15}.
Moreover, if $X$ is represented by a straight line program $\G$,
then $X^\$$ can be represented using a straight-line program of size $2|\G|+1$.
This reduction allows computing $\ID$ instead of $\ED$.

\begin{definition}[Alignment graph]
For two strings $X$ and $Y$, the alignment graph $\G_{X,Y}$ is a weighted undirected graph
with vertex set $\{v_{x,y}: x\in [0\dd |X|], y\in [0\dd |Y|]\}$ and edges:
\begin{itemize}
    \item $v_{x,y-1} \leftrightarrow v_{x,y}$ of length 1, for $x\in [0\dd |X|]$ and $y\in [1\dd |Y|]$;
    \item $v_{x-1,y} \leftrightarrow v_{x,y}$ of length 1, for $x \in [1\dd |X|]$ and $y\in [0\dd |Y|]$;
    \item $v_{x-1,y-1} \leftrightarrow v_{x,y}$ of length 0, for $x\in [1\dd |X|]$
    and $y\in [1\dd |Y|]$ such that $X[x]=Y[y]$.
\end{itemize}
\end{definition}
\begin{observation}\label{obs:dir}
Let $d$ be the metric induced by $G_{X,Y}$. All $x,x'\in [0\dd |X|]$ and $y,y'\in [0\dd |Y|]$,
satisfy
    \[d(v_{x,y},v_{x',y'}) = 
    \begin{cases}
        \ID(X(x\dd x'],Y(y\dd y']) &\text{if }x\le x'\text{ and }y\le y',\\
        \ID(X(x'\dd x],Y(y'\dd y]) &\text{if }x'\le x\text{ and }y'\le y,\\
        |x-x'|+|y-y'| & \text{otherwise.}
    \end{cases}
    \]
\end{observation}

For two ranges $[x\dd x']\sub [0\dd |X|]$ and $[y\dd y']\sub [0\dd |Y|]$,
the subgraph of $G_{X,Y}$ induced by $\{v_{\bar{x},\bar{y}} : \bar{x}\in [x\dd x'],\, \allowbreak \bar{y}\in [y\dd y']\}$
is denoted $G_{X,Y}^{[x\dd x'], [y\dd y']}$ and called a \emph{block} in $G_{X,Y}$.
For a block $B$, we distinguish the \emph{input vertices} $\inv^B = (v_{x',y},v_{x'-1,y},\ldots,v_{x,y},v_{x,y+1},\ldots,v_{x,y'})$ and the \emph{output vertices} $\outv^B = (v_{x',y},v_{x',y+1},\ldots,v_{x',y'},\allowbreak v_{x'-1,y'},\ldots,v_{x,y'})$; both sequences consist of $|B|:=(x'-x)+(y'-y)+1$ vertices.
The $\DIST_B$ table is a $|B|\times |B|$ matrix with entries $\DIST_B[i,j]=d(\inv^B_i,\outv^B_j)$
for $i,j\in [1\dd |B|]$.
The $\DIST_{B}$ table satisfies the Monge property~\cite{Tiskin15}:
$\DIST_{B}[i,j] + \DIST_{B}[i',j']\le \DIST_{B}[i,j']+\DIST_{B}[i,j']$ holds
for all $i,i',j,j'\in[1\dd |B|]$ such that $i\le i'$ and $j\le j'$.
For two strings $X,Y\in \Sigma^*$, we also define $\DIST_{X,Y}$ to be $\DIST_B$ for $B=G_{X,Y}^{[0\dd |X|],[0\dd |Y|]}$.
By Observation~\ref{obs:dir}, if $B=G_{X,Y}^{[x\dd x'],[y\dd y']}$,
then $\DIST_B = \DIST_{X(x\dd x'],Y(y\dd y']}$.

\paragraph{Box decomposition}
For two strings $X,Y\in \Sigma^*$, the \emph{box decomposition} $\B$ of the graph $G_{X,Y}$ is defined
based on decompositions $X = X_1\circ\cdots \circ X_{p_X}$ and $Y=Y_1\circ \cdots \circ Y_{p_Y}$ into 
non-empty fragments, called \emph{phrases}.

Let us define sets $\{b^X_0,\ldots,b^X_{p_X}\}$ and $\{b^Y_0,\ldots,b^Y_{p_Y}\}$ 
of \emph{phrase boundaries} in $X$ and $Y$, respectively, so that the phrases are $X_i = X(b^X_{i-1}\dd b^X_i]$ for $i\in [1\dd p_X]$ and $Y_j = Y(b^Y_{j-1}\dd b^Y_j]$ for $j\in [1\dd p_Y]$. A vertex $v_{x,y}$ is a \emph{boundary vertex} if  $x$ is a phrase boundary in $X$ or $y$ is a phrase boundary in $Y$, and a \emph{grid vertex} if both $x$ is a phrase boundary in $X$ and $y$ is a phrase boundary in $Y$.
The box decomposition $\B$ is an indexed family $(B_{i,j})_{i\in [1\dd p_X],j\in [1\dd p_Y]}$
of \emph{boxes} $B_{i,j}:=G_{X,Y}^{[b_{i-1}^X\dd b_{i}^X],[b_{j-1}^Y\dd b_{j}^Y]}$.

\subsection{Portal-Respecting Walks}
Hermelin et al.~\cite{HermelinLLW13} applied a box decomposition obtained via an analogue of \cref{cor:scheme}
to determine $\ID(X,Y)$ using a dynamic-programming procedure computing $\ID(X[1\dd x],Y[1\dd y])$ for all boundary vertices $v_{x,y}$.
We reduce the number of DP states by considering only a selection $\P$ of boundary vertices, called \emph{portals}.
This allows improving the running time from $\Ohtilde(\frac{NM}{\tau})$ to $\Ohtilde(|\P|)$,
but reduces the search space from the family of all walks $v_{0,0}\leadsto v_{x,y}$
to walks that cross box boundaries only at portals.
Below, we formally define such portal-respecting walks and provide a construction suitable for approximating $\ID(X,Y)$.
\begin{definition}\label{def:portalwalk}
    Let $\B$ be a box decomposition of $G_{X,Y}$ and let $\P$ be a set \emph{portals} (selected boundary vertices).
    We say that a walk $W$ is a \emph{portal-respecting $(i,j)$-walk} if $W$ is a concatenation of walks $W'$ and $W''$  such that: 
    \begin{itemize}
        \item $W''$ starts at an input portal of $B_{i,j}$ and is entirely contained within $B_{i,j}$, and 
        \item $W'$ is the empty walk at $v_{0,0}$, a portal-respecting $(i-1,j)$-walk, or a portal-respecting $(i,j-1)$-walk.
    \end{itemize}
\end{definition}

Let us fix a box decomposition $\B$ of $G_{X,Y}$, and a set of portals $\P$.
For a box $B_{i,j}\in \B$, let $\P_{i,j} = \P \cap \outv^{B_{i,j}}$ denote the output portals of $B_{i,j}$.
Moreover, for a vertex $v_{x,y}\in B_{i,j}$, we denote $d_{x,y}=d(v_{0,0},v_{x,y})=\ID(X[1\dd x],Y[1\dd y])$
and let $D^{i,j}_{x,y}$ be the minimum length of a portal-preserving $(i,j)$-walk ending at $v_{x,y}$.

\begin{lemma}\label{cor:boxdp}
Given a set of portals $\P$ for a box decomposition $\B$ of $G_{X,Y}$,
the
the length of the shortest portal-respecting $(p_X,p_Y)$-walk ending at $v_{|X|,|Y|}$
can be computed in $\Ohtilde(|\P|)$ time provided $\Ohtilde(1)$-time random access to the $\DIST$ matrices of all the boxes of $\B$.
\end{lemma}
\begin{proof}
For each box $B_{i,j}\in \B$, our algorithm computes $D^{i,j}_{x,y}$ for all vertices $v_{x,y}\in \P_{i,j}$. For this, the boxes $B_{i,j}$ containing any output portal are processed in the order of non-decreasing
values $i+j$. 

If $(i,j)=(1,1)$, then \cref{def:portalwalk} and Observation~\ref{obs:dir} yield $D^{1,1}_{x,y}=d(v_{0,0},v_{x,y})$,
    and this value can be retrieved from the $\DIST_{B_{1,1}}$ matrix in $\Ohtilde(1)$ time. 
    Thus, we henceforth assume $(i,j)\ne (1,1)$.
    
    Consider a portal-respecting $(i,j)$-walk $W$ to a vertex $v_{x,y}\in \outv^{B_{i,j}}$.
    By \cref{def:portalwalk}, $W$ is a concatenation of two walks $W'$ and $W''$
    such that $W''$ starts at a vertex $v_{x',y'}\in \P\cap \inv^{B_{i,j}}$ and is entirely contained within $B_{i,j}$,
    whereas $W'$ is a portal-respecting $(i,j-1)$-walk to $v_{x',y'}$ or a portal respecting $(i-1,j)$-walk to $v_{x',y'}$. Observe that, for a fixed portal $v_{x',y'}\in  \P\cap \inv^{B_{i,j}}$, the lengths of $W'$ and $W''$ can be optimized independently.
    Consequently, by Observation~\ref{obs:dir}, \[D^{i,j}_{x,y} = \max\left(\max_{v_{x',y'}\in \P_{i-1,j} \cap \inv^{B_{i,j}}}\left\{ D^{i-1,j}_{x',y'}+ d(v_{x',y'},v_{x,y})\right\},\max_{v_{x',y'}\in \P_{i,j-1} \cap \inv^{B_{i,j}}}\left\{ D^{i,j-1}_{x',y'}+ d(v_{x',y'},v_{x,y})\right\}\right).\]
    
    A matrix (indexed by $v_{x,y}\in \P_{i,j}$ and all vertices $\P_{i-1,j} \cap \inv^{B_{i,j}}$) containing the values $D^{i-1,j}_{x',y'}+d(v_{x',y'},v_{x,y})$ can be obtained from a submatrix of the $\DIST_{B_{i,j}}$ matrix
    by adding $D^{i-1,j}_{x',y'}$ to all entries in the column of $v_{x',y'}$.
    These modifications preserve the Monge property, so the resulting matrix is a Monge matrix with $\Ohtilde(1)$-time random access.
    Consequently, the SMAWK algorithm~\cite{AggarwalKMSW87} allows computing row-minima,
    i.e., the values $\max_{v_{x',y'}\in\P_{i-1,j} \cap \inv^{B_{i,j}}}\big\{ D^{i-1,j}_{x',y'}+ d(v_{x',y'},v_{x,y})\big\}$.
    A symmetric procedure allows computing the values $\max_{v_{x',y'}\in \P_{i,j-1} \cap \inv^{B_{i,j}}}\big\{ D^{i,j-1}_{x',y'}+ d(v_{x',y'},v_{x,y})\big\}$,
    which lets us derive the costs $D^{i,j}_{x,y}$ for all the vertices $v_{x,y}\in \P_{i,j}$.
    The SMAWK algorithm takes nearly linear time with respect to the sum of matrix dimensions, 
    so the overall time complexity is $\Ohtilde(|\P\cap B_{i,j}|)$.

Each vertex belongs to at most four boxes, so the overall running time is $\Ohtilde(|\P|)$.
\end{proof}

\begin{lemma}\label{lem:apportals}
    Let $\B$ be a box decomposition of the graph $G_{X,Y}$ for $X,Y\in \Sigma^+$
    and let $\alpha > 0$ be a real number.
    Suppose that $\P$ consists of all the grid vertices and all the boundary vertices
    $v_{x,y}$ of $\B$ satisfying $|x-y| = \lfloor(1+\alpha)^r\rfloor$ for some integer $r$.
    Then, every vertex $v_{x,y}\in B_{i,j}$ satisfies $D^{i,j}_{x,y}\le (1+2\alpha)^{i+j}d_{x,y}$.
\end{lemma}
\begin{proof}
    We proceed by induction on $i+j$. The base case is trivially satisfied due to $D_{x,y}^{1,1} = d_{x,y}$ for $v_{x,y}\in B_{1,1}$. We henceforth fix $v_{x,y}\in B_{i,j}$ with $(i,j)\ne (1,1)$. By Observation~\ref{obs:dir},
    there is a shortest path from $v_{0,0}$ to $v_{x,y}$ contained within $G_{X[1\dd x],Y[1\dd y]}$. Let $v_{x',y'}$ be the first vertex of $B_{i,j}$ on this path. Observe that $v_{x',y'}\in \inv^{B_{i,j}}$
    and $d_{x,y}=d_{x',y'}+d(v_{x',y'},v_{x,y})$. 
    By symmetry, we may assume without loss of generality that $v_{x',y'}\in \outv^{B_{i-1,j}}$.

    Let us choose $v_{x',y''}\in \P_{i-1,j}\cap \outv^{B_{i-1,j}}$ as close as possible to $v_{x',y'}$.
    Since grid vertices are portals, such $v_{x',y''}$ exists.
    Moreover, by the choice of the remaining portals, $d(v_{x',y'},v_{x',y''})\le \alpha|x'-y'| \le \alpha d_{x',y'}$. 
    Let us construct a portal-respecting $(i,j)$-walk to $v_{x,y}$ by concatenating a shortest portal-respecting $(i-1,j)$-walk to $v_{x',y''}$ and a shortest path from $v_{x',y''}$ to $v_{x,y}$ (by Observation~\ref{obs:dir}, we may assume that this path is contained in $B_{i,j}$). This proves $D^{i,j}_{x,y} \le D^{i-1,j}_{x,y} + d(v_{x',y''},v_{x,y})$. 
    The inductive assumption further yields $D^{i-1,j}_{x,y}\le (1+2\alpha)^{i+j-3}d_{x',y''}$,
    and thus $D^{i,j}_{x,y}  \le (1+2\alpha)^{i+j-3}d_{x',y''} +  d(v_{x',y''},v_{x,y}) 
    \le (1+2\alpha)^{i+j-3}(d_{x',y'}+d(v_{x',y'},v_{x',y''}))+d(v_{x',y'},v_{x',y''})+d_{x,y}-d_{x',y'}
    \le (1+2\alpha)^{i+j-3}(d_{x',y'}+\alpha d_{x',y'})+\alpha d_{x',y'}+d_{x,y}-d_{x',y'} \le (1+2\alpha)^{i+j-2} d_{x,y}$.
\end{proof}

\subsection{A Grammar-Based Box Decomposition}
Hermelin et al.~\cite{HermelinLLW13} presented an algorithm that,
given two grammar-compressed strings $X,Y\in \Sigma^+$ and an integer parameter $\tau$,
constructs a box decomposition $\B$ of $G_{X,Y}$ with $p_X = \Oh(\lceil\frac1\tau|X|\rceil)$
and $p_Y = \Oh(\lceil\frac1\tau |Y|\rceil)$, along with an oracle providing random
access to the $\DIST_{B_{i,j}}$ matrices of all the boxes $B_{i,j}$.
However, their construction costs $\Omega(|X|+|Y|)$ time, which is prohibitive in most of our applications.
In this section, we achieve the same goal avoiding the linear dependency on the lengths of $X$ and $Y$.
The bottleneck of~\cite{HermelinLLW13} is constructing appropriate decompositions of $X$ and $Y$ into phrases.
In the following lemma,  we implement an analogous step more efficiently by 
building a grammar-compressed representation of the sequence of phrases,
with each phrase represented by a symbol in an auxiliary grammar.

\begin{restatable}{lemma}{lempart}\label{lem:part}
    Given a straight-line program $\G$ generating a string $X$
    and an integer $\tau \ge 1$, in $\Oh(|\G|)$ time one can construct straight-line programs $\G^+$
    and $\G^P$ of size $\Oh(|\G|)$ such that:
    \begin{itemize}
        \item the terminal symbols of $\G^P$ are the symbols $A$ of $\G^+$ satisfying $|A|\le \tau$,
        \item $\G^P$ generates a string $P$ such that $\exp_{\G^+}(P)=X$ and $|P| \le \left\lceil\frac{3}{\tau}|X|\right\rceil$.
    \end{itemize}
\end{restatable}

\begin{proof}
If $\tau \ge |X|$, then we simply set $\G^+=\G$ and set $\G^P$ to be a grammar with no non-terminals
whose starting symbol is the starting symbol of $\G$; this construction clearly satisfies the required conditions.

We henceforth assume that $\tau < |X|$. The grammar $\G^+$ is constructed by adding new non-terminals to $\G$.
As for $\G^P$, we include as terminals all symbols $A$ of $\G^+$ with $|A|\le \tau$,
and we add further symbols as non-terminals.
For every symbol $A$ of $\G$ with $|A|>\tau$, we introduce three new non-terminals:
\begin{itemize}
    \item $L(A)$ and $R(A)$ to $\G^+$, satisfying $|L(A)|\le \tau$ and $|R(A)|\le \tau$,
    \item $M(A)$ to $\G^P$.
\end{itemize}
The productions for $L(A)$, $R(A)$, and $M(A)$ are determined based on the production $A\to B_LB_R$:
\begin{enumerate}
\item If $|B_L|\le \tau$ and $|B_R|\le \tau$, then  $M(A)\to \emptystring$, $L(A)\to B_L$, and $R(A)\to B_R$.
\item If $|B_L| > \tau$ and $|B_R|> \tau$, then $M(A)\to M(B_L)R(B_L)L(B_R)M(B_R)$, $L(A)\to L(B_L)$, and $R(A)\to R(B_R)$.
\item If $|B_L| > \tau$ and $|B_R| \le \tau$, then $L(A)\to L(B_L)$ and:
\begin{enumerate}
    \item $R(A)\to R(B_L)B_R$ and $M(A)\to M(B_L)$ if $|R(B_L)|+|B_R|\le \tau$,
    \item $R(A)\to B_R$ and $M(A)\to M(B_L)R(B_L)$ otherwise.
\end{enumerate}
\item If $|B_L| \le \tau$ and $|B_R| > \tau$, then $R(A)\to R(B_R)$ and:
\begin{enumerate}
    \item $L(A)\to B_LL(B_R)$ and $M(A)\to M(B_R)$ if $|B_L|+|L(B_R)|\le \tau$, 
    \item $L(A)\to B_L$ and $M(A)\to L(B_R)M(B_R)$ otherwise.
\end{enumerate}
\end{enumerate}
Additionally, for the starting symbol $S$ of $\G$, we add a starting symbol $S^P\to L(S)M(S)R(S)$ to $\G^P$.

A simple inductive argument shows that every symbol $A$ of $\G$ with $|A|>\tau$ satisfies
\[\exp_{\G}(A)=\exp_{\G^+}(L(A)\circ \exp_{\G^P}(M(A))\circ R(A)).\]
In particular, $P = \exp_{\G^P}(S^P)$ satisfies $\exp_{\G^+}(P)=\exp_{\G}(S)=X$.

It remains to prove that $|P|< \frac{3}{\tau}|X|$.
For this, we inductively show that every symbol $A$ of $\G$ with $|A|>\tau$ satisfies
$|L(A)|+|R(A)|+\tau(|M(A)|+2) < 3|A|$. To prove this claim, we analyze the cases based on the production $A\to B_LB_R$.
\begin{enumerate}
    \item If $|B_L|\le \tau$ and $|B_R|\le \tau$, then 
    \[|L(A)|+|R(A)|+\tau(|M(A)|+2)  = |A|+2\tau < 3|A|.\]
    \item If $|B_L|>\tau$ and $|B_R|>\tau$, then 
    \begin{multline*}|L(A)|+|R(A)|+\tau(|M(A)|+2) = |L(B_L)|+|R(B_R)|+\tau(|M(B_L)|+2+|M(B_R)|+2) <\\
    |L(B_L)|+|R(B_L)|+\tau(|M(B_L)|+2) + |L(B_R)|+|R(B_R)|+\tau(|M(B_R)|+2) < 3|B_L|+3|B_R| = 3|A|.\end{multline*}
    \item If $|B_L|>\tau$, $|B_R|\le \tau$, then
    \begin{itemize}
        \item If $|R(B_L)|+|B_R|\le \tau$, then
         \begin{multline*}
            |L(A)|+|R(A)|+\tau(|M(A)|+2) = |L(B_L)| + |R(B_L)|+|B_R| + \tau(|M(B_L)|+2) <\\ 3|B_L|+|B_R| < 3|A|.
         \end{multline*}
         \item Otherwise, 
         \begin{multline*}|L(A)|+|R(A)|+\tau(|M(A)|+2) = |L(B_L)| + |B_R| + \tau(|M(B_L)|+3)  < \\
            |L(B_L)| + |B_R| + \tau(|M(B_L)|+2)+|R(B_L)|+|B_R| < 3|B_L|+2|B_R| < 3|A|.\end{multline*}
    \end{itemize}
    \item The case involving $|B_L|\le \tau$ and  $|B_R|> \tau$ is symmetric to the previous one.
\end{enumerate}
In particular, this claim holds for $A=S$, so $|S^P| = |M(S)|+2 < \frac{1}{\tau}(3|S|-|L(S)|-|R(S)|) < \frac{3}{\tau}|S| = \frac{3}{\tau}|X|$.
\end{proof}
As for constructing the $\DIST$ matrices, we use the original implementation from~\cite{HermelinLLW13}.
\begin{lemma}[{\cite[Section 5]{HermelinLLW13}}]\label{lem:dist}
    Given straight-line programs $\G_X$ and $\G_Y$ and an integer $\tau \ge 1$,
    in $\Ohtilde(|\G_X||\G_Y|\tau)$ time
    one can construct a data structure that provides $\Ohtilde(1)$-time random access
    to the $\DIST_{\exp(A_X),\exp(A_Y)}$ matrices
    for all symbols $A_X$ of $\G_X$ and $A_Y$ of $\G_Y$ 
    satisfying $|A_X|\le \tau$ and $|A_Y|\le \tau$.
\end{lemma}

Combining \cref{lem:part,lem:dist}, we complete our construction.
\begin{corollary}\label{cor:scheme}
Given a straight-line program $\G_X$ of size $n$ generating a string $X$ of length $N>0$,
a straight-line program $\G_Y$ of size $m$ generating a string $Y$ of length $M>0$,
and an integer $\tau\in [1\dd N+M]$, one can in $\Ohtilde(\frac{N+M}{\tau}+nm\tau)$ time
construct a box decomposition $\B = (B_{i,j})_{i\in [1\dd p_X],j\in [1\dd p_Y]}$
of $G_{X,Y}$ with $p_X = \Oh(\lceil\frac{N}{\tau}\rceil)$ and $p_Y = \Oh(\lceil\frac{M}{\tau}\rceil)$, along with an oracle providing $\Ohtilde(1)$-time random access to the $\DIST_{B_{i,j}}$ matrices.
\end{corollary}
\begin{proof}
    First, we use \cref{lem:part} to obtain grammars $\G^+_X$ and $\G^P_X$.
    The string $P_X$ represented by $\G^P_X$ satisfies $X = \exp_{\G^+_X}(P_X)$,
    so it can be interpreted as a decomposition of $X$ into $p_X := |P_X|$ phrases,
    with the $i$th phrase $X_i$ being an occurrence of $\exp_{\G^+_X}(P_X[i])$.
    The decomposition of $Y$ is obtained in the same way
    based on grammars $\G^+_Y$ and $\G^P_Y$ constructed for $Y$.

    The box decomposition $\B$ is based on these decompositions of $X$ and $Y$.
    Note that each box $B_{i,j}$ satisfies \[\DIST_{B_{i,j}}=\DIST_{\exp_{\G^+_X}(P_X[i]),\exp_{\G^+_Y}(P_Y[j])}.\]
    Due to $|P_X[i]|\le \tau$ and $|P_Y[j]|\le \tau$, \cref{lem:dist} applied to
    $\G^+_X$ and $\G^+_Y$ provides $\Ohtilde(1)$-time oracle access to all these matrices.
    Storing $P_X$ and $P_Y$, we can point to $\DIST_{B_{i,j}}$ in $\Oh(1)$ time given $i,j$.
\end{proof}

\subsection{Algorithm}

\begin{proposition}\label{prp:edap}
   Given a straight-line program $\G_X$ of size $n$ generating a string $X$ of length $N>0$,
   a straight-line program $\G_Y$ of size $m$ generating a string $Y$ of length $M>0$,
   and a parameter $\eps\in (0,1]$,
   a $(1+\eps)$-approximation of $\ID(X,Y)$ can be computed in $\Ohtilde\left((nm(N+M))^{2/3}\eps^{-1/3}\right)$ time.
\end{proposition}
\begin{proof}
    The algorithm uses \cref{cor:scheme,cor:boxdp}
    with the set of portals $\P$ defined as in \cref{lem:apportals},
    where $\alpha=\Omega(\frac{\eps}{p_X+p_Y-2})=\Omega(\frac{\eps \tau}{N+M})$ is chosen so that $(1+2\alpha)^{p_X+p_Y-2}=1+\eps$.
    \cref{lem:apportals} guarantees that the resulting value is a $(1+\eps)$-approximation of $\ID(X,Y)$.
    The number of portals is $
        \Oh\left(\tfrac{NM}{\tau^2} + \tfrac{N}{\tau}\log_{1+\alpha}M + \tfrac{M}{\tau}\log_{1+\alpha}N\right)
        = \Ohtilde\big(\tfrac{(N+M)^2}{\eps \tau^2}\big) $,
    so the overall running time is $\Ohtilde\big(nm\tau + \tfrac{(N+M)^2}{\eps \tau^2}\big)$.
    Optimizing $\tau\in [1\dd N+M]$, we get $\Ohtilde(nm+\eps^{-1}+(nm(N+M))^{2/3}\eps^{-1/3})$ time.
    If the first term dominates, then $nm \ge (nm(N+M))^{2/3}\eps^{-1/3}  \ge (N+M)^2 \eps^{-1}$.
    However, $\Oh(NM) = \Oh((N+M)^2 \eps^{-1})$ time is enough to compute $\ID(X,Y)$ exactly without compression.
    If the second term dominates, then $\eps^{-1} \ge (nm(N+M))^{2/3}\eps^{-1/3}  \ge nm(N+M)$.
    However, $\Ohtilde(\sqrt{nm}(N+M))=\Ohtilde(nm(N+M))$ time is enough to compute $\ID(X,Y)$ exactly using \cref{prp:edexact} with $D=N+M$.
\end{proof}
Theorem~\ref{thm:edap} follows through the reduction from $\ED$ to $\ID$.

\subsection{Exact Output-Sensitive Algorithm}

In this section we prove Theorem~\ref{thm:edexact}:

\begin{theorem}\label{thm:edexact}
    Given a straight-line program $\G_X$ of size $n$ generating a string $X$ of length $N>0$ and
    a straight-line program $\G_Y$ of size $m$ generating a string $Y$ of length $M> 0$,
    the edit distance $\ED(X,Y)$ can be computed in $\Ohtilde\left(\sqrt{(1+\ED(X,Y))nm(N+M)}\right)$ time.
\end{theorem}
The algorithm behind Theorem~\ref{thm:edexact} reduces the problem to a decision version (asking whether $\ED(X,Y)\le D$
for a threshold $D$) and then uses the same scheme with all boundary vertices $(x,y)$ satisfying $|x-y|\le D$ selected as portals.

\begin{lemma}\label{lem:distportals}
Let $\B$ be a box decomposition of the graph $G_{X,Y}$ for $X,Y\in \Sigma^+$
and let $D\ge 0$ be an integer.
Suppose that $\P$ consists of all the boundary vertices $v_{x,y}$ of $\B$ satisfying $|x-y|\le D$.
Then, every vertex $v_{x,y}\in B_{i,j}$ with $d_{x,y}\le D$
satisfies $D^{i,j}_{x,y}=d_{x,y}$.
\end{lemma}
\begin{proof}
We proceed by induction on $i+j$. The base case is trivially satisfied due to $D^{1,1}_{x,y}=d_{x,y}$ for $v_{x,y}\in B_{1,1}$.
We henceforth fix $v_{x,y}\in B_{i,j}$ with $(i,j)\ne (1,1)$ and $d_{x,y}\le D$.
By Observation~\ref{obs:dir}, there is a shortest path from $v_{0,0}$ to $v_{x,y}$ contained within $G_{X[1\dd x],Y[1\dd y]}$.
Let $v_{x',y'}$ be the first vertex on this path that belongs to $B_{i,j}$.
Observe that $v_{x',y'}\in \inv^{B_{i,j}}$ and $d_{x,y}=d_{x',y'}+d(v_{x',y'},v_{x,y})$. 
Consequently, $|x'-y'|\le d_{x',y'}\le d_{x,y}\le D$,
so $v_{x',y'}\in \P_{i-1,j}\cup \P_{i,j-1}$.
By symmetry, we may assume without loss of generality that $v_{x',y'}\in \P_{i-1,j}$.

Let us construct a portal-respecting $(i,j)$-walk to $v_{x,y}$ by concatenating a shortest portal-respecting $(i-1,j)$-walk to $v_{x',y'}$ and a shortest path from $v_{x',y'}$ to $v_{x,y}$ (by Observation~\ref{obs:dir} applied to $G_{X(x'\dd x],Y(y'\dd y]}$, we may assume that this path is contained in $B_{i,j}$).
This proves $D^{i,j}_{x,y} \le D^{i-1,j}_{x',y'}+d(v_{x',y'},v_{x,y})=D^{i,j}_{x,y} \le D^{i-1,j}_{x',y'}+d_{x,y}-d_{x',y'}$. The inductive assumption yields $D^{i-1,j}_{x',y'}=d_{x',y'}$, and thus $D^{i,j}_{x,y}\le d_{x,y}$ holds as claimed.
\end{proof}

\begin{proposition}\label{prp:edexact}
    Given a straight-line program $\G_X$ of size $n$ generating a string $X$ of length $N>0$,
    a straight-line program $\G_Y$ of size $m$ generating a string $Y$ of length $M>0$,
    and an integer $D\in [1\dd N+M]$, one can in $\Ohtilde\left(\sqrt{nmD(N+M)}\right)$ time compute
    $\ID(X,Y)$ or certify that $\ID(X,Y)>D$.
\end{proposition}
\begin{proof}
    The algorithm uses \cref{cor:scheme,cor:boxdp}
    with the set of portals $\P$ defined as in \cref{lem:distportals}.
    The latter lemma guarantees that the resulting value is $\ID(X,Y)$ provided that $\ID(X,Y)\le D$.
    Otherwise, the resulting value exceeds $D$, certifying that $\ID(X,Y)>D$.

    The number of portals is $\Oh(D\cdot \frac{N+M}{\tau})$, 
    so the overall running time is $\Ohtilde(nm\tau + D\cdot \frac{N+M}{\tau})$.
    Optimizing $\tau\in [1\dd N+M]$, we get $\Ohtilde(nm+D+\sqrt{nmD(N+M)})$ time. 
    Since $D \le N+M$, the second term is dominated by the third one.
    If the first term dominates, then $nm > D(N+M)$,
    and thus $\sqrt{nmD(N+M)} \ge D(N+M)$.
    However, $\Ohtilde(N+M+D^2) =\Ohtilde(D(N+M))$ time suffices solve the problem for uncompressed
    strings~\cite{LV88}.
\end{proof}
Theorem~\ref{thm:edexact} follows through exponential search and the reduction from $\ED$ to $\ID$. 

\subsection{LCS Approximation}

In this section we prove Theorem~\ref{thm:lcs}:

\begin{restatable}{theorem}{lcsthm}\label{thm:lcs}
Given a straight-line program $\G_X$ of size $n$ generating a string $X$ of length $N>0$,
   a straight-line program $\G_Y$ of size $m$ generating a string $Y$ of length $M>0$,
   and a parameter $\eps\in (0,1]$,
   a $(1+\epsilon)$-approximation of $\LCS(X,Y)$ can be computed in $\Ohtilde\left((nm(N+M))^{2/3}\eps^{-1/3}\right)$ time.
\end{restatable}

The algorithm behind Theorem~\ref{thm:lcs} is essentially the same as that of Theorem~\ref{thm:edap}, and this is why the running times coincide.
The main difference is that the output portals of the box $B_{i,j}$ are chosen \emph{adaptively} while the dynamic-programming algorithm processes $B_{i,j}$.

As for LCS approximation, our choice of portals is \emph{adaptive}.
For a box $B_{i,j}\in \B$, let $\P_{i,j} = \P \cap \outv^{B_{i,j}}$.
Observe that the value $D^{i,j}_{x,y}$ for $v_{x,y}\in B_{i,j}$ depends only on $\P_{i',j'}$ with $i'+j'< i+j$. 
Hence, except for the grid vertices (all included in $\P$),
we may select the portals $\P_{i,j}$ based on the values $D^{i,j}_{x,y}$ for $v_{x,y}\in \outv^{B_{i,j}}$.

For $v_{x,y}\in B_{i.j}$, let $\ell_{x,y} = \frac12(|X|+|Y|-d_{x,y})=\LCS(X[1\dd x],Y[1\dd y])$ and $L^{i,j}_{x,y}=\frac12(|X|+|Y|-D^{i,j}_{x,y})$.

\begin{lemma}\label{lem:lcsportals}
    Let $\B$ be a box decomposition of the graph $G_{X,Y}$ for $X,Y\in \Sigma^+$
    and let $\alpha > 0$ be a real number.
    Suppose that $\P$ consists of all the grid vertices and, for each box $B_{i,j}\in \B$,
    all vertices $v_{x,y}\in \outv^{B_{i,j}}$ such that:
    \begin{itemize}
        \item $v_{x-1,y} \in \outv^{B_{i,j}}$ and $\floor{\log_{1+\alpha} L^{i,j}_{x,y}} > \floor{\log_{1+\alpha} L^{i,j}_{x-1,y}}$, or
        \item $v_{x,y-1} \in \outv^{B_{i,j}}$ and $\floor{\log_{1+\alpha} L^{i,j}_{x,y}} > \floor{\log_{1+\alpha} L^{i,j}_{x,y-1}}$.
    \end{itemize}
    Then, for each vertex $v_{x,y}\in B_{i,j}$, we have $L^{i,j}_{x,y} \ge (1+\alpha)^{2-i-j}\ell_{x,y}$.
\end{lemma}
\begin{proof}
    We proceed by induction on $i+j$.
    The base case is trivially satisfied due to $L^{1,1}_{x,y}=\ell_{x,y}$ for $v_{x,y}\in B_{1,1}$.
    Thus, we henceforth fix a vertex $v_{x,y}\in B_{i,j}$ with $(i,j)\ne (1,1)$.
    By Observation~\ref{obs:dir}, there is a shortest path from $v_{0,0}$ to $v_{x,y}$ contained within $G_{X[1\dd x],Y[1\dd y]}$. Let $v_{x',y'}$ be the first vertex on this path that belong to $B_{i,j}$. Observe that $v_{x',y'}\in \inv^{B_{i,j}}$ and $\ell_{x,y}=\ell_{x',y'}+\LCS(X(x'\dd x],Y(y'\dd y])$. By symmetry, we may assume without loss of generality that $v_{x',y'}\in \outv^{B_{i-1,j}}$.
    
    Consider the largest value $y''\in [1\dd y']$ such that $v_{x',y''}\in \P_{i-1,j}$.
    Since grid vertices are portals, such $v_{x',y''}$ exits. 
    Moreover, by the choice of the remaining portals, $L^{i-1,j}_{x',y'}\le (1+\alpha) L^{i-1,j}_{x',y''}$.
    Let us construct a portal-respecting $(i,j)$ walk to $v_{x,y}$ by concatenating a shortest portal-respecting $(i-1,j)$-walk to $v_{x',y''}$ and a shortest path from $v_{x',y''}$ to $v_{x,y}$ (by Observation~\ref{obs:dir}, we may assume that this path is contained $B_{i,j}$).
    This proves that $D^{i,j}_{x,y} \le D^{i-1,j}_{x',y''} + d(v_{x',y''},v_{x,y}) \le D^{i-1,j}_{x',y''} + y'-y'' +  d(v_{x',y'},v_{x,y})$,
    i.e., $L^{i,j}_{x,y}\ge L^{i-1,j}_{x',y''}+\LCS(X(x'\dd x],Y(y'\dd y]) \ge (1+\alpha)^{-1} L^{i-1,j}_{x',y'}+\ell_{x,y}-\ell_{x',y'}$.
    The inductive assumption further yields $L^{i-1,j}_{x',y'} \ge (1+\alpha)^{3-i-j}\ell_{x',y'}$,
    and thus $L^{i,j}_{x,y} \ge (1+\alpha)^{2-i-j}\ell_{x',y'} + \ell_{x,y}-\ell_{x',y'} \ge (1+\alpha)^{2-i-j}\ell_{x,y}$
    holds as claimed.
\end{proof}

\newcommand{\Q}{\mathbf{Q}}
\begin{lemma}\label{lem:boxdp}
Given $\Ohtilde(1)$-time random access to the $\DIST_{B_{i,j}}$ matrix,
the values $D^{i-1,j}_{x',y'}$ for all vertices $v_{x',y'}\in \P_{i-1,j}$ (if $i>1$), 
and the values $D^{i,j-1}_{x',y'}$ for all vertices $v_{x',y'}\in \P_{i,j-1}$ (if $j>1$),
the values $D^{i,j}_{x,y}$ for any $q$ \emph{query vertices} $v_{x,y}\in \outv^{B_{i,j}}$ can be computed in $\Ohtilde(q+|\P\cap \inv^{B_{i,j}}|)$ time.
\end{lemma}
\begin{proof}
    If $(i,j)=(1,1)$, then \cref{def:portalwalk} and Observation~\ref{obs:dir} yield $D^{1,1}_{x,y}=d(v_{0,0},v_{x,y})$,
    and this value can be retrieved from the $\DIST_{B_{1,1}}$ matrix in $\Ohtilde(1)$ time. 
    Thus, we henceforth assume $(i,j)\ne (1,1)$.
    
    Consider a portal-respecting $(i,j)$-walk $W$ to a vertex $v_{x,y}\in \outv^{B_{i,j}}$.
    By \cref{def:portalwalk}, $W$ is a concatenation of two walks $W'$ and $W''$
    such that $W''$ starts at a vertex $v_{x',y'}\in \P\cap \inv^{B_{i,j}}$ and is entirely contained within $B_{i,j}$,
    whereas $W'$ is a portal-respecting $(i,j-1)$-walk to $v_{x',y'}$ or a portal respecting $(i-1,j)$-walk to $v_{x',y'}$. Observe that, for a fixed portal $v_{x',y'}\in  \P\cap \inv^{B_{i,j}}$, the lengths of $W'$ and $W''$ can be optimized independently.
    Consequently, by Observation~\ref{obs:dir}, \[D^{i,j}_{x,y} = \max\left(\max_{v_{x',y'}\in \P_{i-1,j} \cap \inv^{B_{i,j}}}\left\{ D^{i-1,j}_{x',y'}+ d(v_{x',y'},v_{x,y})\right\},\max_{v_{x',y'}\in \P_{i,j-1} \cap \inv^{B_{i,j}}}\left\{ D^{i,j-1}_{x',y'}+ d(v_{x',y'},v_{x,y})\right\}\right).\]
    
    A matrix (indexed by the query vertices $v_{x,y}\in \outv^{B_{i,j}}$ and all vertices $\P_{i-1,j} \cap \inv^{B_{i,j}}$) containing the values $D^{i-1,j}_{x',y'}+d(v_{x',y'},v_{x,y})$ can be obtained from a submatrix of the $\DIST_{B_{i,j}}$ matrix
    by adding $D^{i-1,j}_{x',y'}$ to all entries in the column of $v_{x',y'}$.
    These modifications preserve the Monge property, so the resulting matrix is a Monge matrix with $\Ohtilde(1)$-time random access.
    Consequently, the SMAWK algorithm~\cite{AggarwalKMSW87} allows computing row-minima,
    i.e., the values $\max_{v_{x',y'}\in\P_{i-1,j} \cap \inv^{B_{i,j}}}\big\{ D^{i-1,j}_{x',y'}+ d(v_{x',y'},v_{x,y})\big\}$.
    A symmetric procedure allows computing the values $\max_{v_{x',y'}\in \P_{i,j-1} \cap \inv^{B_{i,j}}}\big\{ D^{i,j-1}_{x',y'}+ d(v_{x',y'},v_{x,y})\big\}$,
    which lets us derive the costs $D^{i,j}_{x,y}$ for all the query vertices $v_{x,y}\in \outv^{B_{i,j}}$.
    The SMAWK algorithm takes nearly linear time with respect to the sum of matrix dimensions, 
    so the overall time complexity is $\Ohtilde(q+|\P\cap \inv^{B_{i,j}}|)$.
\end{proof}

\begin{lemma}\label{lem:lcsdp}
Given a box decomposition $\B$ of $G_{X,Y}$, a parameter $\eps \in (0,1]$,
and $\Ohtilde(1)$-time random access to the $\DIST$ matrices of all the boxes of $\B$,
a $(1+\eps)$-approximation of $\LCS(X,Y)$ can be computed in $\Ohtilde(\eps^{-1} (p_X + p_Y)^2)$ time.
\end{lemma}
\begin{proof}
Let us choose $\alpha = \Omega(\frac{\epsilon}{p_X+p_Y-2})$ so that $(1+\alpha)^{p_X+p_Y-2}=1+\eps$.
We process boxes $B_{i,j}\in \B$ in the order of non-decreasing values $i+j$, constructing the output portals $P_{i,j}$ according to \cref{lem:lcsportals} and computing the values $L^{i,j}_{x,y}$
for all $v_{x,y}\in \P_{i,j}$. 
By \cref{lem:lcsportals}, the value $L^{p_X,p_Y}_{|X|,|Y|}$ is guaranteed to be a $(1+\eps)$-approximation of $\LCS(X,Y)$.

The ordering of boxes lets us compute the values $L^{i,j}_{x,y}$ for any $q$ vertices $v\in \outv^{B_{i,j}}$ in $\Ohtilde(q+|\P\cap \inv^{B_{i,j}}|)$ time. 
By symmetry, we may focus without loss of generality on the right boundary of $B_{i,j}$,
i.e., vertices $v_{x,y}$ with $x=b_{i}^X$ and $y\in [b_{j-1}^Y\dd b_j^Y]$.
Note that the corresponding values $L^{i,j}_{x,y}$ are non-decreasing:
$D^{i,j}_{x,y} \le D^{i,j}_{x,y-1}+1$ implies $L^{i,j}_{x,y}\ge L^{i,j}_{x,y-1}$.
First, we apply \cref{lem:boxdp} to derive $L^{i,j}_{x,y}$ for the two extreme values $y\in \{b_{j-1}^Y,b_j^Y\}$.
Next, for each value $r\in [\floor{\log_{1+\alpha}L^{i,j}_{x,b_{j-1}^Y}}\dd \floor{\log_{1+\alpha}L^{i,j}_{x,b_{j}^Y}}]$, we binary search for the smallest $y\in [b_{j-1}^Y\dd b_j^Y]$ such that $\log_{1+\alpha}L^{i,j}_{x,y} \ge r$,
and include $v_{x,y}$ in $\P_{i,j}$.
The binary searches are executed in parallel, with \cref{lem:boxdp} applied to determine $L^{i,j}_{x,y}$ for all the current pivots. This way, the algorithm is implemented in $\Ohtilde(1+\log_{1+\alpha}L^{i,j}_{x,b_{j}^Y} - \log_{1+\alpha} L^{i,j}_{y,b_{j-1}^Y,y}+|\P\cap \inv^{B_{i,j}}|)$ time.
Due to $L^{i,j}_{y,b_{j-1}^Y} \ge L^{i,j-1}_{y,b_{j-1}^Y}$, the first term sums up to $\Ohtilde(\log_{1+\alpha} |Y|)=\Ohtilde(\eps^{-1}(p_X+p_Y))$ across $j\in [1\dd p_Y]$, and to $\Ohtilde(\eps^{-1}(p_X+p_Y)^2)$ across $B_{i,j}\in \B$. This also bounds the number of portals created, so the second term, which sums up to $|\P|$ across all boxes,
is also $\Ohtilde(\eps^{-1}(p_X+p_Y)^2)$.
\end{proof}

\begin{proof}[Proof of Theorem~\ref{thm:lcs}]
    The algorithm uses \cref{cor:scheme,lem:lcsdp}.
    Due to $p_X+p_Y = \frac{N+M}{\tau}$, the overall running time is  $\Ohtilde\big(nm\tau + \tfrac{(N+M)^2}{\eps \tau^2}\big)$.
    Optimizing $\tau$, we get the running time of $\Ohtilde(nm+\eps^{-1}+(nm(N+M))^{2/3}\eps^{-1/3})$.

    If the first term dominates, then $nm \ge (nm(N+M))^{2/3}\eps^{-1/3} \ge (N+M)^2 \eps^{-1}$.
    However, $\Oh(NM) = \Oh((N+M)^2 \eps^{-1})$ time is enough to compute $\LCS(X,Y)$ exactly without compression.
    If the second term dominates, then $\eps^{-1} \ge (nm(N+M))^{2/3}\eps^{-1/3} \ge nm(N+M)$.
    However, $\Ohtilde(\sqrt{nm}(N+M))=\Ohtilde(nm(N+M))$ time is enough to compute $\LCS(X,Y)$ exactly using \cref{prp:edexact} with $D=N+M$.
\end{proof}

\newcommand{\subseq}[2]{[ #1 \dd #2 ]}

\section{FPTAS For Compressed Median k-Edit Distance}\label{sec:k-ed}
The median $k$-edit distance is defined as below.
\begin{definition}
The (median) edit distance $\ED(X_1,\ldots,X_k)$ of $k$ strings $X_1,\ldots,X_k$ is the minimum total number of edits (insertions, deletions, and substitutions) needed to make all strings $X_i$ equal some string $X^*$. That is, $\ED(X_1,\ldots,X_k) = \min_{X^*} \sum_{i=1}^k \ED(X_i,X^*)$.
\end{definition}

For the (median) edit distance between $k$ strings, we show that allowing $(1+\epsilon)$-approximation gives an algorithm circumventing the bound in Theorem~\ref{thm:editDistanceLowerBoundFromSETH}:

\begin{theorem}\label{thm:k-ed}
    Given $k=\Oh(1)$ straight-line programs $\G_{X_i}$ of total size $n$ generating strings $X_i$
    of total length $N>0$ and a parameter $\eps \in (0,1]$,
    an integer between $\ED(X_1,\ldots,X_k)$ and $(1+\eps)\ED(X_1,\ldots,X_k)$
    can be computed in $\Ohtilde\left(\eps^{-\Oh(k)}n^{k/2}N^{k/2}\right)$ time.
\end{theorem}

To prove the above theorem, we use a different set of techniques than in the two-string case. Most approaches for speeding up the textbook DP algorithm for two (compressible) strings, including the aforementioned results in this paper, rely on the ability to perform computations involving DIST matrices efficiently. These computations crucially depend on the fact that DIST matrices satisfy the Monge property. However, for the natural high-dimensional generalization of DIST matrices, we do not know of any analog of the Monge property they satisfy that allows us to perform similar computations even for three-string similarity problems. Indeed, most natural generalizations of the Monge property seem to not hold even in the three-string setting (see~\cref{sec:disttensor} for more details). 
Thus, it appears unlikely that, for example, an algorithm that partitions the DP table into boxes and computes the DP values on the boundary of each box using computations involving DIST matrices would be substantially more efficient than the textbook edit-distance algorithm, even in the three-string setting.

This motivates us to instead use the window-respecting alignment scheme that has appeared in approximation algorithms for edit distance (e.g.,~\cite{CDGKS18, GRS:20}). 

\subsection{Window-Respecting Alignments}\label{section:wra}
We will assume that $\ED(X_1,\ldots,X_k)$ lies between $D$ and $2D$ for some known $D$ at the loss of a $\log N$ factor in the runtime. We partition $X_1$ into $\lceil |X_1|/\tau \rceil$ disjoint windows $W_{1, 1}$ to $W_{1, N/\tau}$ each of length $\tau$ (without loss of generality; we could always e.g. pad each string with an equal amount of a new dummy character to ensure $|X_1|$ is a multiple of $\tau$, without asymptotically affecting their size or compression size). That is, $W_{1, j} = X_1[(j-1)\tau + 1, j\tau]$. 

We define for $X_2, \ldots, X_k$ possibly overlapping windows indexed by (i) $\Delta$, a guess for the (signed) difference between the length of $W$ and the corresponding window in $X_1$ and (ii) the starting position $p$ of the window. More formally, the windows are indexed by $W_{i, \Delta, p}$. Throughout the section, let $\sigma := \max\{\lfloor \epsilon D \tau / |X_1| \rfloor, 1\}$ and $R_d(x)$ denote $x$ rounded down to the nearest multiple of $\sigma$. Then $W_{i, \Delta, p} = X_i[p\sigma + 1\dd R_d(p\sigma + \tau + \Delta)]$ (or is the empty string ``starting'' at position $p\sigma + 1$ if $R_d(\min\{p\sigma + \tau + \Delta, |X_i|\}) < p\sigma + 1$). If $R_d(p\sigma + \tau + \Delta) > |X_1|$, $W_{i, \Delta, p}$ is not included in our set of windows. We will define this window for:

\begin{itemize}
    \item All $\Delta$ in $\{0, 1, -1, \lfloor(1+\epsilon)\rfloor, -\lfloor(1+\epsilon)\rfloor, \lfloor(1+\epsilon)^2\rfloor, \ldots \lfloor(1+\epsilon)^{\lceil \log_{1+\epsilon} 2\tau/\epsilon^2 \rceil}\rfloor\} \cup \{-\tau\}$ for which $\tau + \Delta \geq 0$,
    \item All $p$ from $0$ to $\lfloor|X_i|/\sigma\rfloor$.
\end{itemize}

It suffices to consider windows of size at most $2\tau/\epsilon^2$ by the following lemma:

\begin{lemma}\label{lemma:lowskew}
Given $X_1, X_2, \ldots, X_k$ and a parameter $\tau$, for $J = |X_1|/\tau$, let $X^*$ be the string such that $\ED(X_1, X_2, \ldots X_n) = \sum_i \ED(X_i, X^*)$. There exists a partition of each $X_1$ into substrings $\{X_{1, j}\}_{j \in [J]}$, disjoint substrings of the other $X_i$, $\{X_{i, j}\}_{j \in [J]}$, and a partition of $X^*$ into substrings $\{X_j^*\}_{j \in [J]}$ such that:
\begin{itemize}
    \item $|X_{1,j}| = \tau$ for all $j$.
    \item For any $j$ and $j < j'$, $X_{i, j}$ appears before $X_{i, j'}$ in $X_i$.
    \item $\max_{i, j} |X_{i, j}| \leq 2\tau / \epsilon^2$.
    \item \[\sum_{j \in [J]} \ED(X_{ij}, X^*_j) + |X_i| - \sum_{j \in [J]} |X_{i,j}| \leq (1+3\epsilon)\ED(X_i, X^*), \]
    Which implies:
    \[\sum_{j \in [J]} \ED(X_{1, j}, X_{2, j}, \ldots, X_{k, j}) + \sum_{i > 1} (|X_i| - \sum_{j \in [J]} |X_{i, j}|) \leq (1 + 3\epsilon) \ED(X_1, X_2, \ldots X_k).\]
    
    That is, the cost of the alignment that aligns $X_{1,j}$ with each $X_{i,j}$, and then deletes all characters in $X_2$ to $X_k$ that are unaligned with some $X_{1,j}$ is at most $(1 + 3\epsilon) \ED(X_1, X_2, \ldots X_k)$.
\end{itemize}
\end{lemma}

Effectively, Lemma~\ref{lemma:lowskew} says that there is a near-optimal alignment that aligns the windows of $X_1$ to substrings of the other strings that are not more than $1/\epsilon^2$ times larger. 

\begin{proof}

We partition $X_1$ into substrings of length $\tau$, $\{X_{1,j}\}_{j \in [J]}$. $X^*$ can be partitioned into substrings  $\{\tilde{X}^*_j\}_{j \in [J]}$ such that $\ED(X_1, X^*) = \sum_{j \in [J]} \ED(X_{1, j}, \tilde{X}^*_j)$.

First, we will ``realign'' $X_1$ and $X^*$ to ensure no $\tilde{X}^*_j$ is much larger than $X_{1, j}$. Call a contiguous subsequence of $[J]$, $\subseq{j}{j'}:= \{j, j+1, \ldots j'\}$, ``skewed'' if $\sum_{m \in \subseq{j}{j'}} |X_{1, m}| < \frac{\epsilon}{2} \sum_{m \in \subseq{j}{j'}} |\tilde{X}^*_m|$. Let us take a ``maximal'' set $S$ of disjoint skewed contiguous subsequences, i.e. a set $S$ such that (i) all the subsequences in $S$ are disjoint (ii) for every contiguous subsequence $s$ in $S$, there is no skewed contiguous subsequence $s'$ such that $s \subset s'$ and (iii) there is no skewed contiguous subsequence that is not in $S$ but also is completely disjoint from every element of $S$.

For each skewed contiguous subsequence $\subseq{j}{j'}$ in $S$, note that $j'+1$ does not appear in any element of $S$ (otherwise, $\subseq{j}{j'}$ and this element can be combined to form a longer skewed contiguous subsequence, violating (ii)), and $\subseq{j}{j'+1}$ is not skewed (again, $\subseq{j}{j'+1}$ being skewed would violate (ii) since $\subseq{j}{j'} \subset \subseq{j}{j'+1}$). Take $S$ and replace each $\subseq{j}{j'}$ with $\subseq{j}{j'+1}$ to get $S'$. For each $\subseq{j}{j'+1} \in S'$, we have:

\begin{equation}\label{eq:lengthbound}
\frac{\epsilon}{2} \sum_{m \in \subseq{j}{j'+1}} |\tilde{X}^*_m| \leq \sum_{m \in \subseq{j}{j'+1}} |X_{1, m}| \leq 2 \sum_{m \in \subseq{j}{j'}} |X_{1, m}|< \epsilon \sum_{m \in \subseq{j}{j'+1}} |\tilde{X}^*_m|.
\end{equation}

The right-hand side of~\eqref{eq:lengthbound} implies:

\[\ED(\bigcirc_{m  \in \subseq{j}{j'+1}} X_{1, m}, \bigcirc_{m  \in \subseq{j}{j'+1}} \tilde{X}^*_m) \geq (1 - \epsilon) |\bigcirc_{m  \in \subseq{j}{j'+1}} \tilde{X}^*_m|.\]

It also implies that for any partition of $\bigcirc_{m  \in \subseq{j}{j'+1}} \tilde{X}^*_m$ into substrings $\{X^*_m\}_{m  \in \subseq{j}{j'+1}}$, we have:

\[\sum_{m  \in \subseq{j}{j'+1}} \ED(X_{1, m}, X^*_m) \leq (1+\epsilon) |\bigcirc_{m  \in \subseq{j}{j'+1}} \tilde{X}^*_m|.\]

And so if $\epsilon$ is sufficiently small:

\[\sum_{m\in \cup_{e \in S'} e} \ED(X_{1,m}, X^*_m) \leq \frac{1+\epsilon}{1-\epsilon}\sum_{m\in \cup_{e \in S'} e} \ED(X_{1,m}, \tilde{X}^*_m) \leq (1+3\epsilon)\sum_{m\in \cup_{e \in S'} e} \ED(X_{1,m}, \tilde{X}^*_m)\]

In particular, because of the left-hand side of~\eqref{eq:lengthbound}, we can choose the partition of $\bigcirc_{m  \in \subseq{j}{j'+1}} \tilde{X}^*_m$ that splits it into substrings $\{X^*_m\}_{m  \in \subseq{j}{j'+1}}$, each of length at most $2\tau/\epsilon$. 

Now if we set $X^*_m = \tilde{X}^*_m$ for any $m$ not in a subsequence in $S'$, we trivially have:

\[\sum_{m\notin \cup_{e \in S'} e} \ED(X_{1,m}, X^*_m) \leq \sum_{m\notin \cup_{e \in S'} e} \ED(X_{1,m}, \tilde{X}^*_m)\]

And also $X^*_{m} \leq 2\tau/\epsilon$ for all such $m$ (otherwise, $m$ should appear in some subsequence in $S'$ by condition (iii)). So we've found a partition of $X^*$ into substrings $\{X^*_j\}_{j \in [J]}$ such that $|X^*_j| \leq \frac{1}{\epsilon}|X_{1,j}|$ for all $j$, and:

\[\sum_{j \in [J]} \ED(X_{1,j}, X^*_j) \leq (1+3\epsilon)\sum_{j \in [J]} \ED(X_{1,j}, \tilde{X}^*_j).\]

Now, we will use this partition to determine $\{X_{i,j}\}_{i > 1, j \in [J]}$. For each $i$, $X_i$ can be partitioned into substrings $X'_{i, j}$ such that $\ED(X_i, X^*) = \sum_{j \in [J]} \ED(X'_{i,j}, X^*_j)$. If $|X'_{i,j}| \leq 2\tau/\epsilon^2$, we set $X_{i,j} = X'_{i,j}$. If any $X'_{i,j}$ has length larger than $2\tau/\epsilon^2 > |X^*_j| / \epsilon$, then $\ED(X'_{i,j}, X^*_j) \geq (1-\epsilon)|X'_{i,j}|$. On the other hand:

\[\ED(\gamma, X^*_j) + |X'_{i,j}| \leq (1+\epsilon)|X'_{i,j}| \leq \frac{1+\epsilon}{1-\epsilon} \ED(X'_{i,j}, X^*_j) \leq (1+3\epsilon) \ED(X'_{i,j}, X^*_j) \]
for the empty string $\gamma$. So we can now choose $X_{i,j}$ to be any empty substring of $X'_{i,j}$. These choices of $X_{i,j}$ give the properties of the lemma, completing the proof.
\end{proof}

\newcommand{\start}{\texttt{s}}
\newcommand{\nd}{\texttt{e}}

Let $\mathcal{W}_1$ be the set of all windows we partition $X_1$ into, and $\mathcal{W}_i$ be the set of windows we define for $X_i$. Let $\start(W)$ denote the index of the first character in $W$, and $\nd(W)$ denote the index of the last character. For $k$ strings, we define a window-respecting alignment as follows:

\begin{definition}\label{def:wra}
A window respecting alignment is a function $f: \mathcal{W}_1 \rightarrow \mathcal{W}_1 \times \mathcal{W}_2 \times\cdots \times \mathcal{W}_k$ with the following properties:

\begin{itemize}
    \item For all $W \in \mathcal{W}_1$, $f(W)_1 = W$.
    \item For any $j < j'$ and any $i$, $\nd(f(W_{1, j})_i) < \start(f(W_{1, j'})_i)$.
\end{itemize}

Let $r_i(f)$ denote the number of characters in $X_i$ that are not contained in $f(W)_i$ for any $W \in \mathcal{W}_1$. The cost of a window-respecting alignment is defined as follows:

\[\ED(f) := \sum_{j\in [J]} \ED(f(W_{1, j})) + \sum_i r_i (f).\]
\end{definition}

Let $\mathcal{F}$ be the set of all window-respecting alignments. The following lemma shows that window-respecting alignments approximate normal alignments:
\begin{lemma}\label{lemma:ked-window}
\[\ED(X_1, X_2, \ldots X_k) \leq \min_{f \in \mathcal{F}} \ED(f) \leq (1+13\epsilon k) \ED(X_1, X_2, \ldots X_k) .\]
\end{lemma}

\begin{proof}
The first inequality follows because for any $f$, there is an alignment that for all $j$ exactly aligns $W_{1,j}$ with the windows in $f(W_{i,j})$ at cost at most $\ED(f(W_{i,j}))$, and uses $\sum_i r_i(f)$ deletions to handle the remaining characters in each string. 

Next, we show that there exists $f \in \mathcal{F}$ such that $\ED(f) \leq (1+\epsilon k) \ED(X_1, X_2, \ldots X_k)$. Let us take the substrings $X_{i, j}$ given by Lemma~\ref{lemma:lowskew}. Note that $W_{1, j} = X_{1, j}$.

If $|X_{i, j}| \leq \epsilon \ED(W_{1, j}, X_{2, j} \ldots X_{k , j}) + 3(\sigma-1)$, let $W_{i,j}$ be the empty window ``starting'' at index $\nd(W_{i, j-1})+1$, or if $j = 1$, at index 1. Then we have $|X_{i, j}| - |W_{i, j}| \leq \epsilon \ED(W_{1, j}, X_{2, j} \ldots X_{k , j}) + 3(\sigma-1)$.

Otherwise, let $W_{i, j}$ be the longest window $W_{i, \Delta, p}$ that is a substring of $X_{i, j}$. Note that $|X_{i, j}|$ and $|W_{1, j}|$ differ by at most $\ED(W_{1, j}, X_{2, j} \ldots X_{k , j})$ and $|W_{i, j}| \leq 2\tau/\epsilon^2$ for all $i, j$. If $\epsilon$ is a sufficiently small constant, this implies there is a choice of $W_{i, j}$ such that $|X_{i, j}| - |W_{i, j}| \leq \epsilon \ED(W_{1, j}, X_{2, j} \ldots X_{k , j}) + 3(\sigma-1)$. We can identify $W_{i,j}$ as follows: Take $X_{i,j}$ and delete at most $\sigma - 1$ characters from the beginning until it starts at $p\sigma + 1$ for some integer $p$ to get $\tilde{X}_{i,j}$. We have $|X_{i, j}| - |\tilde{X}_{i, j}| \leq \sigma - 1$, and so $|\tilde{X}_{i, j}|$ and $|W_{1, j}|$ differ by at most $\ED(W_{1, j}, X_{2, j} \ldots X_{k , j}) + (\sigma - 1)$ characters. Choose $\Delta$ such that $\frac{1}{1+\epsilon} (\ED(W_{1, j}, X_{2, j} \ldots X_{k , j}) + (\sigma - 1)) \leq \Delta \leq \ED(W_{1, j}, X_{2, j} \ldots X_{k , j}) + (\sigma - 1)$. $W_{i, \Delta, p}$ is a prefix of $\tilde{X}_{i,j}$ containing all but at most the last $\epsilon (\ED(W_{1, j}, X_{2, j} \ldots X_{k , j}) + (\sigma - 1)) + (\sigma - 1)$ characters of $\tilde{X}_{i,j}$. In turn, if $\epsilon$ is sufficiently small we have $|X_{i, j}| - |W_{i, j}| \leq \epsilon \ED(W_{1, j}, X_{2, j} \ldots X_{k , j}) + 3(\sigma-1)$.

In turn, by triangle inequality and since $\sigma - 1 \leq \epsilon k D \tau / |X_1|$:

\[\ED(W_{1, j}, W_{2, j} \ldots W_{k, j}) \leq (1 + \epsilon k) \ED(W_{1, j}, X_{2, j} \ldots X_{k, j}) + 3\epsilon k D \tau/|X_1|.\]

We now choose $f(W_{1, j}) = (W_{1, j}, W_{2, j} \ldots W_{k, j})$. We also have that the number of characters $f$ does not align within $X_{2, j}, X_{3, j} \ldots X_{k, j}$ is at most $\epsilon k\ED(W_{1, j}, X_{2, j} \ldots X_{k , j}) + 3\epsilon k D \tau/|X_1|$.

Putting it all together and using Lemma~\ref{lemma:lowskew} we get:
\begin{align*}
\ED(f) := &\sum_{j\in [J]} \ED(f(W_{1, j})) + \sum_i r_i (f)\\
\leq &
(1 + \epsilon k) \sum_j \ED(W_{1, j}, X_{2, j} \ldots X_{k , j}) +\frac{|X_1|}{\tau} \cdot \frac{3 \epsilon k D \tau}{|X_1|} \\
&\qquad+ \epsilon k \sum_j \ED(W_{1, j}, X_{2, j} \ldots X_{k , j})+\frac{|X_1|}{\tau} \cdot \frac{3 \epsilon k D \tau}{|X_1|} + \sum_i (|X_i| - \sum_{j \in [J]} |X_{i, j}|)\\
\leq &
(1+2 \epsilon k) \sum_j \ED(X_{1,j}, X_{2,j}, \ldots X_{k,j}) + 6 \epsilon k \ED(X_1, X_2, \ldots, X_k) + \sum_i (|X_i| - \sum_{j \in [J]} |X_{i, j}|)\\
\leq &
 (1+13\epsilon k) \ED(X_1, X_2, \ldots X_k).
\end{align*}
\end{proof}

\subsection{An Efficient Algorithm for Window-Respecting Alignments}

Our algorithm, denoted \textsc{$k$-ED-Alg}, is as follows:

\begin{mdframed}
\begin{enumerate}[leftmargin=*]
    \item Let $\mathcal{D} :=\{1, 2, 4, \ldots, 2k\tau/\epsilon^2\} $. For $X_1$ and each $d$ in $\mathcal{D}$, identify a set $\tilde{\mathcal{W}}_{1, d}$ of ``representative''  strings such that (i) $|\tilde{\mathcal{W}}_{1,d}| = O(n \tau/\epsilon d)$ and (ii) for every window $W_{1, i}$, there is some string $\texttt{shift}_d(W_{1, i}) \in \tilde{\mathcal{W}}_{1, d}$ in such that $\ED(W_{i, \Delta, p}, \texttt{shift}_d(W_{i, \Delta, p})) \leq \epsilon d$.
    \item For each other string $X_i$,each value of $d$ in $\mathcal{D}$, and each value of $\Delta$, identify a set of ``representative'' length $\tau+\Delta$ strings $\tilde{\mathcal{W}}_{i, d, \Delta}$ such that (i) $|\tilde{\mathcal{W}}_{i, d, \Delta}|$ = $O(n (\tau + \Delta)/ \epsilon d)$, and (ii) for every window $W_{i, \Delta, p}$, there is some string $\texttt{shift}_d(W_{i, \Delta, p}) \in \tilde{\mathcal{W}}_{i, d, \Delta}$ such that $\ED(W_{i, \Delta, p}, \texttt{shift}_d(W_{i, \Delta, p})) \leq \epsilon d$.

    \item Let $\tilde{\mathcal{W}}_{i, d} = \cup_\Delta \tilde{\mathcal{W}}_{i, d, \Delta}$. For each $d \in \mathcal{D} $ and every $k$-tuple of strings $\tilde{W}_{1, d}, \tilde{W}_{2, d} \ldots \tilde{W}_{k, d}$ in $\tilde{\mathcal{W}}_{1, d} \times \tilde{\mathcal{W}}_{2, d} \times \cdots \times \tilde{\mathcal{W}}_{k, d}$,  compute the median distance of this $k$-tuple if it is less than $d$. Store this as $\tilde{\ED}(\tilde{W}_{1, d}, \tilde{W}_{2, d} \ldots \tilde{W}_{k, d}) + \epsilon kd$. If the true median distance of these windows is greater than $d$, store $\tilde{\ED}(\tilde{W}_{1, d}, \tilde{W}_{2, d} \ldots \tilde{W}_{k, d}) = \infty$ instead. 
    \item Our algorithm solves the following dynamic program:
    \[c(x_1, x_2, \ldots x_k) = \min
    \begin{cases} 
    \min_{i \neq 1}  c(x_1, x_2, \ldots, x_i - \sigma, \ldots, x_k) + \sigma\\
    \min_{W_1, \ldots W_k \in \mathcal{W}_{1} \times \cdots \times \mathcal{W}_{k}: \forall_i \nd(W_i) = x_i, W_1 = \cdots = W_k} [c(\start(W_1)-1, \ldots, \start(W_k)-1)] \\
    \min_{W_1,\ldots W_k \in \mathcal{W}_{1} \times \cdots \times \mathcal{W}_{k}: \forall_i \nd(W_i) = x_i} [c(\start(W_1)-1, \ldots,
    \start(W_k)-1)\\ \qquad + \min_{d \in \mathcal{D} } \tilde{\ED}(\texttt{shift}_d(W_1), \ldots, \texttt{shift}_d(W_k))] 
    \end{cases} 
    \]
    For every $k$-tuple such that $x_1$ is a multiple of $\tau$, $x_2, \ldots, x_k$ are all multiples of $\sigma$, and such that $|x_i - x_1| \leq 4D$ for all $i$.
    
    The base case is $c(0, 0, \ldots, 0) = 0$, and our final output is $c(|X_1|, R_u(N_2), \ldots R_u(N_k))$, where $R_u(x)$ denotes $x$ rounded up to the nearest multiple of $\sigma$.
\end{enumerate}
\end{mdframed}

At a high-level, in steps 1 and 2 of \textsc{$k$-ED-Alg} we exploit the compression of the input strings to identify a small set of ``representative'' strings for each $X_i$, such that for each window in $X_i$ there is a representative string within small edit distance of that window. In step 3, we then compute the median distance between $k$-tuples of representative strings (instead of between all $k$-tuples of windows). Since all windows are within a small distance of some representative string, this also gives for all $k$-tuples of windows a reasonable approximation of their median distance. Step 4 of \textsc{$k$-ED-Alg} uses these approximations to solve a natural DP for finding an optimal window-respecting alignment. This DP is the same as the standard DP for edit distance, but instead of matching characters we are only allowed to match windows, at cost equal to (the approximation of) their median distance. 

We first bound the runtime of \textsc{$k$-ED-Alg}. The following lemmas show that Steps 1 and 2 of \textsc{$k$-ED-Alg} can be performed efficiently (as well as their correctness):

\begin{lemma}\label{lemma:lz77-repetition}
Given a straight-line program $\G$ of size $n$ that generates a string $X$ of size $n$, a length parameter $\tau$, and a parameter $\delta_{\max} \leq \tau$, there exists an algorithm that in time $O(|X|)$ finds (an implicit representation of) a set $S$ of $O(n\tau/\delta_{\max})$ substrings of length at most $\tau$ such that for every length $\tau$ substring of $X$, $x$, there is a string $\texttt{shift}(x)$ in $S$ such that $\ED(x, \texttt{shift}(x)) \leq \delta_{\max}$. We can also construct a data structure that identifies $\texttt{shift}(x)$ given the starting location of $x$ in $X$ using $O(|X|)$ preprocessing time and $O(1)$ query time.
\end{lemma}

\begin{proof}
If $\delta_{\max} \geq \tau$, we can trivially choose $S$ that only contains the empty substring, and the data structure just returns the empty substring for any query. So assume $\delta_{\max} < \tau$.

Given that $\G$ has size $n$, the optimal LZ77 factorization of $X$ has size at most $n$~\cite{Rytter2003}. We will first show the existence of $S$ for any $X$ that has an LZ77 factorization of size at most $n$. For brevity, we will not go into the details of LZ77 factorization here. The key property we need is that a string $X$ that has a LZ77 factorization of size $n$ can be written as $X_1 \circ X_2 \circ X_3$, where $X_1$ is a string with LZ77 factorization of size $n-1$, $X_2$ is a substring of $X_1$, and $X_3$ is a single character. Moreover, the factorization gives the location of $X_2$ in $X_1$. 

Inductively, suppose we have constructed $S$, a set of at most $3 (n-1) \tau / \delta_{\max}$ substrings that has the desired properties for $X_1$. For all ``good'' indices $i <  |X_2| - \tau$, the length $\tau$ substring starting at the $i$th character in $X_2$ is fully contained in $X_2$, and thus is a substring in $X_1$. This leaves at most $\tau + 1$ ``bad'' indices where the length $\tau$ substring starting at these indices may not have a nearby string in $S$: those starting at indices $|X_1| - \tau + 1$ to $|X_1|$ of $X_1$, and the substring starting at index $|X_2| - \tau + 2$ of $X_2$. Consider the length $\tau$ substring starting at every $(\delta_{\max}/2)$-th position in indices $|X_1| = \tau + 1$ to $|X_1|$ of $X_1$, as well as the length $\tau$ substring starting at index $|X_2| - \tau + 2$ of $X_2$. This set of strings has size at most $2 \tau/\delta_{\max} + 1 \leq 3 \tau/\delta_{\max}$, and every length $\tau$ substring starting at a bad index is within edit distance $\delta_{\max}$ of some string in this set. So adding these strings to $S$ gives that $S$ now has size at most $3 n \tau / \delta_{\max}$ and has the desired properties.

For an efficient implementation of this procedure, we can compute the optimal LZ77 factorization in $O(|X|)$ time~\cite{RodehPE81}. Given the LZ77 factorization, we decompose $X$ into $X_1 \circ X_2 \circ X_3$ as before, and recursively compute an array $A$ for indices in $\subseq{1}{|X_1|-\tau+1}$ and set $B$ with the following property: the length $\tau$ substrings starting at indices $i$ and $A[i]$ are within edit distance $\delta_{\max}$, and $A$ has at most $3(n-1)\tau/\delta_{\max}$ distinct values, which are exactly the values in $B$. 

Since the LZ77 factorization gives us the position of $X_2$ in $X_1$, we can fill in $A$ for the ``good'' indices in $X_2$ in time linear in the number of good indices. We can also fill in the values of $A$ for the bad indices, in time linear in the number of bad indices, and add these values to $B$. Overall, the algorithm takes linear time to compute $A, B$. $A$ now serves as the desired efficient data structure, and $B$ as our implicit representation of $S$.
\end{proof}

There are $O(\log N)$ values of $d$ and $\tilde{O}(\log N / \epsilon)$ values of $\Delta$, so we can do Steps 1 and 2 in time $\tilde{O}(N / \epsilon)$ time.
We also show that Step 3 can be performed efficiently:

\begin{lemma}\label{lemma:k-lms}
Given strings $X_1, X_2, \ldots X_k$, there exists a data structure that can be computed in $O(\sum_i |X_i|)$ time that can answer queries of the following form in $O(d^k)$ time: Given indices $s_1, s_2, \ldots, s_k$ and $e_1, e_2, \ldots, e_k$, if $\ED(X_1[s_1\dd e_1], X_2[s_2\dd e_2], \ldots,\allowbreak X_k[s_k\dd e_k]) \leq d$, output $\ED(X_1[s_1\dd e_1], X_2[s_2\dd e_2], \ldots,\allowbreak X_k[s_k\dd e_k])$, otherwise output $\infty$. 
\end{lemma}

\begin{proof}
Given $X_1, X_2, \ldots X_k$, let $\texttt{Slide}_{d_2, d_3, \ldots d_k}(j) = \max\{q: X_1[j\dd q] = X_2[j+d_2\dd q+d_2] = X_3[j+d_3\dd q+d_3]= \cdots =X_k[j+d_k\dd q+d_k]\}$. We can rewrite $\texttt{Slide}_{d_2, d_3, \ldots d_k}(j)$ as $\min_{i \in \{2, 3, \ldots, k\}} \max\{q : X_1[j\dd j+q] = X_i[j+d_i\dd q+d_i]\}$. Section 2.3 of~\cite{LMS98} shows that we can compute $\max\{q : X_1[j\dd j+q] = X_i[j+d_i;q+d_i]\}$ for any $i, j, d_i$ in $O(1)$ time after $O(|X_1|+|X_i|)$ preprocessing time. So we can compute $\texttt{Slide}_{d_2, d_3, \ldots d_k}(i)$ in $O(1)$ time after $O(\sum_i |X_i|)$ preprocessing time (recall that $k = O(1)$).

Let $L^{h}(d_2, d_3, \ldots d_k)$ be the largest value of $j$ such that $\ED(X_1[1\dd j], X_2[1\dd j+d_2], X_3[1\dd j+d_3], \ldots, X_k[1\dd j+d_k]) \leq h$. We have the following recurrence relation:

\[L^{h}(d_2, d_3 \ldots, d_k) = \texttt{Slide}\left(\max \left.
\begin{cases}
L^{h-1}(d_2+1, d_3+1 \ldots, d_k+1)\\
\max_i L^{h-1}(d_2, d_3, \ldots, d_i - 1, \ldots, d_k)\\
\max_{e \in \{0, 1\}^k} L^{h-w}(d_2 - e_2 + e_1, d_3 - e_2 + e_1, \ldots, d_k - e_k + e_1)\\
\qquad \text{ for }w = \min_{i:e_i = 1} |\{j \neq i: X_i[x_i] \neq X_j[x_j] \lor e_j = 0\}|
\end{cases}
\right\}
\right)
\]

The first case considers deleting from $X_1$, the second case considers deleting a character from any of $X_2, X_3, \ldots X_k$, and the third case considers inserting characters into some subset of the strings (for which $e_i = 0$), and then matching the inserted characters with a character in the remaining strings (for which $e_i = 1$), such that we use at most $w$ insertions or substitutions.

Each $L^h(\cdot)$ only depends on $O(1)$ other values, and so we can compute each value in $O(1)$ time. In turn, we can compute the values $L^h(d_2, d_3, \ldots d_k)$ for all $0 \leq h \leq d, 0 \leq d_2 + d_3 + \cdots + d_k \leq d$ in $O(d^k)$ time. Our output for the edit distance is the smallest $h$ such that $L^h(|X_2|-|X_1|, |X_3|-|X_1|, \ldots, |X_k|-|X_1|) \geq |X_1|$, or $\infty$ if $L^d(|X_2|-|X_1|, |X_3|-|X_1|, \ldots, |X_k|-|X_1|) < |X_1|$.
\end{proof}
The total number of strings in any $\tilde{\mathcal{W}}_{i, d}$ is $\sum_{\Delta \in \mathcal{D}} O(n(\tau + \Delta)/\epsilon d) = O(n\tau/\epsilon^3 d)$. In turn, combined with~\cref{lemma:k-lms}, the total time needed to compute $\tilde{\ED}$ for all $k$-tuples in $\tilde{\mathcal{W}}_{1, d} \times \tilde{\mathcal{W}}_{2, d} \times \cdots \times \tilde{\mathcal{W}}_{k, d}$ is $O(n^k \tau^k / \epsilon^{3k})$. There are $O(\log N)$ choices of $d$, so in total this step takes time $\tilde{O}(n^k \tau^k / \epsilon^{3k})$.

For Step 4, it takes $O(1)$ time to process the first case in the recurrence relation. For the second and third case, there are $O(\log N)$ values of $d$, $O(\log N / \epsilon)$ values of $\Delta$, and for each $i, x_i, \Delta$ there is 1 window $W_{i, \Delta, p}$ such that $\tilde{e}(W_i) = x_i$. 
\cref{lemma:k-lms} gives an $O(1)$-time method to determine if $W_1 = W_2 = \cdots = W_k$ in the second case, and we have precomputed all the necessary values in the third case. So, the time to compute each $c(x_1, x_2, \ldots x_k)$ is $O((\log^2 N/\epsilon)^k)$. 

The number of tuples $x_1, x_2, \ldots x_k$ such that $\sum_{i \neq 1} |x_i - x_1| \leq D$ is $O(ND^{k-1})$. Of these, fraction $O(\frac{1}{\lceil \epsilon D \tau/N \rceil^{k-1} \tau})$ satisfy that $x_1$ is a multiple of $\tau$ and $x_2 \ldots x_k$ are multiples of $\lceil \epsilon D \tau/N \rceil$. So the number of entries we need to compute is $O(N^k/\epsilon^{k} \tau^k)$, and the total time to compute all these entries is $\tilde{O}(N^k/\epsilon^{2k} \tau^k)$.

Putting it all together, Steps 3 and 4 dominate the runtime with total runtime $\tilde{O}(N^k/\epsilon^{2k} \tau^k + n^k \tau^k / \epsilon^{3k})$. Setting $\tau = (N/n\epsilon)^{1/2}$, we get an overall runtime of $\tilde{O}(N^{k/2}n^{k/2} \epsilon^{-5k/2})$.

We complete our analysis by showing that the final value computed by \textsc{$k$-ED-Alg} is close to $\ED(X_1, X_2, \ldots X_k)$.

\begin{lemma}\label{lemma:k-ed-approx}
\textsc{$k$-ED-Alg} outputs $\tilde{D}$ such that
\[\ED(X_1, X_2, \ldots X_k) \leq \tilde{D} \leq (1+19\epsilon k) \ED(X_1, X_2, \ldots X_k).\]
\end{lemma}

\begin{proof}
Consider any window-respecting alignment for which $f(W_{1,j}) = (W_{1,j}, W_{2,j}, \ldots, W_{k,j})$. If $\\ \ED(W_{1,j}, W_{2,j}, \ldots, W_{k,j}) > 0$, let $d_j$ be the smallest value in $\mathcal{D}$ such that $d_j \geq \ED(W_{1,j}, W_{2,j}, \ldots, W_{k,j}) + \epsilon k d_j$. By \cref{lemma:lz77-repetition}, for every $j$ and if $\epsilon$ is sufficiently small, by triangle inequality we have:

\begin{align*}
\ED(W_{1,j},W_{2,j}, \ldots, W_{k,j}) &\leq [\tilde{\ED}(\texttt{shift}_{d_j}(W_{1,j}),\texttt{shift}_{d_j}(W_{2,j}), \ldots, \texttt{shift}_{d_j}(W_{k,j}))]  + \epsilon k d_j \\
&\leq \ED(W_{1,j},W_{2,j}, \ldots, W_{k,j}) + 2 \epsilon k d_j\\
&\leq (1 + 5 \epsilon k)\ED(W_{1,j},W_{2,j}, \ldots, W_{k,j}).\\
\end{align*}

The first inequality implies that any path through the DP table for $c$ has total cost at least that of some window-respecting alignment, which by Lemma~\ref{lemma:ked-window} gives the first inequality in the lemma statement. The second inequality implies that for the best window-respecting alignment, there is a path through the DP table such that the cost of the path through the DP table is no more than $(1+5\epsilon)$ times the cost of the window-respecting alignment. Furthermore, this path only goes through points in the DP table such that $|x_i - x_1| \leq 4D$ for all $i$, i.e., is considered by \textsc{$k$-ED-Alg}. Combined with Lemma~\ref{lemma:ked-window} this gives the second inequality in the lemma statement if $\epsilon$ is sufficiently small.
\end{proof}

We can now compute a $(1+\epsilon)$-approximation of the edit distance by rescaling $\epsilon$ appropriately and running $\textsc{$k$-ED-Alg}$ for all $D$ that are powers of 2, giving Theorem~\ref{thm:k-ed}. One could also extract the alignment achieving this edit distance by using standard techniques to retrieve a path through the DP table, and applying these same techniques to the DP tables used in invocations of Lemma~\ref{lemma:k-lms} as a subroutine; we omit the details here.

%\section{Technical Contributions for k-Strings}
%Now that we have shown one of our $k$-strings results we will summarize the rest of our $k$-strings results. 
%\input{SummaryCompressedApproximate}
%\input{TechnicalSummaryHammingDistance}

\section{FPTAS For Center Distance}~\label{sec:k-ced}
The center distance problem is defined as follows: 
\begin{definition}
The center (edit) distance $\CED(X_1,\ldots,X_k)$ of $k$ strings $X_1,\ldots,X_k$ is defined as $\CED(X_1,\ldots,X_k) = \min_{X^*} \max_i \ED(X_i,X^*)$. That is, it is the smallest value $D$ such that by making at most $D$ edits to each $X_i$, we can transform them all into the same string $X^*$.
\end{definition}

In this section we prove Theorem~\ref{thm:k-center}:

\begin{theorem}\label{thm:k-center}
    Given $k=\Oh(1)$ straight-line programs $\G_{X_i}$ of total size $n$ generating strings $X_i$
    of total length $N>0$ and a parameter $\eps \in (0,1]$,
    an integer between $\CED(X_1,\ldots,X_k)$ and $(1+\eps)\CED(X_1,\ldots,X_k)$
    can be computed in $O\left(\eps^{-\Oh(k)}n^{k/2}N^{k/2+o(1)}\right)$ time.
\end{theorem}

Prior to our work, the best known algorithm result for the center distance problem was the exact $O(N^{2k})$-time algorithm of~\cite{NicolasRivals05}. Our framework for the algorithm is similar to the framework from the previous section which uses window-respecting alignments. 

\newcommand{\tup}[1]{\texttt{tup}_{#1}}
\newcommand{\tupd}{\tup{D}(X_1, X_2, \ldots X_k)}
\newcommand{\norm}[1]{\left|\left|#1\right|\right|}

Our algorithm will actually solve a more general problem of computing an approximation of a set of values which we call the \textit{edit tuples}. We again assume $\CED(X_1, \ldots, X_k)$ lies between $D$ and $2D$ for some known (power of 2) $D$. 

\begin{definition}
Given strings $X_1, X_2, \ldots, X_k$, an edit tuple of these strings is a vector $v \in \mathbb{Z}_{\geq 0}^k$ such that there exists $X^*$ for which $\ED(X_i, X^*) \leq v_i$ for all $i$. We denote the set of all edit tuples in $\{0, 1, \ldots, D\}^k$ of $X_1, X_2, \ldots X_n$ by $\tupd$.

We say that $S$ is a $\Delta$-approximation of $\tupd$ if for each $v \in S$, there is a vector $v' \in \tupd$ such that $v' \leq v$, and for each $v \in \tupd$, there is a vector $v' \in S$ such that  $v' \leq v + \Delta \cdot \textbf{1}$. Here $a \leq b$ denotes $a_i \leq b_i$ for all $i$ and $\textbf{1}$ denotes the all ones vector.
\end{definition}

We will use again use the window-respecting alignment framework. However, our algorithm is now recursive, and thus we need to be careful about choosing the windows to operate with in each level of recursion. Let $\ell = O(\log \log N)$ be a parameter and $\tau_0 = N > \tau_1 > \cdots > \tau_\ell = N^{1/\log \log N}$ be a sequence such that for all $i < \ell$, $\tau_m / \tau_{m-1} = \Theta(N^{1/\log \log N})$ and is integer (that is, these ratios are not necessarily the same but are all within a constant factor of $N^{1/\log \log N}$). We will also eventually choose a sequence of error parameters for each level $\epsilon_0, \epsilon_1, \ldots, \epsilon_\ell$. Let $\mathcal{D} = \{1, 2, 4, \ldots, N\}$, and for each $d \in \mathcal{D}$ let $\sigma_m(d):= \max\{\frac{\epsilon_{m+1} d \tau_{m+1}}{\tau_{m}} + \frac{\epsilon_{m+1} D \tau_{m+1}}{|X_1|}, 1\}$ rounded down to the nearest power of 2. For each $i > 0$, for $\tau_m$, each $d \in \mathcal{D}$, the corresponding $\epsilon_m$, $\sigma_m(d)$, and $R_d$ defined as the function rounding down to the nearest multiple of $\sigma_m(d)$, we define windows in each string just as in Section~\ref{sec:k-ed}. In particular, for $X_1$ we have windows $W_{1, m, p}$ that are again just a partition of $X_1$ into substrings of length $\tau_m$, and for $X_2, \ldots, X_k$ we have windows $W_{i, m, d, \Delta, p} = X_i[p\sigma_m(d) + 1\dd R_d(p\sigma_m(d) + \tau_m + \Delta)]$, where the set of possible $\Delta$ is defined by $\epsilon_m$ and $\tau_m$. We will refer to these as the windows at level $m$. Note that we are using the same guess $D$ to define windows at all levels of recursion, even though at lower levels of recursion the center distance between the substrings we consider is likely to be much smaller even if our guess is accurate at the first level.

We note some properties of our recursion that motivate this choice of windows: In the $i$th level of recursion, if our subproblems' input is $X_1', \ldots, X_k'$, then we will have the guarantee that $X_1'$ is a window in $X_1$ of length $\tau_m$ and $X_2', \ldots, X_k'$ are one of the windows in $X_2, \ldots, X_k$ corresponding to $\tau_m$. When we are solving a subproblem involving a length $\tau_m$ substring of $X_1$, we will use the windows defined by $\tau = \tau_{m+1}$. In addition, when we are solving this subproblem, by our requirement that all $\sigma_m(d)$ be a power of 2, we have the following property: the windows defined on the full strings $X_1, \ldots X_k$ for $\tau_{m+1}$ that are contained within $X_1', \ldots, X_k'$, are equivalent to the windows we would define within $X_1', \ldots, X_k'$ if we used the same choice of parameters $\tau_{m+1}, \sigma_{m+1}$. We will refer to this set of windows as the windows at level $m+1$ restricted to $X_1', \ldots, X_k'$. 

To give some intuition behind the choice of $\sigma_m(d)$, which is crucial for our analysis: The term with $d$ is a ``local'' term. It contributes to the approximation error locally, only adding error proportional to our center distance estimate for the current tuple of windows, and also helps us keep the number of entries in the DP table within one call small. The term with $D$ is a ``global'' term. It contributes to the approximation error globally; across all recursive calls, the final approximation error accumulated at the top level due to this term will be something like $\epsilon D$. It also keeps the number of windows across all recursive calls small.

Now, for a fixed level $m$ and the corresponding windows, we can define window respecting alignments of $X_1, \ldots, X_k$ identically to~\cref{def:wra}. If we are considering a window-respecting alignment of substrings $X_1', \ldots, X_k'$ instead of the full strings, we simply restrict to the windows contained within these substrings, and then define window-respecting alignments of $X_1', \ldots, X_k'$ as before using these sets of windows. We define the edit tuples of a window-respecting alignment $f$, $\tup{D}(f)$, to be:

\[ [\conv_{j \in [J]} \tup{\max_i |f(W_{1,j})_i|}(f(W_{1,j})) \conv \{ r(f)\}] \cap \{0, 1, \ldots D\}^k\]

Where $\conv$ is the convolution of sets of vectors, i.e. $\conv_i S_i = \{\sum_i v_i | v_i \in S_i \forall i\}$, and $r(f)$ is the vector whose $i$th entry is $r_i(f) = |X_i| - \sum_j |f(W_{1,j})_i|$, i.e., the number of characters in $X_i$ not in any window. Similarly to~\cref{lemma:ked-window}, we can show window-respecting alignments approximate the best standard alignment. 

\begin{lemma}\label{lemma:kced-window}
Let $d$ be any value in $\mathcal{D}$. Let $X_1', \ldots, X_k'$ be windows in $X_1, \ldots, X_k$ at the same level $m$. Let $\mathcal{F}$ be the set of window-respecting alignments of $X_1', \ldots, X_k'$, using the windows at level $m+1$ parametrized by $d$, restricted to $X_1', \ldots, X_k'$. Then we have that $\cup_{f \in \mathcal{F}} \tup{3d}(f)$ is a $(13 \epsilon_{m+1} k d + 6 \epsilon_{m+1} D \tau_{m} / |X_1|)$-approximation of $\tup{2d}(X_1', X_2', \ldots X_k')$.
\end{lemma}
\begin{proof}
First, we will show that for any $f$ and $v \in \tup{3d}(f)$, $v$ is also an edit tuple of $X_1', X_2', \ldots, X_k'$. Let $J = \tau_m/\tau_{m+1}$. For $v \in \tup{3d}(f(W_{1,j}))$, it can be decomposed as $\sum_{j \in [J]} v_j + r(f)$, where $v_j$ is an edit tuple of $f(W_{1,j})$. By deleting the $r_i(f)$ characters in each $X_i$ that are not in any $W_{i,j}$, we get the string $\bigcirc_j W_{i,j}$ for each $i$, and $\sum_{j \in [J]} v_j$ is clearly a valid edit tuple for these strings. So $v$ is an edit tuple of $X_1', X_2', \ldots, X_k'$. 

It now suffices to show that for any edit tuple $v$ of $X_1', \ldots X_k'$ in $\{0, 1, \ldots, 2d\}^k$, there exists $f$ and  $v'$ in $\tup{3d}(f)$ such that $v' \leq v + (9 \epsilon_{m+1} k d + 6 \epsilon_{m+1} D \tau_{m} / |X_1|) \cdot \textbf{1}$. Fix any such $v$. We partition $X_1'$ into substrings of length $\tau_{m+1}$, $\{X_{1,j}\}_{j \in [J]}$. Let $X^*$ be the string such that $\ED(X_i', X^*) \leq v[i]$ for all $i$. Using the same procedure as in \cref{lemma:lowskew}, we can find a partition of $X^*$ into substrings $\{X^*_j\}$ such that each $X^*_j$ has length at most $2\tau_{m+1}/\epsilon_{m+1}$ and:

\[ \sum_{j \in [J]} \ED(X_{1, j}, X^*_j) \leq (1+3\epsilon_{m+1}) v[1] \leq v[1] + 6\epsilon_{m+1} d. \]

Given this partition, again using the same procedure as in \cref{lemma:lowskew}, we can find disjoint substrings of $X_i$, $X_{i,j}$, for all $i > 1$ such that each $X_{i,j}$ has length at most $2\tau/\epsilon_{m+1}^2$ and

\[ \sum_{j \in [J]} \ED(X_{i, j}, X^*_j) + |X_i| - \sum_{j \in [J]} |X_{i,j}| \leq (1+3\epsilon_{m+1}) v[i] \leq v[i] + 6\epsilon_{m+1} d. \]

Now, let $W_{1,j} = X_{1,j}$ for all $j$. Similarly to \cref{lemma:ked-window}, for each $i$,
if $|X_{i, j}| \leq 2 \epsilon_{m+1} \max_{i'} \ED(X_{i',j}, X^*_j) + 6 \epsilon_{m+1}  D \tau_{m+1} / |X_1|$, let $W_{i,j}$ be the empty window ``starting'' at index $\nd(W_{i, j-1})+1$ (or index $1$ if $j = 1$). Otherwise, let $W_{i,j}$ be the longest window $W_{i, m+1, \Delta, p}$ that is a substring of $X_{i,j}$. Note that $|X_{i,j}|$ and $|W_{i,j}|$ differ by at most $2\max_{i'} \ED(X_{i',j}, X^*_j)$. If $\epsilon_{m+1}$ is a sufficiently small constant, similarly to the proof of Lemma~\ref{lemma:ked-window}, this implies there is a choice of $W_{i,j}$ such that $|X_{i,j}| - |W_{i,j}| \leq \epsilon_{m+1} \max_{i'} \ED(X_{i',j}, X^*_j) + 3(\sigma_m(d)-1) \leq \epsilon_{m+1} \max_{i'} \ED(X_{i',j}, X^*_j) + 6[\frac{\epsilon_{m+1} d \tau_{m+1}}{\tau_{m}} + \frac{\epsilon_{m+1} D \tau_{m+1}}{|X_1|}]$. Note that $\sum_{j \in [J]} \max_i \ED(X_{i,j}, X^*_j) \leq k\norm{v}_\infty \leq 2kd$. This implies $r_i(f) - (|X_i| - \sum_{j \in [J]} |X_{i,j}|)$ is at most $5 \epsilon_{m+1} k d$.
We also have by triangle inequality that:

\[\ED(W_{i,j}, X^*_j) \leq \ED(X_{i,j}, X^*_j) +  \epsilon \max_{i'} \ED(X_{i',j}, X^*_j) + 6 [\frac{\epsilon_{m+1} d \tau_{m+1}}{\tau_{m}} + \frac{\epsilon_{m+1} D \tau_{m+1}}{|X_1|}].\]

Now consider the alignment that chooses $f(W_{1,j}) = (W_{1,j}, W_{2,j}, \ldots, W_{k,j})$. For each $j$, by the above inequalities, one edit tuple for $f(W_{1,j}) = (W_{1,j}, W_{2,j}, \ldots, W_{k,j})$ arising from a window-respecting alignment is element-wise at most:

\begin{align*}
(&\ED(X_{1,j}, X^*_j),\\
&\ED(X_{2,j}, X^*_j) +  \epsilon_{m+1} \max_i \ED(X_{i,j}, X^*_j) + 6  [\frac{\epsilon_{m+1} d \tau_{m+1}}{\tau_{m}} + \frac{\epsilon_{m+1} D \tau_{m+1}}{|X_1|}],\\
&\ldots,\\
&\ED(X_{k,j}, X^*_j) + \epsilon_{m+1} \max_i \ED(X_{i,j}, X^*_j) + 6  [\frac{\epsilon_{m+1} d \tau_{m+1}}{\tau_{m}} + \frac{\epsilon_{m+1} D \tau_{m+1}}{|X_1|}])
\end{align*}

 So summing up these edit tuples, and adding $r(f)$, we get a vector arising from a window-respecting alignment that is at element-wise at most $v + (\epsilon_{m+1} k \norm{v}_\infty + 11 \epsilon_{m+1} k d + 6 \epsilon_{m+1}  D \tau_{m} / |X_1|) \cdot \textbf{1} \leq v + (13 \epsilon_{m+1} k d + 6 \epsilon_{m+1} D \tau_{m} / |X_1|) \cdot \textbf{1}$.
\end{proof}

We are now ready to state our algorithm. Our recursive algorithm for computing a sparse approximation of $\tupd$, denoted \textsc{$k$-CED-Alg}, is defined as follows:

\begin{mdframed}
\textsc{$k$-CED-Alg}($X_1', X_2', \ldots, X_k'$, $d$, $m$):

Let $\mathcal{W}_i$ denote the windows at level $m+1$ parametrized by $d$ restricted to $X_1', X_2', \ldots, X_k'$, and $\start, \nd$ be the functions that take a window and gives its starting/ending index in the corresponding $X_i'$. We solve the following dynamic program:
    \begin{align*}
    c(x_1, x_2, \ldots x_k) =  & \left(\cup_{i > 1} c(x_1, x_2, \ldots, x_i - \sigma_{m+1}, \ldots, x_k) \conv \{(0, 0, \ldots, \sigma_{m+1}, \ldots, 0)\}\right) \cup\\
     &(\cup_{W_1, W_2, \ldots W_k \in \mathcal{W}_{1} \times \mathcal{W}_{2} \times \cdots \times \mathcal{W}_{k}: \forall i, \nd(W_i) = x_i} [c(\start(W_1)-1, \start(W_2)-1, \ldots, \start(W_k)-1)\\
    &\conv \cup_{d' \in \mathcal{D}: d' \leq 2d} k\textsc{-CED-Alg}(W_1, W_2, \ldots, W_k, d', m+1)])
    \end{align*}
    For every $k$-tuple such that $x_1$ is a multiple of $\tau_{m+1}$, $x_2, \ldots, x_k$ are all multiples of $\sigma_{m+1}$, and such that $|x_i - x_1| \leq 3d + \frac{\epsilon_m D \tau_m}{|X_1|}$. The base case for the dynamic program is $c(0, 0, \ldots, 0) = \{(0, 0, \ldots, 0)\}$. 
    
    After computing each entry $c(x_1, x_2, \ldots x_k)$, we remove all elements of $c(x_1, x_2, \ldots x_k)$ not in $\{0, 1, \ldots 3d + \lfloor \frac{\epsilon_m D \tau_m}{|X_1|} \rfloor \}^k$. After getting a set of edit tuples from a call to $k$-CED-Alg, we round each coordinate of each vector up to the nearest multiple of $\sigma_m(d)$ before taking the convolution.
    
    Our final output is $c(|X_1'|, |X_2'|, \ldots, |X_k'|)$, and then return this set of vectors.
\end{mdframed}

Our base case will be when $m = \ell$, and we have that $|X_1'| = N^{1 / \log \log N}$ and all $|X_i'|$ are at most $2N^{1 / \log \log N}/\epsilon_\ell^2$. To handle the base case, we will enumerate all substrings of length at most $2N^{1 / \log \log N}/\epsilon_\ell^2$ of each of $X_1, \ldots, X_k$, and compute their edit tuples using, e.g., the exact algorithm of~\cite{NicolasRivals05}. Our top-level recursive call is to $k$\textsc{-CED-Alg}($X_1, \ldots, X_k, D, 0$). 

To keep the algorithm's description consistent across levels, in addition to assuming $X_1$'s length is a multiple of $\tau_1$, we will assume that $X_2, \ldots, X_k$ are multiples of $\sigma_1(D)$; we can enforce this assumption by padding each of $X_2, \ldots, X_k$ with at most $\sigma_1(D)$ copies of a new dummy character. This cannot decrease the center distance and the total increase in center distance due to this padding is at most $\sigma_1(D)$, which contributes an additive $o(\epsilon_0 D)$ to our approximation factor, only at the top level of recursion. By construction, at lower levels of the recursion each $X_i'$ will have length that is a multiple of $\sigma_m(d)$ for $i > 1$, so by this assumption we no longer need to worry about rounding the indices of the value in DP table we output at any level.

\subsection{Approximation Guarantee}
We first prove the approximation guarantee of \textsc{$k$-CED-Alg}, as it will be necessary for our runtime analysis to specify what choice of $\epsilon_0$ to $\epsilon_\ell$ is needed for the desired approximation guarantee.

\begin{lemma}\label{lemma:k-ced-apx}
Let the sequence $\epsilon_0, \ldots, \epsilon_\ell$ satisfy $\epsilon_0 \leq 1$ and $\epsilon_{m+1} = \frac{1}{16k} \epsilon_{m}$ for all $m$. Then at level $m$ of the recursion, each invocation of \textsc{$k$-CED-Alg}($X_1', X_2', \ldots, X_k'$, $d$, $m$) returns a set of edit tuples that is a 
$(\epsilon _ m d + \frac{\epsilon_m D \tau_m}{|X_1|})$-approximation of $\tup{2d}(X_1', \ldots, X_k')$.
\end{lemma}
\begin{proof}

We proceed by induction. Clearly the guarantee holds for the base case $m = \ell$, since we solve the base cases using exact algorithms.

Inductively, assuming at level $m+1$, any edit tuple generated returned by \textsc{$k$-CED-Alg} is element-wise greater than some edit tuple of the corresponding windows, by an argument similar to the first part of the proof of Lemma~\ref{lemma:kced-window} the same property holds at level $m$. So we just need to show that each edit tuple returned by \textsc{$k$-CED-Alg} is not too large an overestimate of some edit tuple of its input strings.

Take any edit tuple $v$ for any window-respecting alignment $f$. Assume the approximation guarantee holds for calls made at level $m+1$. We show that for the corresponding path through the DP table for $c$, there is a vector close to $v$ in the edit tuples generated by this path. $v$ can be decomposed as $\sum_{j \in [J]} v_j + r(f)$ where $v_j$ is an edit tuple of $f(W_{1,j})$. Let $d_j$ be the smallest value in $\mathcal{D}$ such that $d_j \geq \norm{v_j}_\infty$. By our inductive hypothesis, for each $j$ we get a $(\epsilon_{m+1} d_j + \frac{\epsilon_{m+1} D \tau_{m+1}}{|X_1})$-approximation of the edit tuples of $\texttt{tup}_{d_j}(f(W_{1,j}))$ from the call to $k$-\textsc{CED-Alg}$(f(W_{1,j}),  d_j, m+1)$, which includes a vector $v'_j$ that is element-wise at most $v_j + (\epsilon_{m+1} d_j + \frac{\epsilon_{m+1} D \tau_{m+1}}{|X_1}) \cdot \textbf{1} \leq v_j + (2 \epsilon_{m+1} \norm{v_j}_\infty + \frac{\epsilon_{m+1} D \tau_{m+1}}{|X_1})$ if $\epsilon$ is sufficiently small. In addition, the sum of the vectors contributed by the first case in the recurrence relation for $c$ is $r(f)$. So there is an edit tuple computed by our algorithm for this path that is element-wise less than:

\[\sum_j [v_j +  (2 \epsilon_{m+1} \norm{v_j}_\infty  + \frac{\epsilon_{m+1} D \tau_{m+1}}{|X_1|})\cdot \textbf{1}] + r(f) \leq v + \left (4 \epsilon_{m+1} k d +\frac{\epsilon_{m+1} D \tau_{m}}{|X_1|} \right) \cdot \textbf{1}. \]

After accounting for the approximation error of window-respecting alignments due to  Lemma~\ref{lemma:kced-window} and the rounding step, the additive error is increased to at most $\left (16 \epsilon_{m+1} k d +8 \frac{\epsilon_{m+1} D \tau_{m}}{|X_1|} \right) \cdot \textbf{1} \leq (\epsilon _ m d + \frac{\epsilon_m D \tau_m}{|X_1|}) \cdot \textbf{1}$ as desired.

Finally, note that since we only remove vectors with values larger than $3d+\frac{\epsilon_m D \tau_m}{|X_1|}$ and assume $\epsilon_0$ (and thus all $\epsilon_m$) is at most 1, we do not remove any vector that would be in a $(\epsilon_m d + \frac{\epsilon_m D \tau_m}{|X_1|})$-approximation of $\tup{2d}(X_1', \ldots, X_k')$.
\end{proof}

If we set $\epsilon_0 = \epsilon - o(1)$, then after accounting for the $o(\epsilon_0 D)$ error introduced by padding $X_2$ to $X_k$, the smallest $\ell_\infty$-norm of any vector in the output of $k$\textsc{-CED-Alg}$(X_1, \ldots, X_k, D, 0)$ gives a $(1+\epsilon)$-multiplicative approximation of the center distance as desired.

\subsection{Runtime Analysis}

We now bound the runtime of \textsc{$k$-CED-Alg}, completing the proof of Theorem~\ref{thm:k-center}.
\begin{lemma}
For the choice of $\epsilon_0, \ldots, \epsilon_\ell$ given in~\cref{lemma:k-ced-apx}, we can compute the output of $k$\textsc{-CED-Alg}($X_1, \ldots, X_k, D, 0$) in time $O(n^{k/2} \cdot N^{k/2+o(k)} / \epsilon^{O(k)})$.
\end{lemma}
\begin{proof}
Throughout the analysis, we will use the fact that for all $m$, $1/\epsilon_m \leq \log^{O(\log k)} N / \epsilon$.

We first bound the time spent on base cases. Since each $X_j'$ at the bottom level of recursion has size at most $2N^{1/ \log \log N}/\epsilon_\ell^2 = O(N^{o(1)}/\epsilon^2)$ by construction, we can compute each base case's edit tuples and round them in $N^{o(k)}$ time. There are $O(\log d)$ choices of $d$ and $O(\log_{1+\epsilon_\ell} (2 \tau_\ell / \epsilon_\ell^2)) = O(\log^{O(1)} (N) / \epsilon^2) $ possible sizes for each choice of $d$, so there are $O(\log^{O(k)}(N) / \epsilon^{2k})$ different tuples of possible window sizes to consider at this level. The proof of Lemma~\ref{lemma:lz77-repetition} implies that for any string $X$ generated by an SLP of size $n$, the number of distinct substrings of length $\tau$ is $O(n\tau)$ (in particular, in that proof when $\delta = 1$ we are simply taking every substring into $S$). Combining these facts, we conclude there are $O(n^k N^{o(k)} / \epsilon^{O(k)})$ distinct base cases, and thus by amortizing the work for base cases, the total time spent on base cases is $O(n^k N^{o(k)} / \epsilon^{O(k)})$.

Besides base-cases, the only work our algorithm does is rounding and convolutions. We can perform the recursion in an amortized fashion. That is, we never make multiple calls to $k$\textsc{-CED-Alg} on the same $k$-tuple of strings with the same choice of $d$. Similarly, for each $d$ and each level $m+1$ call to $k$\textsc{-CED-Alg}, we only round that call's output's coordinates to the nearest multiple of $\sigma_m(d)$ once. The time spent rounding a set of vectors is proportional to its size, and the final set of vectors that we round was produced by a convolution that took time at least the size of the set of vectors. For this convolution, with amortization we only need to round its output at most $\log N$ times, once per value of $d$ in $\mathcal{D}$. Thus, the time spent on rounding is bounded by the time spent on convolutions times $O(\log N)$.

We now just need to bound the time spent on convolutions. Fix a level $m$ of the recursion and a choice of $d$ in the input. We will bound the total work across all calls at level $m$ and with $d$ as input; there are $O(\log N)$ levels and $O(\log \log N)$ levels, so our final bound on time spent on convolutions will be within logarithmic factors of the bound for one choice of $m$ and $d$. 

The time spent on convolutions in any call is bounded by a constant factor times the time spent on convolutions in the second case in the recurrence relation, i.e., convolutions involving recursive calls.  We perform these convolutions on tuples in $\{0, 1, \ldots, 3d\}$ whose coordinates are multiples of $\sigma_m(d)$, i.e., have size at most $O((d/\sigma_m(d))^k) = O((\tau_{m}/\epsilon_{m+1}\tau_{m+1})^k) = O(N^{o(k)} / \epsilon^k)$. Using FFT, we can thus perform these convolutions in $O(N^{o(k)}/\epsilon^k)$ time (e.g., we could divide all entries by $\sigma_m(d)$, take the convolution, and then multiply by $\sigma_m(d)$). In each call to $k$\textsc{-CED-Alg}, by the same argument as in Section~\ref{sec:k-ed}, there are $O((\tau_m/\tau_{m-1}) \cdot (d/\sigma_m(d))^{k-1}) = O((\tau_{m}/\epsilon_{m+1}\tau_{m+1})^k) = O(N^{o(k)} / \epsilon^{O(k)})$ entries to compute, and for each entry we need to do $O((\log^2 N / \epsilon)^k)$ convolutions. So the time spent on convolutions per call to $k$\textsc{-CED-Alg} is $N^{o(k)}/\epsilon^{O(k)}$ as well. 

We now just need to bound the number of calls made to $k$\textsc{-CED-Alg}, and our final runtime will be within an $N^{o(k)} / \epsilon^{O(k)}$ factor of this. We will show for each choice of $m$ and $d$, the number of calls made is $O(n^{k/2} \cdot N^{k/2+o(k)} / \epsilon^{O(k)})$, which gives the desired runtime bound. We bound the number of calls at each level in two ways. The first way is again using the fact that for any string $X$ generated by an SLP of size $n$, the number of distinct substrings of length $\tau$ is $O(n\tau)$, and that at each level there are $O(\log^{O(k)}(N) / \epsilon^{2k})$ tuples of possible lengths for the strings in the input, each at most $\tau_m / \epsilon_m^2$. Putting these facts together, there are at most $O(n^k \tau_m^k N^{o(k)} / \epsilon^{O(k)})$ distinct calls to $k$\textsc{-CED-Alg} at level $m$ with parameter $d$. 

The second way is exactly what we did in Section 4 to bound the number of coordinates in the DP table: For every $k$-tuple of windows we call $k$\textsc{-CED-Alg} on at level $m$ with parameter $d$, the window $X_1'$ ends at an index in $X_1$ that is a multiple of $\tau_m$, and the other windows end at indices in $X_2, \ldots, X_k$ that are multiples of $\sigma_m(d)$. Furthermore, these entries are distance at most $O(D)$ from the diagonal. So the total number of possible tuples of ending indices for these windows is $O((N/\tau_m) \cdot (D/\sigma_m(d))^{k-1}) = O(N^k / \tau_m^k \epsilon^{O(k)})$. For each tuple of ending indices, there are $N^{o(k)} / \epsilon^{O(k)}$ possible tuples of windows that end at those indices. So we get a bound of $O(N^{k + o(k)}/ \tau_m^k \epsilon^{O(k)})$ different calls for each choice of $m$ and $d$. The desired bound of $O(n^{k/2} \cdot N^{k/2+o(k)} / \epsilon^{O(k)})$ calls follows by taking the geometric mean of the first and second bound, which is at least the smaller of the two.
\end{proof}

\section{Lower Bounds}
We will start with a summary and overview of the techniques. 
\subsection{Lower Bound Overview}\label{sec:techContLBs}
We will start with the definitions of our hypotheses, then we will describe the results of the lower bound sections. 

\subsubsection*{Hypotheses}
We use two hypotheses from fine-grained complexity to generate our lower bounds. We use the strong exponential time hypothesis (SETH) and the $k$-OV hypothesis. Note that SETH implies $k$-OV~\cite{ryanThesis}.

\begin{definition}
	The $k$-CNF Satisfiability ($k$-SAT) problem takes as input a formula $\phi$ with $m$ clauses and $n$ variables. The formula is in conjunctive normal form (CNF) which requires that the formula be the and of $m$ clauses. Each clause is the or of at most $k$ variables. 
	Return true if $\phi$ has a satisfying assignment and false otherwise. 
\end{definition}

%Next let us define SETH.
\begin{definition}[The strong exponential time hypothesis (SETH)~\cite{cseth}]
For all constants $\epsilon>0$ there is some constant $k$ such that $k$-SAT requires $\omega(2^{n(1-\epsilon)})$ time.
\end{definition}

We can re-frame this as $k$-SAT requiring $2^{n(1-o(1))}$ time, as long as $k$ is an arbitrarily large constant. 
Next we define the $k$-OV problem. 

\begin{definition}[$k$-OV~\cite{ryanThesis}]
Take as input a list, $L$, of $n$ zero one vectors of dimension $d = n^{o(1)}$. Return true if there are $k$ vectors $v_i \in L$ for $i \in [1,k]$ such that for all $j$  $v_1[j]\cdot v_2[j]\cdots v_k[j]=0$. 

The $k$-OV hypothesis states that for constant $k$, $k$-OV requires $n^{k-o(1)}$ time. The $k$-OV hypothesis is implied by SETH. 
\end{definition}

We use the $k$-OV hypothesis to generate our lower bounds. As the $k$-OV hypothesis is implied by SETH,  SETH also implies our lower bounds. 

\subsubsection*{$k$-LCS lower bound} 
Assuming the well-studied Strong Exponential Time Hypothesis (SETH), in \cref{sec:kLCSLB} we show a lower bound for the $k$-LCS problem in the compressed setting.
Intuitively, SETH states that CNF-satisfiability requires $2^{n-o(n)}$ time~\cite{cseth}. Even more specifically, we use the $k$-Orthogonal Vectors problem ($k$-OV)~\cite{virgiSurvey}. 
At a high level, $k$-OV takes as input a list $L$ with $n$ zero-one vectors of dimension $d$. 
We must return YES if there exist $k$ vectors that, when multiplied element-wise, form the all zeros vector.
The $k$-OV conjecture, which is implied by SETH, states that $k$-OV cannot be solved in $O(n^{k-\Omega(1)})$ time.

\begin{reminder}{Theorem~\ref{thm:compkLCSHardFromseth}}
%Let $a$, $b$, and $k$ be positive integers.
%If $(bk+a(k-1))$-OV requires $n^{bk+a(k-1)-o(1)}$ time, then grammar-compressed $k$-LCS requires $\left(M^{k-1}m\right)^{1-o(1)}$ time when $m = \Theta(M^{b/(a+b)\pm o(1)})$ and the alphabet size is $|\Sigma|=\Theta(k)$. Here,
%$M$ denotes the total length of the $k$ input strings and $m$ is their total compressed size.
If the $k'$-OV hypothesis is true for all constants $k'$, then for any constant $\epsilon \in (0,1]$ grammar-compressed $k$-LCS requires $\left(M^{k-1}m\right)^{1-o(1)}$ time when the alphabet size is $|\Sigma|=\Theta(k)$ and $m=M^{\epsilon \pm o(1)}$. Here, $M$ denotes the total length of the $k$ input strings and $m$ is their total compressed size.
\end{reminder}

Our lower bound relies on two primary tools. First, we use a very compressible representation of $a$-OV instances.
Specifically, given a list $L$ of $n$ zero-one vectors of dimension $d$, consider a new list $\Flist(L)_a$ of $n^a$ zero-one vectors of dimension $d$, with every vector in $\Flist(L)_a$ representing the element wise multiplication of $a$ vectors from $L$. Formally, $\Flist(L)_a$ is indexed by $a$-tuples of indices from $[1\dd n]$,
and each vector $\vec{v} = \Flist(L)_{a}[j_1][j_2]\cdots[j_a]$ is defined, for every coordinate $i\in [1\dd d]$, with:
\[\Flist(L)_{a}[j_1][j_2]\cdots[j_a][i] = \vec{v}[i] = L[j_1][i] \cdot L[j_2][i] \cdots L[j_a][i]\]
Notably, $\Flist(L)_{a}$ contains an all zeros vector if and only if $L$ is a YES-instance of the $a$-OV problem. 
%We can represent the vectors in $\Flist_{a}(L)$ in an interleaved way so that they compress very well ~\cite{compressedLCSSETH}. 
%$$\mathbf{String_I}_{a}(L) = \bigcirc_{i=1}^d \left( \bigcirc_{j_1\in[1,n] \ldots, j_a \in [1,n]} L[j_1][i] \cdot L[j_2][i] \cdot \ldots \cdot L[j_a][i] \right).$$
%Specifically, it gives an order on the vectors in the new list and then puts the first bit of all vectors in order, then the second bit of all vectors etc. Notably, this means positions $\Delta, \Delta+n^{a},\ldots, \Delta+(d-1)n^a$ in $\mathbf{String_I}_{a}(L)$ correspond to a single vector in the original problem. Also, $\mathbf{String_I}_{a}(L)$ has a SLP representation of size $\tilde{O}(nd)$.

In the 2-LCS lower bound of~\cite{compressedLCSSETH}, an $(a+2b)$-OV instance $L$ is first transformed into $A=\Flist(L)_{a}$, $B=\Flist(L)_{b}$, and $C=\Flist(L)_{b}$. 
Then, the following strings are defined for every $\vec{v_b} \in B$ and $\vec{v_c} \in C$:
\begin{align*}
    x_{\vec{v_b}}&=\underbrace{\vec{v_{a_1}}[1]\vec{v_{b}}[1]\vec{v_{a_2}}[1]\vec{v_{b}}[1]\cdots \vec{v_{a_{n^a}}}[1]\vec{v_{b}}[1]}_{\text{first bit}}\cdots \underbrace{\vec{v_{a_1}}[d]\vec{v_{b}}[d]\vec{v_{a_2}}[d]\vec{v_{b}}[d]\cdots\vec{v_{a_{n^a}}}[d]\vec{v_{b}}[d]}_{\text{$d$th bit}},\\
    y_{\vec{v_c}}&=\vec{v_{c}}[1]\underbrace{000000}_{n^a-1}\vec{v_{c}}[2]\cdots\underbrace{000000}_{n^a-1}\vec{v_{c}}[d].
\end{align*}
The string $x_{\vec{v_b}}$ that interleaves $\vec{v_b}$ with bits of $n^a$ vectors $\vec{v_{a_i}}\in A$, referred to as ``interleaved'' representation, is highly compressible, to an SLP of size $O(nd)$.
%It can be shown that $x_{\vec{v_b}}$ compresses to $O(nd)$ using SLP compression whereas its original length is $n^ad$. 
Moreover, if there exists a vector $\vec{v_{a_i}} \in A$ such that $(\vec{v_{a_i}},\vec{v_b},\vec{v_c})$ is orthogonal, Abboud et al.~\cite{compressedLCSSETH} show (using the structural alignment gadget of~\cite{alignmentGadget}) how to perfectly align $(\vec{v_{a_i}}[l],\vec{v_b}[l],\vec{v_c}[l])$ for all $l \in [1\dd d]$. Finally, the gadgets $x_{\vec{v_b}}$ for all $\vec{v_b} \in B$ are concatenated with extra padding to generate $X_B$, and the gadgets $y_{\vec{v_c}}$ for all $\vec{v_c} \in C$ are concatenated with extra padding to generate $Y_C$.
This leads to the $(Mm)^{1-o(1)}$ lower bound since the uncompressed and compressed lengths of $X_B$ and $Y_C$ are (roughly) $O(n^{a+b})$ and $O(n^{b})$, respectively, and we are solving an $(a+2b)$-OV instance.

We may extend the above construction to the compressed $k$-LCS setting by transforming an $(a+kb)$-OV instance $L$ into lists $A=\Flist(L)_{a}$, $B=\Flist(L)_{b}$, and $C_h=\Flist(L)_{b}$ for $h\in [1\dd k-1]$.
We then create $X_B$ and $Y_{C_h}$ for $h\in [1\dd k-1]$.
Since the strings $Y_{C_h}$ are zero-padded, we can easily adapt the same structural alignment gadget of $2$-LCS from~\cite{alignmentGadget} to ensure a perfect alignment.
However, this only leads to a lower bound of $(m^{k-1}M)^{1-o(1)}$ since the uncompressed and compressed lengths of the strings remain (roughly) $O(n^{a+b})$ and $O(n^b)$, respectively, and we are solving an $(a+kb)$-OV instance: $ Mm^{k-1} = \Oh(n^{a+kb})$. 
To get a much stronger lower bound of $(mM^{k-1})^{1-o(1)}$, we need to solve a much higher OV instance. 
In particular, we will solve an $(a(k-1)+kb)$-OV instance by taking $A_h=\Flist(L)_{a}$, $B_h=\Flist(L)_{b}$ for $h\in [1\dd k-1]$, and $C=\Flist(L)_{b}$.
We then create strings $X_{B_h}$ from $A_h$ and $B_h$ for each $h\in [1\dd k-1]$, and $Y_C$. That is, we now have $(k-1)$ interleaved strings and only one zero-padded string.
This makes generalizing the structural alignment gadget substantially more intricate since we may have to deal with $k-1$ different offsets.
In fact, without any zero-padded string, we are not able to show any perfect alignment gadget. Because we are now solving an $(a(k-1)+kb)$-OV instance, we get our desired lower bound by noting $M^{k-1}m=\Oh(n^{a(k-1)+kb})$. 

\subsubsection*{Easy $k$-Median Edit Distance lower bounds via LCS reduction}
As a first lower bound for edit distance, we can reduce from LCS to both median $k$-edit distance and center $k$-edit distance. Suppose, we are given a $k$-LCS instance with strings $S_1,\ldots,S_{k}$ all of length $M$ and let $\emptystring$ denote the empty string. It can be shown that
$$\ED(S_1,\ldots,S_{k},\underbrace{\emptystring,\ldots,\emptystring}_{(k-1)}) = Mk - \lcsf(S_1,\ldots,S_{k}).$$
This increases $k$ since we add $(k-1)$ empty strings, but it does not increase the size of the problem, or the compression size. Using the above relation, we can prove the following theorem.
%We get one lower bound from a general $k$-LCS to $k'$-Edit Distance reduction. This reduction adds $(k'-k)$ empty strings to the $k$-LCS instance. This increases $k$, but does not increase the size of the problem. This general purpose reduction may be useful in other settings as well.
%\begin{theorem}\label{thm:medianEDhardLCS}
%We are given a $k$-LCS instance with strings $S_1,\ldots,S_{k}$ all of length $M$. Let the $k$-LCS distance of these strings be $\lcsf(S_1,\ldots,S_{k})$.
%Then, let $\emptystring$ be the empty string we now relate $(2k-1)$-median edit distance to $k$-LCS. Specifically, $(2k-1)$-median edit distance on $S_1$ through $S_k$ and $k-1$ copies of the empty string and the $k$-LCS of $S_1$ through $S_k$:
%$$\ED(S_1,\ldots,S_{k},\emptystring,\ldots,\emptystring) = Mk - \lcsf(S_1,\ldots,S_{k}).$$
%\end{theorem}

\begin{reminder}{Theorem~\ref{thm:EditDistanceSethHard}}
Given an instance of $k$-median edit distance on strings of lengths $M_1 \leq M_2 \leq \cdots \leq M_{k}$ where these strings can all be compressed into a SLP of size $m=|\sum_{i} M_i|^{\delta\pm o(1)}$ for any constant $\delta \in (0,1]$. 
Then, a $k$-median edit distance algorithm that runs in 
$\left((M_2+1) \cdots (M_{k}+1) \cdot m  \right)^{1-\epsilon}$ time for constant $\epsilon>0$ violates SETH.
\end{reminder}

\noindent We can get a similar lower bound for center $k$-edit distance from $k$-LCS by adding a single empty string.

\begin{theorem}
\label{thm:centerFromLCS}
Given an instance of $k$-center edit distance on strings of lengths $M_1 \leq M_2 \leq \cdots \leq M_{k}$ where these strings can all be compressed into a SLP of size $m=|\sum_{i} M_i|^{\delta\pm o(1)}$ for any constant $\delta \in (0,1]$, then, an algorithm for $k$-center edit distance that runs in time 
$\left((M_2+1) \cdots (M_{k}+1) \cdot m  \right)^{1-\epsilon}$ time for constant $\epsilon>0$ violates SETH.
\end{theorem}

%We can get a similar generic lower bound for center $k$-edit distance from $k$-LCS.
%\begin{theorem}\label{thm:centerEDhardLCS}
%We are given a $k$-LCS instance with strings $S_1,\ldots,S_{k}$ all of length $M$. Let the $k$-LCS distance of these strings be $\lcsf(S_1,\ldots,S_{k})$.
%Then, let $\emptystring$ be the empty string we now relate $(k+1)$-center edit distance and $k$-LCS:
%$$\CED(S_1,\ldots,S_{k},\emptystring) \begin{cases}
%=M/2 & \text{  if }\lcsf(S_1,\ldots,S_{k})=M/2\\
%> M/2 &\text{  if }\lcsf(S_1,\ldots,S_{k})\ne M/2\\
%\end{cases}.$$
%\end{theorem}

These reductions are convenient for propagating results from $k$-LCS to $k$-Edit Distance generically. However, because they add empty strings, they don't prove hardness for some of the most commonly studied cases such as where all strings are of the same length and for median $k$-edit distance with even $k$. To get lower bounds for all $k$ and when all strings are of the same length, we use a reduction directly from SETH, instead of going through $k$-LCS.

\subsubsection*{Stronger $k$-Median Edit Distance Lower Bounds directly from SETH}
We get a lower bound for median $k$-edit distance and center $k$-edit distance over compressed strings from SETH. When $k=2$ this resolves the second open problem suggested by Abboud et al~\cite{compressedLCSSETH}. We also generalize the lower bound for all $k\geq 2$. There are many difficulties introduced by trying to get lower bounds for median $k$-edit distance when $k\geq 2$. We can use some of the ideas from the $k$-LCS reduction. Specifically, the notion of the compressed interleaved strings remains. 
%However, because $k$-edit distance allows not just deletions but also insertions and substitutions we need to change our perfect alignment gadget. 
Notably, we need to allow any choice of $\Delta_1,\ldots,\Delta_{k-1}$ offsets; however, if these offsets are more similar we have many characters that match on all but one string. For $k$-LCS we still need to delete these characters, but, in median $k$-edit distance we can simply insert a character in one string. 
%Matches on $k-1$ strings are cheaper than cases where a character matches on only $k/2$ strings. 
This creates an artificial pressure to make all the $\Delta_i$ values the same. To overcome this, we  can use some of the ideas from the recent paper that gives lower bounds for the uncompressed case for $k$-edit distance~\cite{editDistLBSeth}. There is still an issue, they build their alignment gadget with the crucial use of empty `fake gadgets'. However, we need to guarantee that $\Delta_i\in[0,n^a-1]$, and these fake gadgets allow for values of $\Delta_i$ outside of this range. 
To overcome this we incentivize a match up of the real gadgets, which then forces a restriction on valid $\Delta_i$ values. 

Specifically, we need to add a gadget, which we call a selector gadget. This gadget causes characters lined up inside it to have a low edit distance if they all match, and otherwise have a higher edit distance that is unchanged by exactly how well they match. The selector gadget looks like this:
$SCSG_i(c) = \%^{ix} c^y \%^{(k-i)x}.$
We have gadgets $SCSG_1(c_1),\ldots,SCSG_k(c_k)$ such that we can either try to match the characters $c_i$, or try to line up the $\%$ characters. If we line up the $\%$ characters, the edit distance is $ky$. If we line up the $c_i$ characters and they all match ($c_i=c_j~ \forall i,j$ ), the edit distance is $xk^2/4$ if $k$ is even and $x(k^2-1)/4$ if $k$ is odd. If the characters don't match the edit distance is at least $xk^2/4+y$ if $k$ is even and $x(k^2-1)/4+y$ if $k$ is odd. Consider the case of $k$ even, we can choose integer values of $x$ and $y$ such that
$xk^2/4 < yk \leq xk^2/4+y.$
By doing so, if all the characters match, then the median $k$-edit distance is $xk^2/4$, otherwise it is $yk$.
In some sense this gadget is causing characters to act like they do in $k$-LCS, where only a match across all strings gives us a benefit. Using these selector gadgets and ideas from the edit distance and LCS lower bounds, we get a lower bounds for both median $k$-edit distance and center $k$-edit distance from SETH. 

\begin{theorem}\label{thm:editDistanceLowerBoundFromSETH}
If the $k'$-OV hypothesis is true for all constants $k'$, then for all constant $\epsilon \in (0,1]$ grammar-compressed $k$-median edit distance requires $\left(M^{k-1}m\right)^{1-o(1)}$ time when the alphabet size is $|\Sigma|=\Theta(k)$ and $m=M^{\epsilon \pm o(1)}$. Here, $M$ and $m$ denote the total uncompressed and compressed length of the $k$ input strings respectively.
%Let $b$ be a positive integer $b\geq 1$.
%If for positive integers $a$, $b$, and $k$ we have that $(bk+a(k-1))$-OV requires $n^{bk+a(k-1)-o(1)}$ time then SLP compressed $k-$edit distance requires $\left(M^{k-1}m\right)^{1-o(1)}$ time when $m = M^{b/(a+b)}$ and there is an alphabet of size $|\Sigma|=\Theta(k)$. Recall that $M$ is the length of the $k$ input strings to SLP compressed $k-$edit distance  and $m$ is the size of the SLP compression.
\end{theorem}

The lower bound for median $k$-edit distance immediately implies a lower bound for center $k$-edit distance following~\cite{editDistLBSeth}. %However, $k$-center is a harder problem, making this less tight than the median $k$-edit distance result. 

\begin{reminder}{Theorem~\ref{thm:centerEditDistanceAllSameLength} }
We are given $k$ strings of length $M$ with a SLP of size $m$. The center $k$-edit distance problem on these strings requires $\left(M^{k-1}m \right)^{1-o(1)}$ time if SETH is true. 
\end{reminder}

Given these lower bounds for the case of compressed $k$-LCS, median $k$-edit distance and center $k$-edit distance, we want to consider not just compression but also approximation.

\subsection{Lower Bound with LCS}\label{sec:kLCSLB}

In this section we will argue that if we have $k$ strings each of length $M$ and they have a SLP compression of size $m$ then the problem requires $M^{k-1-o(1)}m^{1-o(1)}$ if SETH is true. 
In the next section we use these hardness results for $k$-LCS to prove hardness for $k'$-Edit Distance. 

The core of this section is building a generalized ``perfect alignment'' gadget. This is a gadget that causes substrings to be aligned with no skips or merges. We use this generalized alignment gadget to generalize the work of~\cite{compressedLCSSETH}. The main idea for this perfect alignment gadget is that between every string we want to align, we add symbols $\$_1\$_2\ldots\$_k$. Additionally, at the end of each string $S_i$ in our gadget, we add many copies of these characters,  excepting $\$_i$. That is, we add $\$_1\ldots\$_{i-1}\$_{i+1}\ldots\$_k$. Via this construction, any valid perfect alignment will match all available copies of $\$_i$ for all $i$. Any alignment that isn't perfect (for example it skips matching some sub-string in the middle of $S_i$) will miss out on one of these $\$_i$ characters in $S_i$, thus lowering the value of a potential $k$-LCS.

Recall that $\lcsf(S_1,\ldots,S_k)$ is a function that returns the $k$-LCS of the strings $S_1,\ldots,S_k$. Recall that $\ID(S_1,\ldots,S_k) =  \sum_{i\in[1,k]}(S_i-\flcs(S_1,\ldots,S_k))$. That is, the count of all unmatched characters. 

\subsection{Representations of Many Lists at Once}
\label{subsec:representatoins}

The key idea is going to be different ways to represent many lists of OV instances at once. This representation comes from~\cite{compressedLCSSETH}.

\begin{definition}
\label{def:list}
Let $L$ be the list of vectors to a $k$-OV instance. Let $|L| =n$. 

The list representation of $\ell$ copies of $L$ is made up of $n^\ell$ vectors $\vec{v} = \Flist_{\ell}(L)[j_1][j_2]...[j_\ell]$. 
$$\Flist(L)_{\ell}[j_1][j_2]...[j_\ell][i] = \vec{v}[i] = L[j_1][i] \cdot L[j_2][i]  \cdots  L[j_\ell][i] $$
As a convenience of notation we will allow indexing with a single index into $\Flist(L)$:
$$\Flist(L)_{\ell}\left[ \sum_{i=1}^\ell j_in^{i-1} \right]  = \Flist(L)_{\ell}[j_1][j_2]...[j_\ell]$$
\end{definition}

When writing down this list of vectors into a string there are two ways to do it. The serial way of writing out each vector in order, or the interleaving way. The serial way of writing vectors is in many ways easier to use for gadgets. However, the interleaved version is easier to compress. We will describe both and use both in our gadgets.

\newcommand{\Bstring}{\mathbf{String_B}}
\newcommand{\Istring}{\mathbf{String_I}}
\begin{definition}
Let $L$ be the list of vectors to a $k$-OV instance. Let $|L| =n$.

We define the serial version as:
$$\mathbf{String_B}_{\ell}(L) = \bigcirc_{j_1\in[1,n] \ldots, j_\ell \in [1,n]}  \left(\bigcirc_{i=1}^d  L[j_1][i] \cdot L[j_2][i] \cdots L[j_\ell][i] \right).$$
Note that this is equivalent to
$$\mathbf{String_B}_{\ell}(L) =  \bigcirc_{j \in [1,n^\ell]} \bigcirc_{i \in [1,d]} \Flist(L)[j][i].$$

We define the interleaving version as:
$$\mathbf{String_I}_{\ell}(L) = \bigcirc_{i=1}^d \left( \bigcirc_{j_1\in[1,n] \ldots, j_\ell \in [1,n]} L[j_1][i] \cdot L[j_2][i] \cdots L[j_\ell][i] \right).$$
Note that this is equivalent to
$$\mathbf{String_I}_{\ell}(L) = \bigcirc_{i \in [1,d]} \bigcirc_{j \in [1,n^\ell]} \Flist(L)[j][i].$$

\end{definition}

So the difference between these versions is really just what order we represent the vectors. 
But crucially if there is a particular vector in $\Flist(L)$ that is of interest, this will appear in different places. 
In $\mathbf{String_B}_{\ell}(L)$ a vector $\vec{v}= \Flist(L)[i]$ appears as bits $[i\cdot d, (i+1)\cdot d -1]$. Where as in $\mathbf{String_I}_{\ell}(L)$ the vector  $\vec{v}= \Flist(L)[i]$ appears as bits $i, i+n^k, \ldots, i+(d-1)n^k$. 

We give one final version that merges a single vector with the interleaved representation. 

\begin{definition}
We will expand the previous definition of an interleaved string to allow a merge with a single other vector. 
Recall that
$$\mathbf{String_I}_{\ell}(L) = \bigcirc_{i \in [1,d]} \bigcirc_{j \in [1,n^\ell]} \Flist(L)[j][i].$$
Recall that for a vector $u =\Flist(L)[j]$ it is represented in bits $j, j+n^k, \ldots, j+(d-1)n^k$ in $\mathbf{String_I}_{\ell}(L)$.

We will define 
$$\mathbf{VecS_I}_{\ell}(L,v) = \bigcirc_{i \in [1,d]} \bigcirc_{j \in [1,n^\ell]} \Flist(L)[j][i]v[i].$$
Note that now if we take bits $j, j+n^{\ell}, \ldots, j+(d-1)n^{\ell}$ we give a vector $w$ such that $w[i] = u[i]v[i]$ where $u = \Flist[j]$.
\label{def:vecInterleave}
\end{definition}

\subsection{Intuition for our Reduction}
We will describe at a high level the reduction of~\cite{compressedLCSSETH} and the idea for generalizing it. In this section we will informally explain how to use the serial and interleaved representations of the vectors to build a reduction from SETH to compressed $k$-LCS. We hope to build understanding for what the different levels of alignment gadgets are doing through small examples and intuition. 

\subsubsection{Why We Care About Lining up the Strings}
Lets say we have a representation $\mathbf{String_I}_{\ell}(L)$ and we have a single vector, $v$ of length $d$. We create a new vector $\hat{v}$ where $\hat{v}[i\cdot n^\ell] = v[i]$ and otherwise $\hat{v}$ is zero.  $\hat{v}$ will have length $n^{\ell}(d-1)+1$. 
%This also comes from~\cite{compressedLCSSETH}.

Now we will note the following: the locations of the bits in $\hat{v}$ have exactly the offsets that single vectors do in $\mathbf{String_I}_{\ell}(L)$! So, if we consider sub-string $\mathbf{String_I}_{\ell}(L)[i,i+n^{\ell}(d-1)]$
then $v$ forms an orthogonal $\ell+1$ tuple with the vectors represented by $\Flist(L)_{\ell}[i]$ if $\hat{v}$ is orthogonal to $\mathbf{String_I}_{\ell}(L)[i,i+n^{\ell}(d-1)]$. 
This is why we care about offsets. The next few subsections will simply be building the gadgets necessary to get this ``perfect alignemnt'' and the gadgets needed to represent $k$-OV coordinates in the edit distance setting. 

\subsubsection{The Case of LCS With Two Strings}
How did all of this work in~\cite{compressedLCSSETH}? 
Start with $k$-OV.  Now consider a $k_1$ and $k_2$ that have this property: $k = k_1 + 2k_2$. They then create three sets: $A$ represents $k_1$ vectors at once, $B$ represents $k_2$ vectors at once and $C$ represents $k_2$ vectors at once. We will give a text explanation and then give a small example.

 For $C$ they create its string $S_C$ by taking $\Flist(L)_{k_2}$ and making $Y[i] = \Flist'_{k_2}(L)[i]$. That is they pad the vector with $n^{k_1}-1$ zeros after each entry in the original vector.
 
The $\mathbf{String_I}_{k_1}(L)$ representation of $A$ and the zeros are all very compressible with straight line programs.
For $B$ they create its string $S_B$ by basically merging each vector $b \in \Flist_{k_2}(L)$ with $A$ which is structured like $\mathbf{String_I}_{k_1}(L)$.

So, while the length of each string is $n^{k_2+k_1}$ the compressions are of size $n^{k+1}$. 
We need a gadget that forces our representation to align the two strings with no gaps in the LCS. If we do so, we can then check if an OV exists.
%\bsnote{Can we show the construction of $C$  and $B$?} ++

Now let us work through a small example. Let $k_1=2$ and $k_2=1$.
\begin{align}
v_1 =& <0,1,1,1>\\
v_2 =& <1,1,0,1>\\
v_3 =& <1,0,1,1>\\
v_4 =& <0,1,1,0>\\
L =& \{v_1, v_2, v_3, v_4\} 
\end{align}
For both $B$ and $C$ we form lists that are concatenations of vectors.
$$B=C=v_1, v_2, v_3 , v_4$$
For $A$ we first we want to generate all the vectors $v_{i,j}[p] = v_i[p] \cdot v_j[p]$.
\begin{align}
v_{1,2} =& <0,1,0,1> \text{    }&v_{1,3} =& <0,0,1,1> \text{     }&v_{1,4} =& <0,1,0,0> \\
v_{2,3} =& <1,0,0,1> \text{     }& v_{2,4} =& <0,1,0,0> \\
v_{3,4} =& <0,0,1,0>
\end{align}
Then we form $A$ by taking the first bit of each of these vectors then the second bits, etc. For this example, we put a semicolon in between the first bits and second bits. We do this here for making it easier to read.
$$A = 0,0,0,1,0,0;1,0,1,0,1,0;0,1,0,0,0,1;1,1,0,1,0,0$$
Note that if we take the bits $p,p+6,p+12,p+18$ these correspond to a single vector $v_{i,j}$.
We want the ability to merge $A$ and a single vector $v_i$. For this, if there is $v_i[p]=0$ then we replace all those bits with zeros, otherwise we leave the bits of $A$ as is. 
\begin{align}
(A\& v_1) =& 0,0,0,0,0,0;1,0,1,0,1,0;0,1,0,0,0,1;1,1,0,1,0,0\\
(A \& v_2) =& 0,0,0,1,0,0;1,0,1,0,1,0;0,0,0,0,0,0;1,1,0,1,0,0\\
(A \& v_3) =& 0,0,0,1,0,0;0,0,0,0,0,0;0,1,0,0,0,1;1,1,0,1,0,0\\
(A \& v_4) =& 0,0,0,0,0,0;1,0,1,0,1,0;0,1,0,0,0,1;0,0,0,0,0,0
\end{align}
We also want to generate the padded vectors for $S_C$. These padded vectors have `real' vector values at locations $0,6,12,18$. We want this because it means that if we line up one of these padded vectors, $(0\& v_i)$, against a vector mixed with $A$, $(A \& v_j)$, the `real' values correspond to a vector in $A$.
\begin{align}
(0 \& v_1) =& 0,0,0,0,0,0;1,0,0,0,0,0;1,0,0,0,0,0;1\\
(0 \& v_2) =& 1,0,0,0,0,0;1,0,0,0,0,0;0,0,0,0,0,0;1\\
(0 \& v_3) =& 1,0,0,0,0,0;0,0,0,0,0,0;1,0,0,0,0,0;1\\
(0 \& v_4) =& 0,0,0,0,0,0;1,0,0,0,0,0;1,0,0,0,0,0;0
\end{align}
 Specifically, if we line up $(0\& v_i)$ and $(A \& v_j)$ with an offset of $\Delta$ every lined up set of entries has a zero if there are $k_1+2k_2$ vectors that are orthogonal. Lets consider this example:
 \begin{align}
(A\& v_1) =& 0,0,0,0,0,0;1,0,1,0,1,0;0,1,0,0,0,1;1,1,0,1,0,0\\
(0 \& v_2) =&~~~~~~~~~~~~~~~~~~1,0,0,0,0,0;1,0,0,0,0,0;0,0,0,0,0,0;1\\
\text{This alignment}=&~~~~~~~~~~~~~~~~~~0;0,0,0,0,0,0;0,0,0,0,0,0;0,0,0,0,0,0
\end{align}
By picking this alignment, $\Delta=5$, we are picking the sixth vector that went into $A$ which is $v_{3,4}$. So, this alignment is checking the orthogonality of $v_1,v_2,v_3,v_4$. Lets look at a set of non-orthogonal vectors to compare. The vectors  $v_1,v_2,v_1,v_4$ are not orthogonal. The vector $v_{1,4}$ corresponds to $\Delta =2$.
 \begin{align}
(A\& v_1) =& 0,0,0,0,0,0;1,0,1,0,1,0;0,1,0,0,0,1;1,1,0,1,0,0\\
(0 \& v_2) =&~~~~~~~1,0,0,0,0,0;1,0,0,0,0,0;0,0,0,0,0,0;1\\
\text{This alignment}=&~~~~~~~0,0,0,0,0,0,1,0,0,0,0,0,0,0,0,0,0,0,0
\end{align}

This is why we want `perfect alignments'. We want to build up representations of these strings and allow for any choice of $\Delta$, but no skipping characters or merging gadgets. Previous work generates a perfect alignment gadget such that the output gadgeted strings $T_{(A\&v_i)}$ and $T_{(0\&v_j)}$ have a low LCS if there is a $\Delta$ such that $(A\&v_i)$ and $(0\&v_j)$ have an all zeros alignment with offset $\Delta$ (like  $(A\&v_1)$ and $(0\&v_2)$ with $\Delta=5$ in our example). We will now show what the perfect alignment gadget looks like. We can not use zeros and ones directly to solve OV (see Lemma~\ref{lem:kZerosAndOnes}). We want gadgets such that $0,0$ and $0,1$ have a low LCS and higher LCS for $1,1$. We add the characters $\$$ and $5$. The $5$ characters make sure we don't skip any symbols from $T_{(0\&v_i)}$. The $\$$ characters make sure we don't skip any symbols from $T_{(A\& v_j)}$.
\begin{align}
T_{(A\& v_1)} \approx&  \$05\$05\$05\$05\$05\$05\$15\$05\$15\$05\$15\$05\$05\$15\$05\$05\$05\$15\$15\$15\$05\$15\$05\$0 \$\\
T_{(0 \& v_2)} \approx& \$~~~~\$~~~~\$~~~~\$~~~~\$~~~~\$ 15\$05\$05\$05\$05\$05\$15\$05\$05\$05\$05\$05\$05\$05\$05\$05\$05\$05\$1 \$^{6}
\end{align}
The extra dollar signs at the ends of the strings allow all the dollar signs in $T_{(A\&v_1)}$ can be matched regardless of the offset $\Delta$. The $5$ symbols make skipping zero or one characters also skip at least one $5$ character. This set of characters (at a high level) form the perfect alignment gadget of~\cite{compressedLCSSETH}.

To form the string $S_B$ we want to basically concatenate $T_{(A\&v_1)}, T_{(A\&v_2)},T_{(A\&v_3)},T_{(A\&v_4)}$. 
To form the string $S_C$ we will basically concatenate $T_{(0\&v_1)}, T_{(0\&v_2)},T_{(0\&v_3)},T_{(0\&v_4)}$. These strings are not really concatenated, but instead have an alignment gadget wrapped around them. This alignment gadget guarantees a low LCS if a pair of strings $T_{(A\&v_i)}$ and $T_{(0\&v_j)}$ have a low LCS, and otherwise has a high LCS. 
In total, this means the strings $S_B$ and $S_C$ have low LCS if there exist $i,j,\Delta$ $(A\&v_i)$ and $(0\&v_j)$ have an all zeros alignment with offset $\Delta$. Such a zero alignment existing implies a $(k_1+2k_2)$-OV exists (a $4$-OV in our example).  

\subsubsection{How to Generalize This (Intuition)}
What we want generically for $k$-LCS is to have $k$ sets of lists that act like $B$ and $C$, and $\ell$ sets of lists that act like $A$. If we make an efficient reduction with these parameters, then we get a lower bound of $(M^\ell m^{k-\ell})^{1-o(1)}$.

To get the easy generalization we set $\ell=1$, and we have one ``$A$ type'' set of lists. Lets call the ``B'' and ``C'' type lists $B_1, \ldots, B_k$. 
We create strings $S_1, \ldots, S_k$ and merge $B_1$ and $A$ into $S_1$ using the method from~\cite{compressedLCSSETH}. For $S_i$ where $i>1$ we instead use the padding with zeros method. 
Now we have a situation where we want a gadget that forces the zero padded strings to line up exactly and they are both on some offset of $i$ from the strings in $S_1$. This as it turns out is easy. The strings are of the same length and just copying the construction used for $C$ in~\cite{compressedLCSSETH} will get us what we want here. 
Specifically, the string $S_1$ will be roughly\footnote[2]{Once again, it isn't really a concatenation. Instead these strings are wrapped in an alignment gadget. However, these alignment gadgets are basically concatenations of the strings but with characters in-between the strings. } a concatenation of $T_{(A \& v_i)}$ strings. The $S_i$ strings for $i>1$ are instead roughly\footnotemark[2] the concatenation of  $T_{(0 \& v_i)}$ strings. If we are given a set of $k$ strings $T_{(A \& v_1)},T_{(0 \& v_2)}, \ldots, T_{(0 \& v_k)}$ we want to allow any offset $\Delta$, but that same $\Delta$ should be shared across all the $T_{(0 \& v_i)}$ strings. As a result, we can just use the same $T_{(0 \& v_i)}$ strings from the two string case for all the strings $i>1$. The $T_{(0 \& v_i)}$ strings are the same in every location except for the $d$ representations of the bits in the vector $v_i$. The structure of the $5$ characters forces all of the $T_{(0 \& v_i)}$ strings to line up together to match all the $5$ characters. The $\$$ characters force any high LCS to not skip any of the zeros or ones in the $T_{(A \& v_i)}$ representation of $(A \& v_i)$. On a high level this reduction is easy because we still have only one offset $\Delta$ that we need to deal with.

What needs to happen if $\ell>1$? The primary hurdle is coming up with a setup where two long strings of the $B$ type from the original construction can be forced to have their optimal setup line them up exactly with no skips when they have two different offsets from the zero padded strings. To get a sense of the difficulty consider how many $\$$ characters should be at the start and end of those strings to allow all $\$$ characters to be matched regardless of offsets. As we grow the number on one string we have to grow the number on the other. So we need different symbols $\$_1,\ldots,\$_k$ for each string. 

For convenience let $0$ and $1$ be stand-ins for the strings we use in LCS reductions from $OV$ (there are longer strings that have the property we want where the LCS of $110,101,011,100,010,\ldots,000$ are all equal). Now, if we want to compare $k$ strings where $X \in \{0,1\}^m$ and $Y_1,\ldots, Y_{k-1} \in \{0,1\}^n$ where $n<m$ and we want the $Y$s to line up exactly and we want them to compare to some substring of $X$ then we can add a special character $\$$. Let $Z = \$^{c}$ where $c$ will be a constant in terms of $k$ that is larger than the full length of string representations of $0$ and $1$. 
\begin{align}
S_X =& Z X[0] Z X[1] Z \ldots Z X[m] Z\\
S_{Y_i} =& Z^{m-n} Y_i[0] Z Y_i[1] Z \ldots Z Y_i[m] Z^{m-n}
\end{align}

Now, if there is a sub-string of $X$ that is orthogonal to $Y_1,\ldots, Y_{k-1}$ then the optimal LCS will match all the $Z$s in $X$ and match each character $Y_i[j]$ with some character in $X$. If we spread out our matches of $Y_i$ and don't match to a sub-string of $X$ but instead to a subsequence, we will miss out on some $Z$ characters. So, if there isn't a match that corresponds to an OV we will lose out. And this works with many strings at once. 

This gadget is forcing not just any alignment of the underlying strings in $X_1,\ldots,X_{k-1},Y$, but a \textit{perfect} alignment. We will use this structure to build a perfect alignment gadget.

\subsection{Reduction}

We will prove lemmas building up the gadgets for this construction. We will describe the details of our gadgets and reductions here. The intuition described above of both why we care about perfect alignment and how to achieve it is used in the next subsection on our alignment gadget. 

\subsubsection{Alignment Gadget}
Now we will prove that the alignment gadget works as desired. First let us define what an alignment and perfect alignment are. 
\begin{definition}
\label{def:alignment}
We will generalize the Structured Alignment Cost definition of previous work~\cite{compressedLCSSETH}.
We are given as input $k$ lists of strings $X_1,\ldots,X_{k-1},Y$. Where $|X_i|=n$, $|Y|=m$, and $m<n$.
 
An alignment, $\Lambda$ is a list of $t$ $k$-tuples:
$$((i_{1,1},i_{1,2}, \ldots,i_{1,k}), \ldots,(i_{t,1},i_{t,2}, \ldots,i_{t,k}))$$ 
where $i_{j,p} < i_{j',p}$ if $j<j'$. 
We call an alignment perfect if 
$i_{j,p}+1 = i_{j+1,p}$ and $t=m$. 
\end{definition}

Now we will create some gadgets to maintain alignment. We will define them here and then below prove the various properties we care about. 

\begin{definition}[Perfect Alignment Gadget]
We are given as input $k$ lists of strings $X_1,\ldots, X_{k-1},Y$. Where $|X_i|=n$, $|Y|=m$, and $m<n$. 

We will add $k$ new symbols $\$_1,\ldots,\$_{k}$.  We will define $A = \$_1^{2\ell} \circ \ldots \circ \$_{k}^{2\ell}$ where $\circ$ is the concatenation operator. Let $A_{-i} = \$_1^{2\ell} \circ \ldots \circ \$_{i-1}^{2\ell} \circ \$_{i+1}^{2\ell} \circ \ldots \circ \$_{k}^{2\ell}$. Note that this is just $A$ with all the $\$_i$ symbols removed.
We also add a character $\%$ which we use to pad out our strings to give them more value. Let $B=\%^{2\ell}$ (note that the $\%$ character is serving the purpose of the $5$ character in our earlier example). We define $f=n-m$.
Then the generalized structured alignment gadget would produce strings:
\begin{align}
    AG_i(X_i)= A_{-i}^f \circ A \circ X_i[1] \circ B \circ A \circ X_i[2] \circ B \circ  \ldots \circ A \circ  X_i[n] \circ B \circ A_{-i}^f& \\
    AG_y(Y) = A_{-k}^f \circ A \circ Y[1] \circ B \circ A \circ Y[2] \circ B \circ \ldots \circ A \circ  Y_m \circ B \circ A_{-k}^f & 
\end{align}

\end{definition}

We also want gadgets to be a selector around the alignment gadget.
We add a new character $@$. We will leave $D$ unset for now.
We define our $SAG$ gadgets:
$$SAG_i(X_i) = @^{D-1} AG_i(X_i)$$
$$SAG_y(Y) = AG_y(Y) @^{D-1}.$$

We will now prove that these work as perfect alignment gadgets. This setup requires that we are trying to detect if there is a perfect alignment in which all strings match as much as they can. 

\begin{theorem}
We are given as input $k$ lists of strings $X_1,\ldots, X_{k-1},Y$. Where $|X_i|=n$, $|Y|=m$, and $m<n$. 
Furthermore all strings $S\in X_i$ and $S' \in Y$ have length $|S|=|S'|=\ell$. 

Additionally, given any set of $k$ strings $S_i \in X_i$ and $S_y \in Y$ the LCS distance is either $z$ or $z+1$ for some constant $z$. 

Let $\Lambda$ be a perfect alignment which is a list of $m$ $k$-tuples: $(i_{1,1},i_{1,2}, \ldots,i_{1,k}).$ Where $i_{j,h}+1=i_{j+1,h}$.

Then, if there is a perfect alignment $\Lambda$ in which there are exactly $m$ $k$-tuples such that  $$\flcs(X_1[i_{j,1}],\ldots,X_{k-1}[i_{j,k-1}],X_y[i_{j,k}]) = z+1,$$
then 
$$\flcs(SAG_1(X_1), SAG_2(X_2),\ldots, SAG_{k-1}(X_{k-1}), SAG_y(Y)) = D,$$
if all perfect alignments have less than $m$ $k$-tuples with a $\lcsf$ of $z+1$ then
%\agnote{Less than $m$ $k$-tuples with the above property?} ++ good call 
$$\flcs(SAG_1(X_1), SAG_2(X_2),\ldots, SAG_{k-1}(X_{k-1}), SAG_y(Y)) = D-1.$$
These strings use an additional alphabet of size $O(k)$. The total length of strings is $O(\ell n)$. The value of is $D=2\ell(2m+(k-1)n)+ (z+1)m$.
\label{thm:PerfectAlignmentWorks}
\end{theorem}
\begin{proof}
In the $SAG$ gadgets note that if we match any $@$ character we have a maximum $\flcs$ of $D-1$, and we can always achieve this. If we match any characters from $AG_i$, then we can match no $@$ symbols. Thus, what remains to be proven is that if there is a perfect alignment $\Lambda$ in which there are $m$ $k$-tuples such that if we have $m$ matches where $$\flcs(X_1[i_{j,1}],\ldots,X_{k-1}[i_{j,k-1}],X_y[i_{j,k}]) = z+1,$$
then 
$$\flcs(AG_1(X_1), AG_2(X_2),\ldots, AG_{k-1}(X_{k-1}), SAG_y(Y)) = D,$$
otherwise
$$\flcs(AG_1(X_1), AG_2(X_2),\ldots, AG_{k-1}(X_{k-1}), SAG_y(Y)) \leq D-1.$$

Recall $D=2\ell(2m+(k-1)n)+ (z+1)m$. Now we have three cases to argue. 

\paragraph{Case 1 [There are $m$ ``good'' $k$-tuples, lower bound]:} 
Consider aligning the strings in the perfect alignment $\Lambda$ which has $m$ $k$-tuples of strings which have a $\flcs$ of $z+1$. Now, we can match $m$ copies of $B$. What about $\$_i$ symbols? There are $2\ell n$ copies of the $\$_i$ symbol in $AG_i(X_i)$, they only appear in the copies of $A$ (they don't appear in $A_{-i}$). If we are matching up with a perfect alignment we can match all $2\ell n$ of these symbols. They either line up with copies of $A$ in other strings, or copies of $A_{-j}$. Finally, there are $2\ell m$ copies of $\$_y$  in $AG_y(Y)$. In a perfect alignment these symbols will all get matched to symbols that appear in copies of $A$ in other strings. So, in total 
$$\flcs(AG_1(X_1), AG_2(X_2),\ldots, AG_{k-1}(X_{k-1}), AG_y(Y)) \geq 2\ell m+ 2\ell(m+(k-1)n)+ (z+1)m \geq D.$$

\paragraph{Case 2 [There are $m$ ``good'' $k$-tuples, upper bound]:} Across all $\%$ and $\$_i$ symbols the maximum number of matches is $2\ell (m +m +n(k-1))$. What if we don't match the strings in a perfect alignment manner? There are two ways to do this. One is to skip matching some strings in $Y$ (e.g. merging $Y[j]$ and $Y[j+1]$ and matching that to some single string somewhere else, or simply skipping over a string in $Y$). If this happens we miss out on the characters in at least one $B$. The advantage gleaned for every skipped string in $Y$ is at most $|Y[j]|=\ell$, but skipping out on $B$ is worse, we lose 
%\agnote{loose = lose?} XD yep! 
$2\ell$ matches. The next  case is skipping strings in $X_i$. That is, matching $Y[j]$ with $X_i[j']$ but matching $Y[j]$ with $X_i[j'+1+\Delta]$ for some $\Delta\geq 1$. This looses at least $2\ell$ $\$_i$ characters. Any $k$-tuple of strings has, by assumption in the lemma, a $k$-LCS of at most $z+1$. So, this causes the $k$-LCS less than $D$. Finally, any time we match multiple strings in $X_i$ with a single string in $Y$, $Y[j]$, can increase the match  in $Y[j]$ by at most $\ell$. However, we loose at least $2\ell$ symbols $\$_i$ in $X_i$. This means the $k$-LCS of $AG_i$ strings is at most $D-\ell$. 

\paragraph{Case 3 [There are less than $m$ ``good'' $k$-tuples]:} 
Now, if there are less than $m$ matches with a $k$-LCS of $z+1$ then what is the maximum $k$-LCS? 
Similarly to case 2, if we skip as string in $Y$ or merge $Y[j]$ and $Y[j+1]$ we loose at least one $B$. Because the $B$ copies are in-between every pair of adjacent $Y$ strings. 
If we matching some symbol in each $Y[j]$, then the maximum value if we match one string from each $X_i$ to each string in $Y$ is $D-2\ell$ because we must skip some $2\ell$ characters $\$_i$, so the maximum match would be at most $D-2\ell$. Finally, we could merge multiple strings from $X_i$ and match with a single string from $Y$, in this setting we could potentially get $D-2\ell+\ell$. While we can potentially match all $\ell$ characters in $Y$, we must miss out on at least $2\ell$ $\$_i$ characters.

So, we have proven the result that if we have $m$ matches where $$\flcs(X_1[i_{j,1}],\ldots,X_{k-1}[i_{j,k-1}],X_y[i_{j,k}]) = z+1,$$
then 
$$\flcs(AG_1(X_1), AG_2(X_2),\ldots, AG_{k-1}(X_{k-1}), SAG_y(Y)) = D,$$
otherwise
$$\flcs(AG_1(X_1), AG_2(X_2),\ldots, AG_{k-1}(X_{k-1}), AG_y(Y)) \leq D-1.$$

Then, because of the $@$ symbols if there is a perfect alignment $\Lambda$ in which there are exactly $m$ $k$-tuples such that  $$\flcs(X_1[i_{j,1}],\ldots,X_{k-1}[i_{j,k-1}],X_y[i_{j,k}]) = z+1,$$
then 
$$\flcs(SAG_1(X_1), SAG_2(X_2),\ldots, SAG_{k-1}(X_{k-1}), SAG_y(Y)) = D,$$
if all perfect alignments have less than $m$ $k$-tuples then
$$\flcs(SAG_1(X_1), SAG_2(X_2),\ldots, SAG_{k-1}(X_{k-1}), SAG_y(Y)) = D-1.$$

Proving our desired result.

\end{proof}

Now, we will also want to use this gadget for regular alignment. In this case we will care about distinguishing between a single $k$-tuple with high value versus no strings having high value. 

\begin{theorem}
We are given as input $k$ lists of strings $X_1,\ldots, X_{k-1},Y$. Where $|X_i|=n$, $|Y|=m$, and $m<n$. 
Furthermore all strings $S\in X_i$ and $S' \in Y$ have length $|S|=|S'|=\ell$. 

Additionally, given any set of $k$ strings $S_i \in X_i$ and $S_y \in Y$ the LCS distance is either $z$ or $z+1$ for some constant $z$. 

Finally define $\hat{X}_i$ as a list that is simply two copies of $X_i$. That is $\hat{X}_i[j]=\hat{X}_i[n+j]=X_i[j]$.

In the first case there is exactly one $k$-tuple  where $$\flcs(X_1[i_{1,1}],\ldots,X_{k-1}[i_{1,k-1}],X_y[i_{1,k}]) = z+1,$$ and there exists a perfect alignment that can align this $k$-tuple. That is $n-i_{1,j}\geq m-i_{1,k}$ and $i_{1,j}\leq i_{1,k}$. 
In this first case we want:
$$\flcs(SAG_1(X_1), SAG_2(X_2),\ldots, SAG_{k-1}(X_{k-1}), SAG_y(Y)) = D.$$
In the second case there are zero $k$-tuples that have an LCS of $z+1$ then we want:
$$\flcs(SAG_1(X_1), SAG_2(X_2),\ldots, SAG_{k-1}(X_{k-1}), SAG_y(Y)) = D-1.$$
These strings use an additional alphabet of size $O(k)$. The total length of strings is $O(\ell n)$. The value of is $D=2\ell(2m+(k-1)n)+ zm +1$.

Additionally, let $c$ be the size of an SLP that gives a single variable for all $(k-1)n+m$ strings $X_i[j]$ and $Y[j]$. 
Then there is a SLP representation of all the strings $AG_1(X_1),\ldots, AG_{k-1}(X_{k-1}), AG_y(Y)$ of size $O(c+\lg(\ell)k+\lg(n)+kn+m)$.
\label{thm:BasicAlignment}
\end{theorem}
\begin{proof}
As in the above theorem: in the $SAG$ gadgets note that if we match any $@$ character we have a maximum $\flcs$ of $D-1$, and we can always achieve this. If we match any characters from $AG_i$, then we can match no $@$ symbols. Thus, what remains to be proven is that if there is an alignment with exactly one $k$-tuple such that $$\flcs(X_1[i_{j,1}],\ldots,X_{k-1}[i_{j,k-1}],X_y[i_{j,k}]) = z+1,$$
then 
$$\flcs(AG_1(X_1), AG_2(X_2),\ldots, AG_{k-1}(X_{k-1}), SAG_y(Y)) = D,$$
if there are no $k$-tuples with an LCS of $z+1$ then:
$$\flcs(AG_1(X_1), AG_2(X_2),\ldots, AG_{k-1}(X_{k-1}), SAG_y(Y)) \leq D-1.$$

\paragraph{Case 1 [There is $1$ ``good'' $k$-tuple, lower bound]:} 
Consider aligning the strings where the $k$-tuples of strings which have a $\flcs$ of $z+1$ line up. Now, we can match $m$ copies of $B$. What about $\$_i$ symbols? There are $2\ell n$ copies of the $\$_i$ symbol in $AG_i(X_i)$, they only appear in the copies of $A$ (they don't appear in $A_{-i}$). If we are matching up with a perfect alignment we can match all $2\ell n$ of these symbols. They either line up with copies of $A$ in other strings, or copies of $A_{-j}$. Finally, there are $2\ell m$ copies of $\$_y$  in $AG_y(Y)$. In a perfect alignment these symbols will all get matched to symbols that appear in copies of $A$ in other strings. So, in total 
$$\flcs(AG_1(X_1), AG_2(X_2),\ldots, AG_{k-1}(X_{k-1}), AG_y(Y)) \geq 2\ell m+ 2\ell(m+(k-1)n)+ zm +1 = D.$$

\paragraph{Case 2 [There is $1$ ``good'' $k$-tuple, upper bound]:} 
If there is a singular ``good'' $k$-tuple we can not re-arrange the strings to get a larger alignment. If we skip or merge any strings in $Y$ we loose $2\ell$ symbols from $B$ at least, giving a maximum LCS of $D-2\ell$. If we skip or merge strings in $X_i$ then we loose at least $2 \ell$ symbols from $\$_i$ characters. We gain at most $\ell$ matches, giving a maximum LCS if we merge or skip of $D-\ell$.  Thus, the largest LCS possible is $D$.

\paragraph{Case 3 [There are zero ``good'' $k$-tuples, upper bound]:} In this case, if we follow a perfect alignment we achieve an LCS of $ 2\ell m+ 2\ell(m+(k-1)n)+ zm = D-1$. If we skip a string in $Y$ we miss out on $2\ell$ characters. If we match a single string in $Y$ to multiple strings in $X_i$ we loose at least $2\ell$ characters and match an additional $\ell$ characters at most for a total LCS of at most $D-\ell$. Finally, if we skip over strings in $X_i$, we miss out on $2\ell$ characters $\$_i$, for no benefit. All $k$-tuples have a value of $z$ regardless, so the maximum LCS is $D-2\ell$.

From all these cases we can say that if there is exactly one ``good'' $k$-tuple, and it is reachable in a perfect alignment then 
$$\flcs(SAG_1(X_1), SAG_2(X_2),\ldots, SAG_{k-1}(X_{k-1}), SAG_y(Y)) = D.$$
If there are no ``good'' $k$-tuples 
$$\flcs(SAG_1(X_1), SAG_2(X_2),\ldots, SAG_{k-1}(X_{k-1}), SAG_y(Y)) = D-1.$$

For the compression, $D<\ell n$. So we add an additional $\lg(n)+\lg(\ell).$ So if $c$ is the size of an SLP that gives a single variable for all $(k-1)n+m$ strings $X_i[j]$ and $Y[j]$. 
Then there is a SLP representation of all the strings $AG_1(X_1),\ldots, AG_{k-1}(X_{k-1}), AG_y(Y)$ of size $O(c+\lg(\ell)k+\lg(n)+kn+m)$. Unchanged. 
\end{proof}

\subsubsection{Zero and One Strings}
First we will use the gadgets for representing zeros and ones from~\cite{LCSisHard}.

\begin{lemma}[Zero and One Strings~\cite{LCSisHard}]
There are strings $CG_i(0), CG_i(1)$ such that:
$$k-LCS(CG_1(b_1),\ldots,CG_k(b_k)) = \begin{cases} C &\text{   if } b_1\cdots b_k=0\\
C+1 &\text{   if } \forall b_1\cdots b_k=1 \end{cases}$$
for some positive integer $C$ that is a function of $k$.
Note this corresponds to our desired relationship from $k$-OV. If we have one ``zero string'' then we get a small $k$-LCS, if there are all ``one strings'' then we get a larger $k$-LCS. 
These strings use an alphabet of size $O(1)$. If $k$ is a constant the length of these strings is $O(1)$.
\label{lem:kZerosAndOnes}
\end{lemma}

\subsubsection{Interleave Gadget}
Now we will build the gadget that checks if our interleave representations represent a yes instance of orthogonal vectors.
For this next Lemma recall Definition~\ref{def:vecInterleave}, where we define $\mathbf{VecS_I}_{\ell}(L,v_i)$.

\begin{lemma}
Let $L$ be a list of $n$ $\{0,1\}$ vectors each of dimension $d=n^{o(1)}$. Let $v_1,\ldots, v_k$ be $\{0,1\}$ vectors each of dimension $d$. 
Given the lists, 
We produce strings: $IVG_1(L,\ell,v_1),\ldots, IVG_{k-1}(L,\ell,v_{k-1})$, and $EIVG(v_k)$ such that 
$$\flcs(IVG_1(L,\ell,v_1),\ldots, IVG_{k-1}(L,\ell,v_{k-1}), EIVG(v_k))=C$$
if there are $(k-1)\ell$ vectors in $L$ such that are orthogonal with $v_1, \ldots, v_{k}$.
If there do not exist $(k-1)\ell$ vectors in $L$ that are orthogonal with $v_1, \ldots, v_{k}$ then 
$$\flcs(IVG_1(L,\ell,v_1),\ldots, IVG_{k-1}(L,\ell,v_{k-1}), EIVG(v_k))=C-1.$$

These strings have length at most $O(n^{\ell}d)$ and an alphabet of size $O(k)$. 

Additionally we can compress $x$ strings $IVG_i(L,\ell,v_1),\ldots, IVG_i(L,\ell,v_x)$ or  $EIVG(v_1),\ldots, EIVG(v_x)$ with a total compression size of $O(xd+nd + \ell  \lg(n))$.
\label{lem:interleaveGadget}
\end{lemma}
\begin{proof}
Consider the $k-1$ lists 
$\mathbf{VecS_I}_{\ell}(L,v_i)$ for $i\in [1,k-1]$.
Additionally, let 
$$Vec_E(v_k)[j] = 
\begin{cases}
v_k[h] & \text{ if } j=h\cdot n^\ell\\
0 & \text{ else.}
\end{cases}$$ 
Where the total length is $|Vec_E(v_k)|=n^\ell \cdot (d-1) +1$.

Now create the lists 
$$X_i[j] = CG_i(\mathbf{VecS_I}_{\ell}(L,v_i)[j])$$
$$Y[j] = CG_k(Vec_E(v_k)[j]).$$

Finally create the strings
$$IVG_i(L,\ell,v_1) = SAG_i(X_i)$$
$$EIVG(v_k) = SAG_y(Y).$$

Now, note that by construction the $k$-LCS of coordinate gadgets $CG_i(\cdot)$ is either some value $z$ or $z+1$. 

Now, note that $Vec_E(v_k)$ has zeros in every location except $h\cdot n^\ell$ for $h\in[1,d]$. If we perfectly align $Vec_E(v_k)$  with the $(k-1)$ vectors $\mathbf{VecS_I}_{\ell}(L,v_i)$, then these locations correspond with a vector in each! That is, as mentioned in Definition~\ref{def:vecInterleave}, the bits in locations  $j, j+n^{\ell}, \ldots, j+(d-1)n^{\ell}$ correspond to a vector $w$ where $w[h]=u[h]v_i[h]$ and $u =List[h]$.
If there are $(k-1)\ell$ vectors in $L$ that are orthogonal with $v_1, \ldots, v_{k}$, then there should be a perfect alignment of these vectors such that in every alignment location there is at least one zero. That is, an alignment where every aligned $k$-tuple has a $\flcs$ of $z+1$ as opposed to $z$. 

Let $c=|CG_i(\cdot)|$. Let $z+1$ be the value of $\delta_LCS(CG_1(b_1),\ldots,CG_k(b_k))$ if $b_1\cdots b_k=0$.\\ So 
$\delta_LCS(CG_1(b_1),\ldots,CG_k(b_k))=z$ if $b_1\cdots b_k=0$.

So, by Theorem~\ref{thm:PerfectAlignmentWorks}, if there are $(k-1)\ell$ vectors in $L$ that are orthogonal with $v_1, \ldots, v_{k}$ then 
$$\flcs(IVG_1(L,\ell,v_1),\ldots, EIVG(v_k))=C =  2c(2((d-1)n^{\ell}+1) + (k-1)dn^{\ell}) + (z+1)((d-1)n^{\ell}+1).$$
Otherwise, 
$$\flcs(IVG_1(L,\ell,v_1),\ldots, IVG_{k-1}(L,\ell,v_{k-1}), EIVG(v_k))=C-1.$$

Now we are going to argue that we can compress $x$ strings $IVG_i(L,\ell,v_1),\ldots, IVG_i(L,\ell,v_x)$ or \\ $EIVG(v_1),\ldots, EIVG(v_x)$ with a total compression size of $O(xd+nd)$.
First, note that $A_{-i}^f$ can be represented with an SLP of size $O(\lg(\ell)+\lg(n))$.
Now the rest of our string is a series of entries that look like $A \cdot CG_i(b) \cdot B$. We can create a SLP for both $S_0 = A \cdot CG_i(0) \cdot B$ and $S_1 = A \cdot CG_i(1) \cdot B$ with size $O(\lg(\ell))$. We give the names of $S_0,S_1$ to simplify the notation. 

Next, if we are compressing $x$ different $EIVG(v_i)$ gadgets, first we want to compress $S_0^{n^{\ell}-1}$ which appears repeatedly. We can do this with size $O(\ell \lg(n))$. Finally, we need to add the bits that correspond with each vector. We can do this with size $dx$. This gives a total size of $O(\ell \lg(n)+dx)$. 

Finally, consider the case of compressing $x$ different $IVG(L,\ell,v_j)$ gadgets. This part of the proof, where interleaves are very compressible, borrows very heavily from~\cite{compressedLCSSETH}. For this we note first that $\mathbf{VecS_I}_{\ell}(L,v_j)$ has $d$ sections of size $n^{\ell}$ that correspond to the $d$ dimensions of the vectors. Any section that corresponds to a $h$ where $v_j[h]=0$ has $n^{\ell}$ copies of $S_0$. This can be represented with size $O(\ell \lg(n))$. When $v_j[h]=1$ the section instead corresponds to $\bigcirc_{j\in[1,n^{\ell}]} \Flist[j][h]$. By using Definition~\ref{def:list} we can re-write this as 
$$\textbf{SubList}[L,k] = \bigcirc_{j_1 \in [1,n]}\bigcirc_{j_2 \in [1,n]}\bigcirc_{j_3 \in [1,n]} \ldots \bigcirc_{j_k \in[1,n]} L[j_1][i] \cdot L[j_2][i] \cdots L[j_k][i].$$
Now note that $\textbf{SubList}[L,1]$ is simply a string of $n$ bits and thus has a compression of size $n$. 
Now note that $\textbf{SubList}[L,g]$ is simply a string formed by appending either $\textbf{SubList}[L,g-1]$ or $n^{g-1}$ zeros. We can represent $n^{g-1}$ zeros with $O(g\lg(n))$ variables. So, if $\textbf{SubList}[L,g]$ has an SLP $f(g)$ then   $\textbf{SubList}[L,g+1]$ has an SLP $f(g+1) = n + f(g)+ g\lg(n)$. We have that $f(1)=n$, so $f(k) = kn+ k^2\lg(n) = O(n)$. Thus, the total size of compressing $x$ different $IVG(L,\ell,v_j)$ gadgets is $O(\ell \lg(n)+dx +dn)$.
\end{proof}

\subsubsection{Putting it all Together}

Now that we have an interleave gadget, we want to put many of these gadgets one after each other. However, we want to line up these gadgets. So, we put our interleave gadgets into the perfect alignment gadget. 

\begin{lemma}
\label{lem:OVgadget}
Let $a, b,$ and $k$ be positive constant integers.
Let $L$ be a list of $n$ vectors of length $d=n^{o(1)}$. 

There are $k$ strings $GOV_i(a,b,k,L)$ such that:
$$
\flcs(GOV_1(a,b,k,L),\ldots,GOV_k(a,b,k,L))=
\begin{cases}
D & \text{ if there is a }(bk+a(k-1))\text{-OV in }L \\
D-1 &\text{ else }
\end{cases}
$$
The length of each of these strings is $O(n^{a+b+o(1)})$ with an alphabet of size $O(k)$.
\end{lemma}
\begin{proof}
Unique $(bk+a(k-1))$-OV is equivalent to the normal detection problem of $(bk+a(k-1))$-OV, via a randomized reduction [folklore]\cite{virgiSurvey}. So, we can consider the case where a single $(bk+a(k-1))$-tuple of vectors are orthogonal.

Let $L_b = \Flist(L)_{b}$ as defined in Definition~\ref{def:list}. Now we will use these to define lists of vectors. For all $i\in[1,k-1]$ let
$$X_i[j] = IVG_i(L,a,L_b[j]),$$
$$Y[j] = EIVG(L_b[j]).$$
So all $X_i$ and $Y$ have length $n^b$. The strings inside the gadgets is $n^{a+o(1)}$. All $k$-tuples of these strings have $\flcs$ values of either $C$ or $C-1$. We basically want to wrap an alignment gadget around these strings. However, we want to allow any $k$-tuple to be compared so we will double all the $X_i$ lists:
$$\hat{X}_i[j]=\hat{X}_i[j+n^b]=X_i[j].$$

Now, for any $k$ tuple $j_1,\ldots,j_k$ where $j_i\in[0,n^b-1]$ there is some offset $\Delta_i\in[0,n^b-1]$ for all $\hat{X}_i$ such that $j_k = j_1+\Delta_i \mod n^b$. So, we can align the $\hat{X}_i$ strings with $Y$ and get any $k$-tuple lined up. So, we can now build our gadgets. For all $i\in[1,k-1]$:
\begin{align}
GOV_i(a,b,k,L) &= SAG_i(\hat{X}_i)\\
GOV_k(a,b,k,L) &= SAG_y(Y)
\end{align}

Now, if there is a single $(bk+a(k-1))$-OV in our unique OV instance then there is a single $k-$tuple $X_1[j_1],\ldots, Y[j_k]$ of strings that have a $\flcs$ of $C$. Otherwise, they will all have a $\flcs$ of $C-1$.
By Theorem~\ref{thm:BasicAlignment} we have that in the first case where an $(bk+a(k-1))$-OV exists
$$\flcs(GOV_1(a,b,k,L),\ldots,GOV_k(a,b,k,L))=D$$
otherwise
$$\flcs(GOV_1(a,b,k,L),\ldots,GOV_k(a,b,k,L)) = D-1.$$
Note that having only two possible LCS values hinges crucially on using a unique $(bk+a(k-1))$-OV instance. 

The total size of the alphabet is $O(k)$ for the interleaves and another $O(k)$ for the $SAG$ gadgets, with a total alphabet size of $O(k)$. 

The total length of the $IVG$ gadgets is $n^{a+o(1)}$. We have   $n^{b}$ copies of these gadgets. Giving a total length of $n^{a+b+o(1)}$. 
\end{proof}

We will now bound the size of the compressed length of these gadgets.

\begin{lemma}
\label{lem:CompressionOVGadget}
Let $k$ be a constant integer and let $d=n^{o(1)}$.
Given the $k$ strings $GOV_i(a,b,k,L)$ defined in Lemma~\ref{lem:OVgadget} the size of the compression is $O(n^{b+o(1)}+n^{1+o(1)})$.
\end{lemma}
\begin{proof}
We will start by describing the compression size of $CG_i(0),CG_i(1)$. These strings have length $O(1)$, thus the total size of the compression is at most $O(1)$. There are $2k$ of these strings and $k$ is a constant. So we can have $2k$ variables, one for each string and still the total size of the compression will be $O(1)$.

Next, we need to analyze the size of the compression of the $kn^b$ interleave gadgets $IVG_i(L, \ell, v_i)$. These include the zero and one bit representations, and then are wrapped in a perfect alignment gadget. By Lemma~\ref{lem:interleaveGadget} we have that the total size of the compression of these strings is $O(kd(n+n^{b}))$. Because we have $n^b$ distinct copies of $v_i$ we are generating the $IVG$ strings with, so the $x$ of the lemma is $n^a$ in our case. Note that $kd = n^{o(1)}$. So the total size of these compressions is $O(n^{1+o(1)}+n^{b+o(1)})$

Next, we need to analyze the size of the compression of the entire string. We use the compression of the interleave gadgets and the coordinate gadgets. In addition to this we are wrapping our interleave gadgets in an alignment gadget (from Lemma~\ref{thm:BasicAlignment}). This adds an additional $O(\lg(n)+n^b)$ variables to the SLP.

This gives an SLP of total size of $O(n^{b+o(1)}+n^{1+o(1)})$. 
\end{proof}

Now we will now combine the previous lemmas to give the hardness of $k$-LCS with compression.

\begin{reminder}{Theorem~\ref{thm:compkLCSHardFromseth}}
If the $k'$-OV hypothesis is true for all constants $k'$, then for any constant $\epsilon \in (0,1]$ grammar-compressed $k$-LCS requires $\left(M^{k-1}m\right)^{1-o(1)}$ time when the alphabet size is $|\Sigma|=\Theta(k)$ and $m=M^{\epsilon \pm o(1)}$. Here, $M$ denotes the total length of the $k$ input strings and $m$ is their total compressed size.
\end{reminder}
\begin{proof} 
Use the gadgets from Lemma~\ref{lem:OVgadget}. Call the strings $S_1,\ldots,S_k$ Consider positive integers $a$ and $b$.

These have an alphabet of size $O(k)$ and length $M=O(n^{a+b+o(1)})$ by Lemma~\ref{lem:OVgadget}.
These have a compression of total size $m=O(n^{b+o(1)}+n^{1+o(1)})$ by Lemma~\ref{lem:CompressionOVGadget}.

The size of the alphabet of the reduction is $O(k)$, so if the alphabet is allowed to be size $\Theta(k)$, then this lower bound applies. 

So  $M^{k-1}m = O(n^{(k-1)a+kb+o(1)})$. 
If $(bk+a(k-1))$-OV requires $n^{bk+a(k-1)-o(1)}$ time, then this corresponds to a lower bound of $\left(M^{k-1}m\right)^{1-o(1)}$ for SLP compressed $k-$LCS.

Consider a contradiction to our theorem statement. There would be an algorithm running in $(M^{k-1}m)^{1-\gamma}$ time to solve grammar compressed $k$-LCS when $m=M^{\epsilon \pm o(1)}$ and $\eps \in (0,1]$. In the easiest case we can pick an $a,b$ such that $b/(a+b) = \eps$, in this case we are done. For irrational $\eps$ we need to approximate and then pad the strings. Choose an $a$ and $b$ such that $\eps \leq b/(a+b) < \eps (1+\gamma/2)$. Such $a,b$ exist that are in $O(\frac{1}{\eps\gamma})$. We add a new character $3$.
Let $S'_i = S_i3^x$, where we will set $x\in[M,M^2]$ later. Note that $\lcsf(S'_1,\ldots,S'_k) = \lcsf(S_1,\ldots,S_k)+x$. Note that the compression of these strings is $m' = m+\lg(x) = \Theta(m)$ where as the length is $M'=M+x = \Theta(x)$. Choose $x= m^{1/\eps \pm o(1)} = M^{b/(a+b) \cdot 1/\eps}$.
So now $m' = M'^{1/\eps \pm o(1)}$. Note that $1\leq b/(a+b) \cdot 1/\eps \leq (1+\gamma/2)$. 
Now consider running the claimed fast algorithm on our new $S'_1,\ldots, S'_k$ instance.
The running time is 
$$(M'^{(k-1)}m')^{1-\gamma} = O( (M^{(1+\gamma/2)(k-1)} m)^{1-\gamma-o(1)}).$$
This running time can be simplified to $O(M^{(k-1)} m)^{(1+\gamma/2)(1-\gamma-o(1))}$. Note that $(1+\gamma/2)(1-\gamma-o(1))$ is less than $1-o(1)$. This algorithm violates the lower bound for the original $S_1,\ldots,S_k$ instance. This is a contradiction.

So any algorithm running in $(M^{k-1}m)^{1-\gamma}$ time to solve grammar compressed $k$-LCS when $m=M^{\epsilon \pm o(1)}$ and $\eps \in (0,1]$ violates  $k'$-OV for some $k'$ that depends on $\eps,\gamma$. This implies our theorem statement. 
\end{proof}

\subsection{Easy Edit Distance Lower Bounds from LCS}\label{sec:EditDistLB}
In this section we will prove that $k$-median edit distance is hard from $k'$-LCS. We take a $k'$-LCS instance and add various numbers of empty strings. This pushes the $k$-median edit distance problem towards deletions. So, we increase the number of strings, but we don't increase the total uncompressed or compressed length of the input.

Nicolas and Rivals show NP-hardness for $k$-edit distance through $k'$-LCS for large $k$ and $k'$~\cite{NicolasRivals05}. We take inspiration from their reduction to build our own, removing some of their restrictions, and making it fine-grained efficient. 
We then use the hardness results we have for $k'$-LCS to get lower bounds for $k$-edit distance. 
We will be focusing on a version of edit distance where the strings are allowed to be of very different sizes. We will give an explicit definition now. 

\begin{definition}
Given $k$ strings $S_1,\ldots,S_k$ of lengths $M_1,M_2,\ldots, M_k$ the $k$-edit distance (or $k$-median edit distance) of those strings is the minimum sum across all strings of edits needed to make all strings equal some new string $S'$.
The allowed edits are deleting a character, adding a character and replacing a character (Levenshtein distance). 

More formally: Recall that $\ED(S_i,S')$ denotes the edit distance of $S_i$ and $S'$. Recall that the $k$-median distance is:
$$\ED(S_1,S_2,\ldots, S_k) = \min_{S'\in \text{All Strings}} \left( \sum_{i=[1,k]} \ED(S_i,S') \right).$$
\end{definition}

We can use inspiration from~\cite{NicolasRivals05} to get lower bounds for the center version of this problem as well. Let us remind the definition of $k$-center edit distance problem.   
\begin{definition}
Given $k$ strings $S_1,\ldots,S_k$ of lengths $M_1,M_2,\ldots, M_k$. We define the $k$-center edit distance of those strings is the minimum of the maximum of the distances of those strings to a string $S'$.
The allowed edits are deleting a character, adding a character and replacing a character (Levenshtein distance). 

More formally: Let $\ED(S_i,S')$ be the edit distance of $S_i$ and $S'$. Now the $k$-center distance is:
$$\CED(S_1,S_2,\ldots, S_k) = \min_{S'\in \text{All Strings}} \left ( \max_{i=[1,k]} \ED(S_i,S') \right).$$
\end{definition}

%We will now explain the implications from the $k$-LCS results and the Lemmas of Nicolas and Rivals. 

\subsection{Median Edit Distance}

%We will now show that $k$-median edit distance can be used to compute the $k'$-LCS of strings. To do this we add many empty strings. By adding empty strings we make deletions preferable to other allowed edits. We don't increase the total length of strings by adding empty strings, however, this does increase the number of strings we are computing over. \\

\begin{theorem} 
\label{thm:medianEDhardLCS}
We are given a $k$-LCS instance with strings $S_1,\ldots,S_{k}$ all of length $M$. Let the $k$-LCS of these strings be denoted by $\lcsf(S_1,\ldots,S_{k})$.
The $(2k-1)$-median edit distance on $S_1$ through $S_k$ and $k-1$ copies of the empty string $\emptystring$ is related to the $k$-LCS of $S_1$ through $S_k$:
$$\ED(S_1,\ldots,S_{k},\emptystring,\ldots,\emptystring) = Mk - \lcsf(S_1,\ldots,S_{k}).$$
\end{theorem}
\begin{proof}

First, let us prove that
$\ED(S_1,\ldots,S_{k},\emptystring,\ldots,\emptystring) \leq Mk - \lcsf(S_1,\ldots,S_{k}).$ Let $T$ be the target string of $\lcsf(S_1,\ldots,S_{k})=|T|$. Then, the sum of edit distances to $T$ is $k(n-|T|)+(k-1)|T| = kn-|T|$.

Second, let us prove that 
$\ED(S_1,\ldots,S_{k},\emptystring,\ldots,\emptystring) \geq Mk - \lcsf(S_1,\ldots,S_{k}).$
Let $T'$ be a target string. Now let $d_j,i_j,b_j$ be the number of deletions, insertions and substitutions to go from $S_j$ to $T'$. Let $e_j = d_j+i_j+b_j$ be the edit distance of $S_j$ to $T'$. Now note that $e_j \geq M-|T|+2i_j+b_j$. Additionally, note that the distance from $\emptystring$ to $T'$ is $|T'|$. 
So the total distance is 
$$kM-k|T'|+ \sum_{j=1}^k 2i_j+b_j+(k-1)|T'| = kM-|T'|+\sum_{j=1}^k 2i_j+b_j.$$
So, $\ED(S_1,\ldots,S_{k},\emptystring,\ldots,\emptystring)$ can only be less if $|T'|>|T|$. Note, that $\sum_{j=1}^k 2i_j + b_j \geq  |T'|-|T|$. The target $T$ is the longest string to be achieved with only deletions. Any change from this $T$ (notably added characters) must involve at least one substitution or an insertion. So we can say that the total distance is
$$kM-|T'|+\sum_{j=1}^k 2i_j+b_j \geq kM-|T'|+|T'|-|T| = kM-|T|.$$
So, $\ED(S_1,\ldots,S_{k},\emptystring,\ldots,\emptystring) \geq Mk - \lcsf(S_1,\ldots,S_{k}).$

Thus, $\ED(S_1,\ldots,S_{k},\emptystring,\ldots,\emptystring) = Mk - \lcsf(S_1,\ldots,S_{k}).$
\end{proof}

Now that we have a tight relationship between the edit distance and \LCS, we can use this to get a lower bound from SETH through \LCS.

\begin{theorem}\label{thm:EditDistanceSethHard}
Given an instance of $k$-median edit distance on strings of lengths $M_1 \leq M_2 \leq \cdots \leq M_{k}$ where these strings can all be compressed into a SLP of size $m$. 
Then, an algorithm for $k$-median edit distance that runs in 
$\left((M_2+1) \cdots (M_{2k-1}+1) \cdot m  \right)^{1-\epsilon}$ time for constant $\epsilon>0$ violates SETH.
\end{theorem}
\begin{proof}
We will use Theorem~\ref{thm:compkLCSHardFromseth} and Theorem~\ref{thm:medianEDhardLCS}. 

Say we are given an instance of $k$-LCS with strings $S_1,\ldots, S_{k}$ of length $M$ and a SLP compression of all strings of size $m$. Then, by Theorem~\ref{thm:medianEDhardLCS} we can solve this with an instance of $(2k-1)$-median edit distance on $k$ strings $S_1,\ldots, S_{k},\emptystring,\ldots,\emptystring$. We can compress these $k$ strings with a compression of size $m' =m+O(k)$ (we need only compress the empty string in addition). 

$k$-LCS requires $(M^{k-1}m)^{1-o(1)}$ if SETH is true. Note that for our chosen strings $M^{k-1}= M_2 \cdots M_{k}$. Now note that our compression is of size $m' = O(m)$. 
The reduction takes constant time (simply append the empty string and make a call to $k$-median edit distance). So $k$-median edit distance requires $\left((M_2+1) \cdots (M_{2k-1}+1)  \cdot m  \right)^{1-o(1)}$ time if SETH is true. We can re-state this as an algorithm running time $\left((M_2+1) \cdots (M_{2k-1}+1)  \cdot m  \right)^{1-\epsilon}$ time for constant $\epsilon>0$ violates SETH.
\end{proof}

Next we will use similar ideas to show hardness for center edit distance. 

\subsection{Center Edit Distance}

Nicolas and Rivals present a very simple reduction from a specific version of $(k-1)$-LCS to $k$-Center Edit Distance. This reduction simply adds the empty string as the last string.  The same concept works here. We can distinguish between the case where a  $(k-1)$-LCS is greater than or equal to some constant $c$. Because, if all the strings in the $k$-LCS are of length $M$ adding a single empty string distinguishes between the $(k-1)$-LCS less than $M/2$ or greater than or equal to it. Why? Because, if the $k$-LCS at least $M/2$ then every string is $M/2$ deletions away from the target string \emph{and} the new empty string is as well! 
Otherwise, if the LCS is less than $M/2$, we are more than $M/2$ edits away provably. By adding characters to our $(k-1)-$LCS strings we can artificially increase the match (adding a large number of matching characters to each string), or artificially decrease it (add a large number of not-matching characters). By doing this we can go from our $(k-1)$-LCS being $c$ to our  $(k-1)$-LCS being $M'/2$, for our new length of strings.

\begin{lemma}\label{lem:kLCStohalfLCS}
Assume a $k$-LCS instance over $k$ strings of length exactly $M$.
If deciding whether the $k$-LCS distance is equal to $M/2$ over an alphabet of size $|\Sigma|$ can be done in $T(M)$ time, then we can decide whether the $k$-LCS distance is equal to $C$ over an alphabet of size $|\Sigma|+k+1$ for any constant $C$ in time $O(T(M)+kM)$. 
\end{lemma}
\begin{proof}
Let $S_1,\ldots,S_k$ be an instance of $k$-LCS where we want to decide if the distance is exactly $C$. Let the $k$-LCS be $\lcsf(S_1,\ldots,S_k)$.

For integers $a$ and $b$ let
\[S_i' = @^a S_i \#_i^b.\]
That is, we append $a$ $@$ symbols at the start and  $b$ $\#_i$ symbols at the end of each string. The $\#_i$ strings can not be matched. The $@$ symbols can be trivially matched. 
So we have that $|S'_i| = M' = M+a+b$ and
$$\lcsf(S'_1,\ldots,S'_k) = \lcsf(S_1,\ldots,S_k)+a = C+a.$$
We simply want to choose values of $a$ and $b$ such that $2(C+a) = n+a+b$. This simplifies to $a=M+b-2C$. 
If $2C>M$ then $b=2C-M$ and $a=0$. 
If $2C<M$ then $b=0$ and $a=M-2C$. 

The length of these strings is $M'=2C$ or $M'=2M-2C$, both are less than $2M$.  So, in $O(kM)$ time we can produce new strings of length $M'$ where determining if the $k$-LCS is exactly $M'/2$ determines if the original instance had distance exactly $C$. 
\end{proof}

Now we will show that $k$-center edit distance solves $(k-1)$-LCS.

\begin{theorem}\label{thm:centerEDhardLCS}
We are given a $k$-LCS instance with strings $S_1,\ldots,S_{k}$ all of length $M$ where $k$-LCS of these strings is denoted by $\lcsf(S_1,\ldots,S_{k})$. The $(k+1)$-center edit distance of $S_1,\ldots,S_{k}$ and emptystring $\emptystring$ and $k$-LCS are related as follows.
\[\CED(S_1,\ldots,S_{k},\emptystring) \begin{cases}
=M/2 & \text{  if }\lcsf(S_1,\ldots,S_{k})\geq M/2,\\
> M/2 &\text{  if }\lcsf(S_1,\ldots,S_{k})< M/2.\\
\end{cases}\]
\end{theorem}
\begin{proof}
Consider first, what's the length of a target string for $\CED(S_1,\ldots,S_{k},\emptystring)=M/2$. Call this target central string $T$. 
If $|T|>M/2$ then the distance from $\emptystring$ to $T$ is greater than $M/2$. If $|T|<M/2$ then the strings $S_i$ must have more than $M/2$ deletions, giving a distance greater than $M/2$. So, to hit $M/2$ the target string must have length $M/2$ exactly.

Next note that for the empty string to reach length $M/2$ it must simply have $M/2$ insertions. For any of the $S_i$ strings to go down to $M/2$ they must simply have $M/2$ deletions. 

Note that $\lcsf(S_1,\ldots,S_{k})$ is defined as $M$ minus the number of deletions needed in each string to reach the minimal target. Thus, with this addition of an empty string  
$$\CED(S_1,\ldots,S_{k},\emptystring) \begin{cases}
=M/2 & \text{  if }\lcsf(S_1,\ldots,S_{k})\geq M/2\\
> M/2 &\text{  if }\lcsf(S_1,\ldots,S_{k})<M/2\\
\end{cases}.$$
\end{proof}

Now we will apply the above Lemma~\ref{lem:kLCStohalfLCS} and Theorem~\ref{thm:centerEDhardLCS} to get a $k$-center edit distance lower bound from SETH.

\begin{reminder}{Theorem~\ref{thm:centerFromLCS}}
Given an instance of $k$-center edit distance on strings of lengths $M_1 \leq M_2 \leq \cdots \leq M_{k}$ where these strings can all be compressed into a SLP of size $m$, then, an algorithm for $k$-center edit distance that runs in time 
$\left((M_2+1) \cdots (M_{k}+1) \cdot m  \right)^{1-\epsilon}$ time for constant $\epsilon>0$ violates SETH.
\end{reminder}
\begin{proof}

We will use Theorem~\ref{thm:compkLCSHardFromseth}, Lemma~\ref{lem:kLCStohalfLCS} and Theorem~\ref{thm:centerEDhardLCS}. 

Say we are given an instance of $k$-LCS with strings $S_1,\ldots, S_{k}$ of length $M$ and an SLP compression of all strings of size $m$. Determining if the $k-$LCS is some integer $C$ requires $(M^{k-1}m)^{1-o(1)}$ time if SETH is true by Theorem~\ref{thm:compkLCSHardFromseth}.

Then  Lemma~\ref{lem:kLCStohalfLCS} simply appends at most $M$ symbols $@$ or $\#_i$ to each string making a new problem $S'_1,\ldots, S'_{k}$ of length $M'$. Note that the size of the compression is now $m' = m+O(k\lg(M))$. Now determining if the $k-$LCS is $M'/2$ requires $\left((M')^{k-1}m'\right)^{1-o(1)}$ if SETH is true. 

Now we will apply Theorem~\ref{thm:centerEDhardLCS}. We can produce an instance of $(k+1)$-center edit distance that has strings  $S'_1,\ldots, S'_{k}, \emptystring$ that distinguishes between the $k$-LCS of $S'_1,\ldots, S'_{k}$ being $M'/2$ or not. 
Now note that $M_i = |S'_i|$ and $M_{k+1}=0$. So $(M_2+1) \cdots (M_{k+1}+1) = \Theta((M')^{k-1})$. The compression of this empty string means that the new compression has size $m'' = m' +O(1) = m+O(k\lg(M))$. 

We can run this a second time where we add two characters to each string: $S'_i = S_i \%_i\%_i$. These characters are unmatchable. Also, if the LCS used to be at least $M'/2+1$ it will still be at least half the length of the strings. So, we can distinguish the exact value. Similarly, the compression of these strings is of size at most $m''+O(k) = m''+O(1)$. 

Hence, an algorithm for $(k+1)$-median edit distance that runs in 
$\left((M_2+1) \cdots (M_{k+1}+1) \cdot m  \right)^{1-\epsilon}$ time for constant $\epsilon>0$ violates SETH.
\end{proof}

\subsection{Edit Distance Lower Bounds from SETH}\label{sec:SETHEditDistLB}
\newcommand{\EDOV}{EDOV}

In this section we show a better lower bound for $k$-edit-distance by reducing from SETH directly. A recent paper has given $M^{k-o(1)}$ lower bounds for Edit Distance from SETH where $M$ is the length of the strings~\cite{editDistLBSeth}. In this section we show a $M^{k-1-o(1)}m$ lower bound for compressed $k$-edit-distance where $m$ is the size of the SLP describing the strings. Our reduction uses the ideas from the SETH lower bound for $k$-edit-distance to achieve this. We will use the same ideas and list structures that we used in the $k$-LCS lower bound. We use many of the same notions of gadgets, however, to distinguish between them, we add $ED$ to the end of the name of the gadgets (for Edit Distance). Note that the structure of this proof mirror almost exactly the $k$-LCS lower bound. However, due to the different distance measures we need to generate different gadgets.

The main takeaway of this section is that in order to build an interleave gadget for edit distance we need to generate a selector gadget that has one value if all values match, and another if not all values match. 

The primary difficulty in generalizing this lower bound comes from the variable costs of partial matches. That is, if we have the edit distance of $\ED(a,a,a,a,b)=1$, where as $\ED(a,a,b,b,c)=3$. By contrast, the LCS of both is $\flcs(a,a,a,a,b)=\flcs(a,a,b,b,c)=0$. So, the overall structure needs to account for this in some way. We want to re-create a perfect alignment gadget, but for Edit-Distance. This will give us two results. First we generalize the 2-LCS lower bound into a 2-edit distance lower bound, answering an explicit open problem given by~\cite{compressedLCSSETH}.

We will use the pre-existing coordinate gadgets and alignment gadgets from~\cite{editDistLBSeth}. So, we have two primary tasks. We need to generate and prove the correctness of \emph{perfect} alignment gadgets. Additionally, we need to analyze the size of the compression of both our gadgets and the~\cite{editDistLBSeth} gadgets. 

\subsection{Selection Gadgets}

We want an additional gadget. A selector gadget. These allow us to say either strings $A_1,\ldots,A_k$ are compared or strings $B_1,\ldots, B_k$ are compared but not both. 
We will use a version that works for single characters. 

\begin{lemma}
There exist single character selection gadgets $SCSG_i(\cdot)$ such that the length is polynomial in $k$ and they add a single character to the alphabet. 
The $k$-edit distance of $k$ $SCSG_i(c_i)$ strings is either some constant $Q$ if all characters match or $Q+v$ (where $v$ is a positive integer) if one character does not match. 
\label{lem:singleCharSelector}
\end{lemma}
\begin{proof}
First let us define the gadget in terms of two free variables we will set later, $a$ and $b$:
$$SCSG_i(c) = \#^{ib}c^a\#^{(k-1-i)b} $$

Now note that if we match the characters $c$ together we must fail to match many $\#$ characters. Specifically these induce an edit distance of:
$$\begin{cases}
bk^2/4 &\text{ if } k \text{ is even}\\
b(k^2-1)/4 &\text{ if } k \text{ is odd}
\end{cases}$$

Now note that if we instead match the $\#$ characters then we have an edit distance of $ak$, as we have to delete the characters input to the gadget. 

Also note that if we match the characters $c$ and one or more symbols are off the edit distance will be \textbf{at least}:
$$\begin{cases}
bk^2/4+a &\text{ if } k \text{ is even}\\
b(k^2-1)/4+a &\text{ if } k \text{ is odd}
\end{cases}$$

So if we can choose an $a$ and $b$ such that:
$$\begin{cases}
 bk^2/4< ak \leq bk^2/4+a &\text{ if } k \text{ is even}\\
b(k^2-1)/4< ak \leq b(k^2-1)/4+a &\text{ if } k \text{ is odd}
\end{cases}$$
then, if all characters match we get an edit distance of $bk^2/4$, otherwise, we get an edit distance of $ak$. 
\end{proof}

Next we will note the existence of coordinate gadgets from previous work. Then we will combine the coordinate gadgets with our selector gadgets to make interleave gadgets. 

\subsection{Coordinate Gadgets}
\newcommand{\CGED}{CGED}
We will use the coordinate gadgets from the~\cite{editDistLBSeth} in our reduction.
\begin{lemma}[Coordinate Gadget Lemma From~\cite{editDistLBSeth}]
Let $b_1,\ldots,b_k$ be in $\{0,1\}$.
Let $C^{-} = 2(k-1)^2$ and let $C^{+}=C^-+k-1 = (2k-1)(k-1)$. Then,
$$\ED(\CGED_1(b_1), \ldots, \CGED_k(b_k)) = \begin{cases}
C^- & \text{ if } b_1\cdots b_k=0\\
C^+ & \text{ otherwise}
\end{cases}$$
\end{lemma}

We will use these inside our interleave vector gadgets. 

\subsection{Interleave Vector Gadget}
\newcommand{\IED}{IED'}
\newcommand{\EIED}{EED'}
\newcommand{\IEDF}{IED}
\newcommand{\EIEDF}{EED}
We are given a $(a(k-1)+bk)$-OV with a list of $n$ vectors each of length $d$. We want to take every vector $v_j = \Flist(L)_{b}[j]$ for $j\in[0,n^{b}]$ and combine them with the interleave representation of $a$ lists. Recall that in Definition~\ref{def:vecInterleave} we define $\mathbf{VecS_I}_{a}(L,v_j)$ as the explicit distribution of the interleave representation of $a$ lists mixed with a single vector. So, we want to have $k-1$ strings that hold representations of $\mathbf{VecS_I}_{a}(L,v_j)$  for all $j\in[0,n^{b}]$. Finally, we need one string that is full of representations of vectors $v_j$ for all $j\in[0,n^{b}]$, padded with many zeros. If we do this and we can force a perfect alignment of these gadgets. We will use an altered version of the sliding pyramids from previous work~\cite{editDistLBSeth} (see Figure~\ref{fig:slidingPyramid}).

\begin{lemma}
Treat $k,\ell$ as  constants. 
We are given as input a list $L$ of $n$ vectors each of length $d$.  Where $d=n^{o(1)}$.
 Let $v_1,\ldots, v_k$ be $\{0,1\}$ vectors each of dimension $d$. Then there are gadgets $\IED_i(L,v_i)$ and $\EIED(v_k)$ such that:
$$\ED(\IED_1(L,\ell,v_1),\ldots, \IED_{k-1}(L,\ell,v_{k-1}), \EIED(v_k))=C$$
if there are $(k-1)\ell$ vectors in $L$ such that are orthogonal with $v_1, \ldots, v_{k}$.
If there do not exist $(k-1)\ell$ vectors in $L$ that are orthogonal with $v_1, \ldots, v_{k}$ then 
$$\ED(\IED_1(L,\ell,v_1),\ldots, \IED_{k-1}(L,\ell,v_{k-1}), \EIED(v_k))\geq C+1.$$

Additionally we can compress $x$ strings 
$\IED_i(L,\ell,v_1),\ldots, \IED_i(L,\ell,v_x)$\\ or  
$\EIED(v_1),\ldots, \EIED(v_x)$
with a total compression size of $O(xd+nd + \lg(n))$.
\label{lem:edinterleaveGadgethelper}
\end{lemma}
\begin{proof}
As in the $k$-LCS Lemma~\ref{lem:interleaveGadget} consider the $k-1$ lists 
$\mathbf{VecS_I}_{\ell}(L,v_i)$ for $i\in [1,k-1]$.
Additionally, let 
$$Vec_E(v_k)[j] = 
\begin{cases}
v_k[h] & \text{ if } j=h\cdot n^\ell\\
0 & \text{ else.}
\end{cases}$$ 
We will build our gadgets \IED~and \EIED~from these lists. Every entry in these lists is either a zero or a one. We want to force a prefect alignment and check orthogonality of the perfect alignment. That is, we want to hit one value if there exists a $\Delta_1,\ldots,\Delta_{k-1}$ such that 
$$\sum_{j=0}^{n^{\ell}(d-1) +1} Vec_E(v_k)[j]\cdot \mathbf{VecS_I}_{\ell}(L,v_1)[j+\Delta_1]\cdots \mathbf{VecS_I}_{\ell}(L,v_{k-1})[j+\Delta_{k-1}] =0.$$

To do this we will use the very convenient coordinate gadgets, but alter them with a selector. We want the selector gadget to force an alignment of the true correct values. We add a new character $2$, this character is just there to be matched in these gadgets. We add another new character $3$, which encourages lining up coordinate gadgets. We will set $x = 100|\CGED_i(b_i)|$. We want to have enough copies of the $SCSG$ gadgets that lining up real gadgets with each-other is optimal. We set $y=100x|SCSG_i(2)|$, we want enough copies of $3$ to force coordinate alignments to be optimal. Finally, our updated coordinate gadgets are below 
$$\CGED'_i(b_i) = 3^{y} \circ (SCSG_i(2))^x \circ \CGED_i(b_i).$$

Next we need to generate ``fake'' coordinate gadgets to fill out space, so that any offset will be valid. We add the characters $\%_i$ for $i\in [1,k]$. The character $\%_i$ will only appear in the $i^{th}$ string. This will guarantee these characters are unmatched. A fake gadget will have a selector gadget wrapped around one of these unmatchable characters and a coordinate gadget of a zero:
$$F_i = 3^{y} \circ (SCSG_i(\%_i))^x \circ \CGED_i(0).$$

Let $f=|\mathbf{VecS_I}_{\ell}(L,v_i)| -|Vec_E(v_k)|=n^{\ell}-1$. 
Now, we will define the three parts of the strings. The section of real gadgets, the section of fake gadgets, and the section of unmatchable characters. We add $k$ new characters $\#_i$, with the intention of making them unmatchable. Let $z= |CGED'_i(b_i)| = |F_i|$. See Figure~\ref{fig:slidingPyramid} for a visual depiction.
\begin{align*}
    REAL_i &= \bigcirc_{j\in[1,d n^\ell]} \left(\CGED'_i(\mathbf{VecS_I}_{\ell}(L,v_i)[j]) \right)\text{~~~~~~~when } i\in[1,k-1]\\
    REAL_k &= \bigcirc_{j\in[1,n^\ell(d-1)+1]} \left(\CGED'_k(Vec_E(v_k)[j]) \right)\\
    FAKE_i &= F_i^{f(k+1-i)} \\
    UNMA_i &= \#_i^{(k+i+1)zf}.
\end{align*}

Now that we have defined these useful parts we can define the overall gadgets:
\begin{align*}
    \IED_i(L,\ell,v_i) &= UNMA_i \circ FAKE_i \circ REAL_i \circ FAKE_i \circ UNMA_i \\
    \EIED(v_k) &= UNMA_k \circ FAKE_k \circ REAL_k \circ FAKE_k \circ UNMA_k
\end{align*}

Now let us consider what happens if there are $(k-1)\ell$ vectors that are orthogonal to $v_1,\ldots,v_k$ then we want to compute the edit distance. Note that all characters in $UNMA_i$ will either be deleted or substituted which has an edit distance of one per character. This gives a total edit distance cost of $zfk(k+5)/2$.
Let $p_i = \ED(F_1,\ldots, F_i)$. Now, on any valid alignment we have $2f$ fake gadgets completely unmatched, $2f$ fake gadgets matched with one other fake gadget, $2f$ fake gadgets matched with two other fake gadgets, etc. Until we get to $(k-1)$ fake gadgets matched together at which point we have $3f$ of these. When the fake gadgets are matched with each other they are also ``matched'' with the unmatchable characters. Those characters will simply substitute/delete to equal the output string. We have enough unmatchable characters that their length is longer than the overhanging fake characters. So we can be assured no insertions will need to happen. So the edit distance contribution of these is
$$f p_{k-1}+\sum_{i=1}^{k-1}2f p_i.$$
Now, we have $2f$ fake gadgets which line up with a mix of fake and real gadgets. Each of these $k$ tuples of lined up gadgets have a contribution of  $x(Q+v)$ from their $SCSG$ gadgets, and the $\CGED$ gadgets contribute $C^-$ edit distance (the $3$ symbols match perfectly and thus have no contribution to the edit distance). So these give a contribution of $2f(x(Q+v)+C^-)$. 
Finally, we have $M= (d-1)n^{\ell}+1$ real gadgets which line up with other real gadgets and these contribute $xQ +C^-$  edit distance each. So our total edit distance is
$$C =zfk(k+5)/2 +f p_{k-1}+\sum_{i=1}^{k-1}2f p_i+2f(x(Q+v)+C^-)+M(xQ+C^-).$$

What happens if we don't have a set of $(k-1)\ell$ vectors which are orthogonal to $v_1,\ldots,v_k$? If we similarly line up gadgets in a valid way, as above, then we have at least one of the real $k$-tuples of $\CGED'$ gadgets where the internal $\CGED$ gadgets contribute $C^+$, increasing the above cost by $C^+-C^-$. What if we instead don't do a valid alignment? If we skip aligning a coordinate gadget we skip some $3^y$ symbols, these then cost an additional $y$ in the edit distance. If we instead align some of the $M$ real gadgets of string $k$ to fake gadgets we miss out on matches of the $SCSG$ gadgets, costing $xv$ in the edit distance. So if there is no set of $(k-1)\ell$ vectors which are orthogonal to $v_1,\ldots,v_k$ then the edit distance is higher than $C$. 

First, we can generate the SLP for $FAKE_i$ and $UNMA_i$. The size of $F_i$ and $\%_i$ are both $O(1)$ (assuming $k$ is a constant). So, we simply need to handle many repetitions. This requires an SLP of size $\lg(f(k+1-i))$ and one of size $\lg((i+1)zf)$ for each $i\in[1,k]$. Luckily for us, in total this SLP will have size $O(\ell \lg(n))$. We additionally need to represent the various values for $REAL_i$. First note that $CGED'_i$ gadgets can be represented with size $O(1)$ SLPs (when $k$ is constant). For $REAL_k$ this is easy from this point on. There are long strings of zero gadgets with only $d$ instances of non zeros. The total representation is $O(d+\ell\lg(n))$. So we just need $REAL_i$ for $i\in[1,k-1]$. Now note that we can use the same structure we used in the $k$-LCS SLP for these interleave gadgets. We can build the structures for different values of $\ell$. 

For convenience let $REAL_i^\ell$ be the real gadget for $\IED_i(L,\ell,v_i)$. Now note that $REAL_i^1$ has an SLP of size $O(dn)$ trivially, it only has length $O(dn)$ in the first place. Now consider separating out the parts that correspond to each of the $d$ dimensions of the vector. Next note that we can form these $d$ parts of $REAL_i^{j+1}$ by concatenating $n$ instances of the parts of either $REAL_i^{j}$ or $n^{j}$ zero coordinate gadgets. So, given an SLP for $REAL_i^{j}$ we can create an SLP for $REAL_i^{j+1}$ with an additional size of $n + j\lg(n)$. This gives a total size of $O(\ell^2\lg(n)+\ell dn)$, as $\ell$ is a constant we have have that the size is $O(dn)$.

So the total size of the SLP will be $O(dn+dx+lg(n))$.
\end{proof}

Notice that we can get away without having the same type of $\$_i$ interleaved symbols that we used in $k$-LCS. This is due to the cost of edits varying even when only some subset of the $k$ strings match on a symbol. We can guarantee we don't skip characters because it will cost us in edits, even if those characters could only be matched up to one other string. However, we aren't quite done. We want to wrap this so that the value is either a match or one higher than a match. We don't want to have the final interleave gadgets give variable outputs depending on how orthogonal vectors are. We want the same value no matter what. 

\begin{lemma}
We are given as input a list $L$ of $n$ vectors each of length $d$, where $d=n^{o(1)}$. Let $\ell$ be a constant. 
 Let $v_1,\ldots, v_k$ be $\{0,1\}$ vectors each of dimension $d$. Then there are gadgets $\IEDF_i(L,v_i)$ and $\EIEDF(v_k)$ such that for some constants $D$ and $w$:
$$\ED(\IEDF_1(L,\ell,v_1),\ldots, \IEDF_{k-1}(L,\ell,v_{k-1}), \EIEDF(v_k))=D$$
if there are $(k-1)\ell$ vectors in $L$ such that are orthogonal with $v_1, \ldots, v_{k}$.
If there do not exist $(k-1)\ell$ vectors in $L$ that are orthogonal with $v_1, \ldots, v_{k}$ then 
$$\ED(\IEDF_1(L,\ell,v_1),\ldots, \IEDF_{k-1}(L,\ell,v_{k-1}), \EIEDF(v_k))= D+w.$$

Additionally we can compress $x$ strings 
$\IEDF_i(L,\ell,v_1),\ldots, \IEDF_i(L,\ell,v_x)$\\ or  
$\EIEDF(v_1),\ldots, \EIEDF(v_x)$
with a total compression size of $O(xd+nd+ x \lg(n))$.

The strings $\IEDF$ and $\EIEDF$ have length $n^{\ell+o(1)}$.
\label{lem:edinterleaveGadget}
\end{lemma}
\begin{proof}
Let $u$ be an all zero vector of length $d$. 
Let $v'_i$ be the vector $v_i$ but with an added last index $v'_i[d+1]=1$ if $i\in[1,k]$. 
Let $u'$ be the all zeros vector except for an added last index $u[d+1]=0$. 
Let $v^*_k$ be the vector $v_k$ but an added last index $v^*_k[d+1]=0$.
Now we add additional characters $5$ and $4$. We add copies where $p = 100|\IED_i(L,\ell,v'_i)|$ and $q = 10p$. Now we generate the following:
\begin{align*}
\IEDF_i(L,\ell,v_i) &= 5^q 4^p \IED_i(L,\ell,v'_i) 4^p 5^q\\
\EIEDF(v_k) &= 5^q \EIED(v^*_k) 4^p \EIED(u') 5^q.
\end{align*}
Here we match up the $5$ characters. Finally, we must match the $4^q$ symbols. There will be $q$ unmatched $4$ symbols. Finally, we will have $|\EIED(v^*_k)|=|\EIED(u')|$ unmatched symbols no matter what. Now, how much comes from matching the $\IED$ and $\EIED$ gadgets? 
If there are $(k-1)\ell$ vectors are orthogonal to $v_1,\ldots,v_k$ then the cost is $C$. If there aren't then the cost of matching the symbols to  $\EIED(u')$ instead is $C+C^+-C^-$. 
So, $D= C+q+|\EIED(v^*_k)|$, and $w = C^+-C^-$. 

For the size of the SLP we are doubling the number of $\EIED$ gadgets, and we are adding in the $5$ and $4$ symbols. So the total size should be
$O(2xd+nd+\ell \lg(n) + 7x(\lg(p)+\lg(q))).$ Given the size of $p$ and $q$ this gives:
$O(xd+nd+ x \lg(n)).$

For the length of the strings we have at most $O(dn^\ell)$ coordinate gadgets and the number of unmatached symbols is $O(dn^{\ell})$. Note that the size of coordinate gadgets is constant when $k$ is constant and $d=n^{o(1)}$. So the total length of our generated strings $\IEDF$  and $\EIEDF$  is $O(n^{\ell+o(1)})$.
\end{proof}

Now that we have generated interleave vector gadgets we will put multiple copies of them and align them. We want to set this up using the same `sliding pyramid' set up as we used for the interleave gadgets. 

\subsection{Aligning Interleave Vector Gadgets}
For aligning our gadgets we generalize the idea from~\cite{editDistLBSeth} for aligning gadgets. First, we create a fake list of vectors $F$ that is $n$ vectors of dimension $d$ where every entry of the vectors is $1$. Then we create ``fake'' versions of the alignment and empty vector gadgets, build from $F$ instead of $L$. We concatenate the ``real'' $\IED$ and $\EIED$ gadgets with many matchable symbols in between. We surround these real gadgets with our ``fake'' gadgets. We also build a ``pyramid'' that allows the strings to have any valid alignment of the real gadgets while having the same number of matches of the fake gadgets. Around these we put an additional number of unmatchable characters (characters that are unique to each set).  See Figure~\ref{fig:slidingPyramid} for a depiction. These fake gadgets allow for any choice of alignment of the real gadgets to have the same value of matches outside of the $k$ gadgets we are matching with the alignment. 

\begin{figure}[h]
\centering
\includegraphics[scale=0.75]{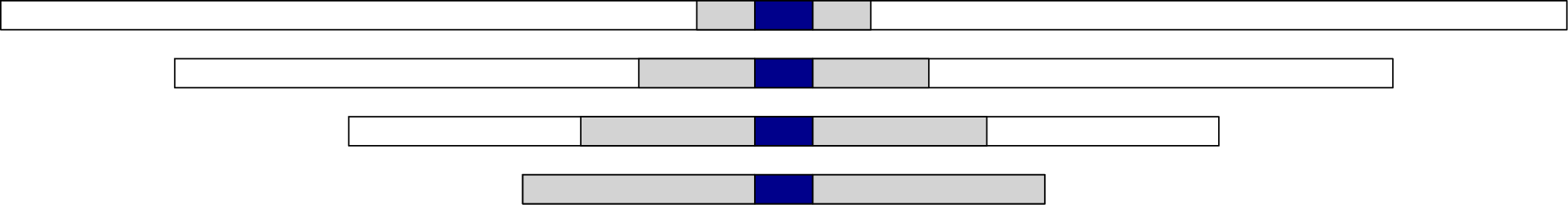}
\caption{A visual depiction of the structure of our alignment. This is using the ideas from~\cite{editDistLBSeth}. The dark blue section is a depiction are the real gadgets. The light-gray section are the ``fake'' gadgets. The white sections are unmatchable characters (a distinct character in each string).}
\label{fig:slidingPyramid}
\end{figure}

We start by proving we can create strings such that $k$-edit distance can be used to detect $(a(k-1)+bk)$-$OV$s. 

\begin{lemma}
\label{lem:EDOVgadget}
Let $a, b,$ and $k$ be positive constant integers.
Let $L$ be a list of $n$ length $d$ vectors, where $d=n^{o(1)}$.

There are $k$ strings $\EDOV_i(a,b,k,L)$ such that:
$$
\ED(\EDOV_1(a,b,k,L),\ldots,\EDOV_k(a,b,k,L))
\begin{cases}
\geq E & \text{ if there is a }(a(k-1)+bk)\text{-OV in }L \\
\leq E-1 &\text{ else }
\end{cases}
$$
The length of each of these strings is $O(n^{a+b+o(1)})$ with an alphabet of size $O(k)$.
\end{lemma}
\begin{proof}
We are given as input a $(a(k-1)+bk)-$OV instance. Say the list is $L$. It contains $n$ zero one vectors of length $d$. 

We define some new symbols. We add a new symbol $\$$, we will use this to encourage matching in lined up sections. We will also add $\%_i$ symbols for $i\in[1,k]$. A symbol $\%_i$ appears only in string $i$, thus it can not be matched, it must be deleted or substituted. 

Let us now define the real sections of these strings (the blue section at the center of the strings in Figure~\ref{fig:slidingPyramid}). First we will define this for $i\in[1,k-1]$:
$$REAL_i = \bigcirc_{j\in[1,n^b]} \IED(L,\Flist(L)_{b}[j],a)$$
and for the last string
$$REAL_k = \bigcirc_{j\in[1,n^b]} \EIED(\Flist(L)_{b}[j]). $$

Next let us define the fake gadget part of these strings. First a single fake gadget is generated by plugging in $\hat{L}_F$, a list of $n$ all ones vectors all of length $d$. And the vector $\hat{u}_F$, a length $d$ vector of all ones. For $i\in[1,k-1]$
$$F_i = \IED(\hat{L}_F,\hat{u}_F,a)$$
and then for $i=k$
$$F_k = \EIED(\hat{u}_F).$$
We can now define our fake gadgets, the gray parts of Figure~\ref{fig:slidingPyramid}:
$$FAKE_i = F_i^{in^b}$$

Now we will define the unmatched symbol sections. We will add new symbols $\&_i$ for $i\in[1,k]$. Note that $\&_i$ appears only in string $i$.
$$UNMA_i = \&_i^{|F_i|n^b (2k-i)}$$

Now let us define our gadgets for $i\in[1,k-1]$:
$$\EDOV_i(a,b,k,L) = UNMA_i \circ FAKE_i \circ REAL_i \circ FAKE_i \circ UNMA_i. $$

Now note that if there is a $(a(k-1)+bk)$-OV this corresponds to a $k$-tuple of $\IED_1,\ldots,\IED_{k-1},\EIED$ gadgets in this construction having an edit distance of $C$ (smaller than $C+w$). Additionally, note that if there is no $(a(k-1)+bk)$-OV then all $k$-tuple of $\IED_1,\ldots,\IED_{k-1},\EIED$ gadgets have an edit distance of $C+w$.  

Now, if there is a $(a(k-1)+bk)$-OV  then, we can align the gadgets and give an upper bound on total cost. First, all unmatched characters will cost $1$ so they have total cost of $|F_i|n^b(2k^2-k(k+1)/2).$ Next, consider the fake gadgets that overhang and interact with fewer than $k$ other gadgets. Let $p_i = \ED(F_k,\ldots,F_i)$. There are $2n$ fake gadgets that are aligned with $i$ total gadgets and otherwise aligned with the unmatchable characters. The cost for these is $\sum_{i\in[1,k]} 2np_i$. Finally, there are $3n$ gadgets which line up with a full $k$ other gadgets. If $m$ of these represent underlying orthogonal vectors then all $m$ tuples will have a cost of $D$, the rest will have a cost of $D+w$ by Lemma~\ref{lem:edinterleaveGadget}. This means if there is at least one match then the cost is at most:
$$E = |F_i|n^b(2k^2-k(k+1)/2) + \sum_{i\in[1,k]} 2np_i + 3n(D+w) -w.$$

What if there are no $(a(k-1)+bk)$-OVs? Well, any valid alignment (where we skip no characters and have the same size of overhangs) will cost at least $E+w$. If we don't align gadgets then at some point we are skipping $5^q$ symbols, this could potentially allow us to improve our edit distance by $2|\IED_i|+p$, however we set $q= 50p + 500|\IED_i|$ in Lemma~\ref{lem:edinterleaveGadget}. These skipped symbols increase the cost due to the unmatchable characters. When we foreshorten our string by skipping these $5^q$ symbols we sill need to pay the cost in the unmatchable characters as deletions (instead of substitutions) but we also need to pay for the deletion of the $5^q$ characters. So, if there is no valid alignment our cost is at least $E+w$, in fact it is exactly $E+w$.

The length of the generated strings is $O(n^b |\IED_i|) = O(n^bn^ad)$. Because $d=n^{o(1)}$ we can simplify this to $O(n^{b+a+o(1)})$.
\end{proof}

Next, we show that the strings we produced compress well. 

\begin{lemma}
\label{lem:EDCompressionOVGadget}
Let $k,a,b$ be constant integers. Let $d= n^{o(1)}$.
Given the $k$ strings $\EDOV_i(a,b,k,L)$ defined in Lemma~\ref{lem:OVgadget} the size of the compression is $O(n^{b+o(1)}+n^{1+o(1)})$.
\end{lemma}
\begin{proof}
We want to represent $O(n^b)$ instances of $\EIED$ and $\IED$ gadgets. By Lemma~\ref{lem:edinterleaveGadget} we have there is an SLP to represent these of size $O(n^bd+nd+ n^b a \lg(n))$. This can be simplified to $n^{b+o(1)}n^{1+o(1)}$. 

We additionally want to represent the unmatchable characters. These have a total length of $n^{b+a+o(1)}$, so there is an SLP to represent these of size $(b+a+o(1))\lg(n) = n^{o(1)}$. 

So the total size of the SLP is $n^{b+o(1)}n^{1+o(1)}$. 
\end{proof}

\subsection{Putting it all Together}
Now that we have proven the above lemmas, we can prove our desired result.\\

\begin{reminder}{Theorem~\ref{thm:editDistanceLowerBoundFromSETH}}
If the $k'$-OV hypothesis is true for all constants $k'$, then for all constant $\epsilon \in (0,1]$ grammar-compressed $k$-median edit distance requires $\left(M^{k-1}m\right)^{1-o(1)}$ time when the alphabet size is $|\Sigma|=\Theta(k)$ and $m=M^{\epsilon \pm o(1)}$. Here, $M$ and $m$ denote the total uncompressed and compressed length of the $k$ input strings respectively.
\end{reminder}
\begin{proof}
We will use Lemma~\ref{lem:EDOVgadget} and Lemma~\ref{lem:EDCompressionOVGadget}. Given an instance of $(bk+a(k-1))$-OV we can produce strings $\EDOV_1(a,b,k,L),\ldots,\EDOV_k(a,b,k,L)$ such that they have length $M=n^{a+b+o(1)}$ and $m=n^{b+o(1)}+n^{1+o(1)} = n^{b+o(1)}$ when $b\geq 1$.

Our alphabet is of size $|\Sigma|=O(k)$, so this lower bound applies as long as the size of the alphabet is $\Theta(k)$.

Now note that $M^{k-1}m = n^{(k-1)a+(k-1)b+b} = n^{(k-1)a+kb}$. So, an algorithm that runs in  $(M^{k-1}m)^{1-\eps}$ time $\eps>0$ implies an algorithm for $(bk+a(k-1))$-OV that violates our assumption. 
Thus, $k$-edit distance requires $\left(M^{k-1}m\right)^{1-o(1)}$ time given the assumption on $(bk+a(k-1))$-OV.

Now we consider a contradiction to our theorem statement. There would be an algorithm running in $(M^{k-1}m)^{1-\gamma}$ time to solve grammar compressed $k$-median edit distance when $m=M^{\epsilon \pm o(1)}$ and $\eps \in (0,1]$. In the easiest case we can pick an $a,b$ such that $b/(a+b) = \eps$, in this case we are done. For irrational $\eps$ we need to approximate and then pad the strings. Choose an $a$ and $b$ such that $\eps \leq b/(a+b) < \eps (1+\gamma/2)$. Such $a,b$ exist that are in $O(\frac{1}{\eps\gamma})$. We add a new character $3$ to our alphabet. 
Let $S'_i = S_i3^x$, where we will set $x\in[M,M^2]$ later. Note that $\ED(S_1,\ldots,S_k) =\ED(S'_1,\ldots,S'_k)$. Note that the compression of these strings is $m' = m+\lg(x) = \Theta(m)$ where as the length is $M'=M+x = \Theta(x)$. Choose $x= m^{1/\eps \pm o(1)} = M^{b/(a+b) \cdot 1/\eps}$.
So now $m' = M'^{1/\eps \pm o(1)}$. Note that $1\leq b/(a+b) \cdot 1/\eps \leq (1+\gamma/2)$. 
Now consider running the claimed fast algorithm on our new $S'_1,\ldots, S'_k$ instance.
The running time is 
$$(M'^{(k-1)}m')^{1-\gamma} = O( (M^{(1+\gamma/2)(k-1)} m)^{1-\gamma-o(1)}).$$
This running time can be simplified to $O(M^{(k-1)} m)^{(1+\gamma/2)(1-\gamma-o(1))}$. Note that $(1+\gamma/2)(1-\gamma-o(1))$ is less than $1-o(1)$. This algorithm violates the lower bound for the original $S_1,\ldots,S_k$ instance. This is a contradiction.

So any algorithm running in $(M^{k-1}m)^{1-\gamma}$ time to solve grammar compressed $k$-median edit distance when $m=M^{\epsilon \pm o(1)}$ and $\eps \in (0,1]$ violates $k'$-OV for some $k'$ that depends on $\eps,\gamma$. This implies our theorem statement. 
\end{proof}

Finally we will apply this same lower bound to $k$-center edit distance using a reduction from~\cite{editDistLBSeth}.

\subsection{k-Center Edit Distance}

In Section 3 of~\cite{editDistLBSeth}  they present a reduction from median $k$-edit distance to $k$-center distance~\cite{editDistLBSeth}. We will restate their reduction here. 

Say we are given a $k$-median edit distance instance with $k$ strings $X_1,\ldots,X_k$ where $|X_i|=N$. Then, as~\cite{editDistLBSeth} suggest, construct the following strings:
\begin{align*}
    Y_1 &= X_1~\$^N~X_2~\$^N~\cdots~\$^N~X_{k-1}~\$^N~X_k \\
    Y_2 &= X_2~\$^N~X_3~\$^N~\cdots~\$^N~X_k~~~~~\$^N~X_1 \\
    &\vdots\\
    Y_k &=X_{k}~\$^N~X_1~\$^N~\cdots~\$^N~X_{k-2}~\$^N~X_{k-1}
\end{align*}

\textbf{Claim of Section 3 in~\cite{editDistLBSeth}}: $\CED(Y_1,Y_2,\ldots,Y_k) = \ED(X_1,X_2,\ldots,X_k)$.

As a quick intuition for this claim, we have to match the $\$^N$ sections. First note that we can achieve this bound by taking a string $T$ which is one of the median edit distance minimizing strings of the $X_i$ and creating a center string for the $Y_i$ strings 
$C= T~\$^N~T~\$^N~\ldots~\$^N~T.$
Now note that the distance to this string from all $Y_i$ is $\ED(X_1,X_2,\ldots,X_k)$. Thus,  $\CED(Y_1,Y_2,\ldots,Y_k) \leq \ED(X_1,X_2,\ldots,X_k)$. For the other side of the inequality note that the center string $C$ should have shape $C= T_1~\$^N~T_2~\$^N~\ldots~\$^N~T_k$ for some  strings $T_1,\ldots,T_k$. Now note that for all $j$ 
$$\sum_{i \in [1,k]} \ED(T_j, X_i)\geq \ED(X_1,X_2,\ldots,X_k).$$
Because, the $k$-median edit distance minimizes this sum over all possible strings, and $T_j$ is simply an instantiation of a string. 
Now note that 
$$\sum_{j\in[1,k]}\sum_{i \in [1,k]} \ED(T_j, X_i)\geq k\ED(X_1,X_2,\ldots,X_k),$$
which implies
$$\sum_{j\in[1,k]} \ED(Y_j, C)\geq k\ED(X_1,X_2,\ldots,X_k).$$
So, the max over all $j$ of $\ED(Y_j, C) \geq \ED(X_1,X_2,\ldots,X_k)$. This is the definition of the center edit distance, so we have shown both sides of the inequality. \\

\begin{theorem} \label{thm:centerEditDistanceAllSameLength}
We are given $k$ strings of length $M$ with a SLP of size $m$. The $k$-center-edit-distance problem on these strings requires $\left(M^{k-1}m \right)^{1-o(1)}$ time if SETH is true. 
\end{theorem}
\begin{proof}
By Theorem~\ref{thm:editDistanceLowerBoundFromSETH}, given $k$ strings, $X_1,\ldots,X_k$, of length $N$ which all compress to length $n$, $k$-edit distance requires  $\left(N^{k-1}n\right)^{1-o(1)}$ time if SETH is true.

We use the transformation of~\cite{editDistLBSeth} and produce strings $Y_1,\ldots,Y_k$. These strings have length $M=kN$ and an SLP of size $m = kn+k+\lg(N)$. As a result $M^{k-1}m= O\left(N^{k-1}(n+\lg(N))\right)$. Thus, an algorithm that runs in $\left(M^{k-1}m\right)^{1-\eps}$ time for $k$-center edit distance for some constant $\eps>0$ implies a violation of SETH. Thus, $k$-center edit distance requires $\left(M^{k-1}m\right)^{1-o(1)}$ time. 
\end{proof}

\section{Hamming Distance and Beyond}\label{sec:Hamming}

Given $k$ equal-length strings $X_1,\ldots,X_k$ with $X_i\in \Sigma_i^N$,
we define a string $X = \bigotimes_{i=1}^k X_i \in (\bigtimes_{i=1}^k \Sigma_i)^N$
with $X[j] = (X_1[j],\ldots,X_k[j])$ for $j\in [1\dd N]$.
In this section, we show that if each string $X_i$ can be represented using a straight-line program of size $n_i$,
then $X$ can be represented using a straight-line program of size $\Oh((\prod_{i=1}^k n_i)^{1/k}N^{1-1/k})$.
Next, we apply this construction for computing Hamming distance of two grammar-compressed strings
and propose several generalizations for $k=\Oh(1)$ grammar-compressed strings.

\begin{proposition}\label{prp:prodgram}
Given $k=\Oh(1)$ straight-line programs $\G_{i}$ of sizes $n_i$ representing strings $X_i$ of the same length $N>0$, 
a straight-line program $\G$ of size $\Oh((\prod_{i=1}^k n_i)^{1/k}N^{1-1/k})$ representing $X=\bigotimes_{i=1}^k X_i$ can be constructed in time $\Oh((\prod_{i=1}^k n_i)^{1/k}N^{1-1/k})$.
\end{proposition}
\begin{proof}
We proceed based on a threshold $\tau\in [1\dd N]$ to be fixed later.
For each grammar $\G_i$, we first use \cref{lem:part} to derive grammars $\G^+_i$ and $\G^P_i$ of size $\Oh(n_i)$.

Next, we consider \emph{relevant tuples} $\F=(F_1,\ldots,F_k)$ such that:
\begin{itemize}
    \item each $F_i$ is a fragment of $\exp(A_i)$ for a symbol $A_i$ of $\G^+_i$ satisfying $|A_i|\le \tau$,
    \item $|F_1|=\cdots = |F_k|$,
    \item there exists $i_p\in [1\dd k]$ such that $F_{i_{p}}$ is a prefix of $\exp(A_{i_p})$,
    \item there exists $i_s\in [1\dd k]$ such that $F_{i_{s}}$ is a suffix of $\exp(A_{i_s})$.
\end{itemize}
The number of relevant tuples does not exceed $\tau^{k-1}\cdot k \cdot \prod_{i=1}^k n_i$,
because each $\F$ is uniquely determined by the choices of symbols $A_i$,
the choice of $i_p$, and the starting positions of $F_i$ in $\exp(A_i)$ for $i\ne i_p$.
(The common length $|F_1|=\cdots =|F_k|$ is uniquely determined due to the constraint that at least one $F_i$ is a suffix of $\exp(A_i)$.)

For each relevant tuple $\F$, we add to $\G$ a symbol $A_\F$ aiming at $\exp(A_\F) = \bigotimes_{i=1}^k F_i$.
The symbols $A_\F$ are ordered consistently with the lexicographic order of tuples $(A_1,\ldots,A_k)$
based on the order of symbols $A_i$ within each grammar $\G_i$.

If each $A_i$ is a terminal of $\G_i$, then $F_i = \exp(A_i)=A_i$ holds for each $i$,
and we set $A_{\F} = (A_1,\ldots,A_k)$ to be a terminal of $\G$.
Otherwise, we set $A_\F$ to be a non-terminal, and we need to specify $\rhs(A_\F)$.
For this, let us fix an arbitrary index $j$ such that $A_j$ is a non-terminal of $\G_j$ and $A_j\to A'_j A''_j$.
We then consider three cases:
\begin{enumerate}
    \item $F_j$ is contained within the prefix $\exp(A'_j)$ of $\exp(A_j)$.
    In this case, we set $A_\F\to A_{\F'}$, where $F'_j$ is the fragment of $\exp(A'_j)$ corresponding to $F_j$
    and $F'_i=F_i$ for $i\ne j$. 
    Note that $F_j$ cannot be a suffix of $\exp(A_j)$ and, if $F_j$ is a prefix of $\exp(A_j)$,
    then $F'_j$ is a prefix of $\exp(A'_j)$. Thus, $\F'$ is a relevant tuple.
    
    \item $F_j$ is contained within the suffix $\exp(A''_j)$ of $\exp(A_j)$.
    In this case, we set $A_\F\to A_{\F''}$, where $F''_j$ is the fragment of $\exp(A''_j)$ corresponding to $F_j$
    and $F''_i=F_i$ for $i\ne j$. 
    Note that $F_j$ cannot be a prefix of $\exp(A_j)$ and, if $F_j$ is a suffix of $\exp(A_j)$,
    then $F''_j$ is a suffix of $\exp(A''_j)$. Thus, $\F''$ is a relevant tuple.
    
    \item $F_j$ overlaps both the prefix $\exp(A'_j)$ and the suffix $\exp(A''_j)$ of $\exp(A_j)$.
    In this case, we set $A_\F\to A_{\F'}A_{\F''}$, where $F'_j$ is the suffix of $\exp(A'_j)$ overlapping $F_j$,
    $F''_j$ is the prefix of $\exp(A''_j)$ overlapping $F_j$,
    and for every $j\ne i$ we have $F_i = F'_i F''_i$ with $|F'_i|=|F'_j|$ and $|F''_i|=|F''_j|$.

    Note that $F'_j$ is a suffix of $\exp(A'_j)$ and $F''_j$ is a prefix of $\exp(A''_j)$.
    Moreover, if $F_j$ is a prefix of $\exp(A_j)$, then $F'_j$ is a prefix of $\exp(A'_j)$,
    and if $F_j$ is a suffix of $\exp(A_j)$, then $F''_j$ is a suffix of $\exp(A''_j)$.
    Finally, for $i\ne j$, if $F_i$ is a prefix of $\exp(A_i)$, then $F'_i$ is a prefix of $\exp(A_i)$,
    and if $F_i$ is a suffix of $\exp(A_i)$, then $F''_i$ is a suffix of $\exp(A_i)$.
    Thus, both $\F'$ and $\F''$ are relevant tuples. 
\end{enumerate}

Next, for each $i$, we interpret the string $P_i$ generated by $\G^P_i$
as a decomposition of $X_i$ into $|P_i| =\Oh(\frac{N}{\tau})$ phrases.
Each phrase is of the form $\exp(A)$ for a symbol $A$ of $\G^+_i$ satisfying $|A|\le \tau$.
Let $B_i$ be the set of phrase boundaries of this decomposition of $X_i$
(i.e., $B_i = \{|\exp(P_i[1\dd j])| : i\in [0\dd |P_i|]\}$),
and let $B = \bigcup_{i=1}^k B_i$.

For each string $X_i$, let us construct another partition $X_i=X_{i,1}\circ \cdots\circ X_{i,r}$ with phrase boundaries $B$ (if $0=b_0<\cdots < b_r = N$ are the elements of $B$, then $X_{i,j}=X_i(b_{j-1}\dd b_j]$).
Since $B_i \sub B$, each phrase $X_{i,j}$ can be represented as a fragment of $\exp(A_{i,j})$ for a symbol $A_{i,j}$
of $\G^+_i$ satisfying $|A_{i,j}|\le \tau$. 
Moreover, for each $j$, there exists $i_p\in [1\dd k]$ such that $X_{i_p,j}$ is a prefix of $\exp(A_{i_p,j})$
and $i_s\in [1\dd k]$ such that  $X_{i_s,j}$ is a suffix of $\exp(A_{i_s,j})$.
(This is because $b_{j-1}\in B_{i_p}$ and $b_j\in B_{i_s}$ holds for some $i_p$ and $i_s$.)
Hence, for each $j\in [1\dd r]$, there exists a relevant tuple $\F_j = (X_{i,j})_{i=1}^k$. 
Thus, it suffices to add to $\G$ a starting symbol $S \to \bigcirc_{j=1}^r A_{\F_j}$.

Due to the assumption that $k=\Oh(1)$, the total running time and the size $|\G|$ are both
$\Oh(\frac{N}{\tau} + \tau^{k-1}\prod_{i=1}^k n_i)$.
Optimizing for $\tau\in [1\dd N]$, this becomes
$\Oh(\prod_{i=1}^k n_i+(\prod_{i=1}^k n_i)^{1/k}N^{1-1/k})$.
If the first term dominates, then $\prod_{i=1}^k n_i > N$.
However, a trivial $\Oh(N)$-size straight-line program representing $X$ can be constructed in $\Oh(N)$ time by decompressing each string $X_i$. Thus, we can always construct a straight-line program representing $X$
in time $\Oh((\prod_{i=1}^k n_i)^{1/k}N^{1-1/k})$.
\end{proof}

\begin{corollary}\label{cor:additive}
Given $k=\Oh(1)$ straight-line programs $\G_i$ of size $n_i$ representing strings $X_i\in \Sigma_i^N$
of the same length $N>0$
and a function $\delta: \bigtimes_{i=1}^k \Sigma_i \to \R$ that can be evaluated in $\Oh(1)$ time,
the value $\delta(X_1,\ldots, X_k) := \sum_{j=1}^N \delta(X_1[j],\ldots, X_k[j])$ can
be computed in  $\Oh((\prod_{i=1}^k n_i)^{1/k}N^{1-1/k})$ time.
\end{corollary}
\begin{proof}
Let $X = \bigotimes_{i=1}^k X_i$ and let $\G$ be a straight-line program representing $X$ 
constructed using \cref{prp:prodgram}.
For each symbol $A$ of $\G$, we compute a value $\delta(A)$ defined as 
$\sum_{j=1}^{|A|}\delta(\exp(A)[j])$.
Note that if $A=(A_1,\ldots,A_k)$ is a non-terminal, then $\delta(A)=\delta(A_1,\ldots,A_k)$ can be evaluated in $\Oh(1)$ time.
Otherwise, if $A\to \bigcirc_{\ell=1}^r B_\ell$,
then $\delta(A)=\sum_{\ell=1}^r \delta(B_\ell)$, so $\delta(A)$ can be computed in $\Oh(|\rhs(A)|)$ time.
Consequently, constructing $\G$ and computing $\delta(A)$ for every symbol $A$ of $\G$
costs $\Oh((\prod_{i=1}^k n_i)^{1/k}N^{1-1/k})$ time in total.
This allows retrieving $\delta(X_1,\ldots,X_k)$ as the value $\delta(S)$ for the starting symbol $S$ of $\G$.
\end{proof}

In particular, we can set $\delta=\HD$ for $k=2$ (defined for characters $x,y$ with $\HD(x,y)=0$ if $x= y$ and $\HD(x,y)=1$ if $x\ne y$).
Possible generalizations to an arbitrary number of strings include the following two definitions of $\delta(x_1,\ldots,x_k)$
for characters $x_1,\ldots,x_k$:
\begin{itemize}
\item $\delta(x_1,\ldots,x_k)=0$ if $x_1=\cdots =x_k$ and $\delta(x_1,\ldots,x_k)=1$ otherwise (the straightforward generalization).
\item $\delta(x_1,\ldots,x_k)=\min_{i=1}^k \sum_{j=1}^k \HD(x_i,x_j)$ (the generalization corresponding to the median string problem).
\end{itemize}
In either case, \cref{cor:additive} allows computing $\delta(X_1,\ldots,X_k)$ in  $\Oh((\prod_{i=1}^k n_i)^{1/k}N^{1-1/k})$ time. %This proves \cref{thm:Hamming}.

\section{Shift Distance: Lower Bound \& Upper Bound}\label{sec:BestAlign}%%%%%%%%%%
%%%

In this section we will explore a problem where we can get tight upper and lower bounds, but there is no efficiency to be gained by having compressible strings in your input. The problem is \ComBkAlig. When $k=2$ this problem is (basically) equivalent to the Hamming distance substring problem mentioned in~\cite{compressedLCSSETH}. This problem is a natural extension of pattern matching. This problem asks, given a set of $k$ strings, how can we best line them up to maximize the number of matched characters? So, the alignment that minimizes the Hamming Distance between all the strings. This problem was studied in the average-case for $k=2$ by~\cite{AndoniIKH13}. They called the problem ``shift finding''. We give the natural generalization of this problem to $k$ strings and present upper bounds and lower bounds. We also present an approximation algorithm. We present these results in part because they give an example of a $k$-string comparison problem where there is no efficiency to be gained from having a compressible input. 

The core of this section is showing tight upper and lower bounds for this problem of finding the ideal alignment of strings that minimizes Hamming distance. We show that in cyclic shift there is no advantage to be gained from compression. The upper and lower bounds are tight and unchanged even with compression. We are also able to use FFT to get a fast algorithm for the problem of finding the best alignment with respect to Hamming distance. 

We will now re-state the definition of \ComBkAlig, with more commentary. 

\begin{reminder}{Definition \ComBkAlig~(\CBKA)}
We are given $k$ strings as input: $X_1, \ldots, X_k$. These strings have characters from the alphabet $\Sigma$. Each string has length $N$ and compresses via SLP to a length of $n$. For convenience of notation let $X_j[i]$ when $i \notin [0,n-1]$ refer to $X_j[i']$ where $i' \in [0,n-1]$ and $i' \equiv i \mod n$ (so we let indices ``wrap around''). 

We must return the best alignment of the $k$ strings. The alignment where in the maximum number of locations all strings have the same symbol. We will give a precise definition below. Let $\hat{\Delta} = \max(\Delta_1, \ldots, \Delta_{k-1})$. And let $\llbracket\cdot\rrbracket$ be an operator that turns $True$ to a $1$ and $False$ to a $0$.
\[\CBKA(X_1,\ldots, X_k) = \max_{\Delta_1, \ldots, \Delta_{k-1} \in [0,N-1]}  \sum_{i=1}^{N} {\biggl\llbracket}X_1[i+\Delta_1]=X_2[i+\Delta_2]=\cdots = X_{k-1}[i+\Delta_{k-1}]=X_k[i] {\biggr \rrbracket}\]
So, we want the offsets such that the maximum number of characters are all shared between all $k$ strings. 

We will also define the term of the offset score. Given strings $X_1,\ldots, X_k$ and a particular set of deltas $\Delta_1, \ldots, \Delta_{k-1}$ we will call the value: 
\[\sum_{i=1}^{N} \left[X_1[i+\Delta_1]=X_2[i+\Delta_2]=\cdots = X_{k-1}[i+\Delta_{k-1}]=X_k[i] \right] \]
the offset score of the strings $X_1,\ldots, X_k$ and the deltas $\Delta_1, \ldots, \Delta_{k-1}$.
\end{reminder}

\subsection{The Algorithm}

We will use FFT to get a fast algorithm here. We start by showing how to do this when $k=2$.

\begin{lemma}[From~\cite{compressedLCSSETH}]
There is an $O\left(|\Sigma|N\lg\left(N|\Sigma|\right)\right)$ algorithm for \CBtA~with an alphabet $\Sigma$ (\CBKA~when $k=2$).
\end{lemma}

We can then generalize to $k$ by making calls to \CBtA.

\begin{theorem}\label{thm:cycShiftUB}
There is an $O(|\Sigma|N^{k-1}\lg(|\Sigma|N))$ algorithm for \ComBkAlig.
\end{theorem}
\begin{proof}
Let our $k$ input strings be: $X_1, \ldots, X_k$, each of length $N$. 

We are going to reduce \CBKA~with alphabet $\Sigma$ and strings of length $N$ to \CBtA~with alphabet $\Sigma \cup \{@\}$ and strings of length $N$.
First, we will try all $N^{k-2}$ possible offsets $\Delta_2,\ldots, \Delta_{k-1}$. Now, for each of these we will Create a new string $Y$ which will be a ``merge'' of the strings $X_2,\ldots,X_k$. 
The string $Y$ will have length $N$. The $i^{th}$ bit of $Y$ is:
\[Y[i] =\begin{cases} X_k[i]
& \text{ if } \left[X_2[i+\Delta_2]=\cdots = X_{k-1}[i+\Delta_{k-1}]=X_k[i] \right]\\
@ & \text{ else}
\end{cases}.
 \]
 If all the strings agree given our choice of offset we set it to the agreed character. Otherwise, we use the new special character $@$ which does not appear in $X_1$ (as $@$ is not in $\Sigma$). Note that we can produce  $Y$ in $kN$ time, so  over all offsets we take $N^{k-1}$ time to produce the inputs $X_1, Y$. 
 
 Now, we will make $N^{k-2}$ calls to \CBtA. Each call takes time $(|\Sigma|+1)N\lg(N(|\Sigma|+1))$ time. Thus, we take time $O(|\Sigma|N^{k-1}\lg(N|\Sigma|))$. 
 
 This gives a total time of $O(|\Sigma|N^{k-1}\lg(N|\Sigma|))$.
\end{proof}

\subsection{The Lower Bound}
We will now show that we can produce $k$ strings of length $N=n^{a}$ that compress to length $an$ such that an algorithm that runs faster than $n^{(k-1)a-o(1)}=N^{k-1-o(1)}$ violates SETH and the $(k-1)a-$OV hypothesis. 
This will give a tight lower bound. Additionally, it says that strings that compress from length $N$ to length $N^{\epsilon}$ do not have faster algorithms than those that don't compress. 

To do this we will use the interleaving representation defined previously. Recall that we defined the interleaving version as:
\[\mathbf{String_I}_{\ell}(L) = \bigcirc_{i=1}^d \left( \bigcirc_{j_1\in[1,n] \ldots, j_\ell \in [1,n]} L[j_1][i] \cdot L[j_2][i] \cdot \cdots \cdot L[j_\ell][i] \right).\]
Recall that this is equivalent to
\[\mathbf{String_I}_{\ell}(L) = \bigcirc_{i \in [1,d]} \bigcirc_{j \in [1,n^\ell]} \Flist(L)[j][i].\]
Finally, recall that in $\mathbf{String_I}_{\ell}(L)$ the vector  $\vec{v}= \Flist(L)[i]$ appears as bits $i, i+n^k, \ldots, i+(d-1)n^k$. 

\begin{theorem}\label{thm:cycShiftLB}
Let $a$ and $k$ be constants. Let $N$ be the input string length for \CBKA. 

If the $(a(k-1))-$OV hypothesis holds then
\CBKA~requires $N^{k-1-o(1)}$ time even when the strings compress down to length $m=N^{1/a + o(1)}$ with an alphabet of size $O(1)$.
\end{theorem}
\begin{proof}
We take as input a $(a(k-1))$-OV instance with $(a(k-1))$ lists of $n$ vectors each. Each vector has length $d$ and $d =n^{o(1)}$.  
Recall that the $(a(k-1))$-OV hypothesis states that $(a(k-1))$-OV requires $n^{a(k-1)-o(1)}$ time.

We will use four characters $0,1, \%,@,*_1,\ldots, *_{2^k}$. The zero and ones will be used to signify the zeros and ones of the OV instance. The $@$ and $*_i$ symbols will be used to force alignment in a way that is easy to prove. We note that one can almost certainly prove the same result with a smaller alphabet. However, allowing this larger alphabet makes our proof much easier.

Given a $(a(k-1))$-OV instance split the lists of vectors up into $k-1$ groups each with $a$ lists of vectors. Call these groups of $a$ lists $L_1,\ldots, L_{k-1}$. We are going to form strings $X_1,\ldots,X_{k-1}$ by slight alterations to $\mathbf{String_I}_{a}(L_1), \ldots, \mathbf{String_I}_{a}(L_1)$. 
The final string $X_k$ will remain constant regardless of the input instance of OV. 

Let $\hat{X_i} = \mathbf{String_I}_{a}(L_i)$ then $i\in[1,k-1]$. 
Let $\hat{X_k}$ be a string of all zeros except in positions $0, n^a, \ldots, (d-1)n^a$. Like the other strings we give a total length of $dn^a$ for $\hat{X_k}$.
Note that by choosing an offset for each string from $\hat{X_k}$ we are effectively choosing one vector from each list $L_1, \ldots, L_{k-1}$ to align with the ones in $\hat{X_k}$. We want to design ways to right out the zeros and ones that simultaneously: (1) force alignment and (2) have the same value if there is at least one zero and a lower match value if they are all ones. If we can do this, then the best alignment will be picking the ``most orthogonal'' set of $k-1$ vectors, which will let us find if any vectors are fully orthogonal. 

We will now design $h_{1,i}$ and $h_{0,i}$ which will have the property that the offset score of $h_{b_1,1},h_{b_2,2}, \ldots, h_{b_k,k}$ with all deltas zero is $0$ if all $b_i=1$ and is $1$ otherwise. 
We will consider all strings in $\{0,1\}^k$. Let $H$ be all $2^k$ of those strings in sorted order, with the all ones string last. Let $H[j]$ be the $j^{th}$ string in $H$ note that $H[2^k]$ is the all ones string (we will one index this list).  Now
\[h_{b,i}[j] = \begin{cases}
1 & \text{ if } H[j][i]=b \text{ and } j\ne 2^k \\
0 & \text{ if } i<k \text{ and the above does not apply} \\
\% & \text{ else }
\end{cases}.\]
So we get strings of length $2^k$. Note that $h_{b,i}$ for $i\in[1,k-1]$ uses only $0,1$ symbols, however, $h_{b,k}$ uses only $0,\%$ symbols.  If we are aligning $h_{b_1,1},h_{b_2,2}, \ldots, h_{b_k,k}$ we are simply counting locations where they are all $1$. This only occurs in the location that is associated with the string in $H$ $b_1 b_2 \cdots b_k$, if it is not the all ones string. So, the offset score of $h_{b_1,1},h_{b_2,2}, \ldots, h_{b_k,k}$ with all deltas zero is $0$ if $b_i=1$ for all $i$ and is $1$ otherwise, as desired.

We will now design $T_{1,i}$ and $T_{0,i}$ that will force alignment. Let $\bigcirc$ represent concatenation:
\[T_{b,i} = \bigcirc_{j=0}^{2^k} *_j h_{b,i}[j].\]
Note that this wrapper is just adding special characters that force alignment of the bits in $h_{b,i}$ by making the only way to match the $*_j$ characters also force an alignment of the  $h_{b,i}[j]$ characters. Note that $|T_{b,i}| = 2\cdot 2^k  = \ell$.
Note that the offset score of $T_{b_1,1},T_{b_2,2}, \ldots, T_{b_k,k}$ with all deltas zero is $2^k$ if $b_i=1$ for all $i$ and is $1+2^k$.

Let $S_{1,i}$ be the representation of a $1$ in string $X_i$. Let $S_{0,i}$ be the representation of a $0$ in string $X_i$. We will set
\[S_{1,i} = @^{\ell} T_{1,i}\quad\text{and}\quad S_{0,i} = @^{\ell} T_{0,i}.\]
Note that this wrapper adds these $@$ characters which further enforce alignment. Note that the offset score of $S_{b_1,1},S_{b_2,2}, \ldots, S_{b_k,k}$ with all deltas zero is $2^k+\ell$ if $b_i=1$ for all $i$ and is $1+2^k+\ell$ otherwise.

\paragraph{Correctness}
Now, we want to claim that one of the best alignments of $X_1,\ldots, X_k$ will have deltas that are multiples of $|S_{b,i}| =  2\ell$. That is, the best alignment will align these representations of single bits. 
Consider if $\Delta_i \mod 2\ell =f$. If $f \ne 0 \mod 2\ell$ then the $*_j$ symbols can't be aligned with those in $X_k$. Additionally, at most $\ell-j$ of the $@$ characters will be matched. Giving a maximum match of:
$\ell-j+2^k$ (even if every $0$, $1$, and $\%$ characters were matched, which is of course unrealistic, we can't match $0$ characters as none appear in the $X_k$ string). This is worse than the worst alignments when $\Delta_i$s are multiples of $2\ell$.

So, the best alignment has all $\Delta_i$ as multiples of $2\ell$. Thus, the alignment of $X_1,\ldots, X_k$ is an alignment of $|\hat{X}_i|$ $S_{b,i}$ gadgets. Each gadget promises to return $2^k+\ell$ if $b_i=1$ for all $i$ and is $1+2^k+\ell$ otherwise.

Now, note that given our construction of  $\hat{X}_1,\ldots, \hat{X}_k$, if we choose a set of deltas $\Delta_i = 2\ell \delta_i$ we are effectively picking $k-1$ vectors and comparing them because of how we structured $\hat{X}_k$. So, if there are an orthogonal $k-1$ vectors which are orthogonal in our list representation (which corresponds to $a(k-1)$ vectors in the original OV instance) then we get a score of:
$|\hat{X}_1|(1+2^k+\ell)$. Otherwise, we get a score at least one less than that. This shows our reduction will give the correct answer. 

\paragraph{Time}
So with $k$ strings of length $n^a$ and a constant sized alphabet ($|\Sigma| = O(2^k)$) we can solve $(a(k+1))-$OV.\@
Notably $N= n^{a+o(1)}$. So an algorithm running in faster than $N^{k-1-o(1)}$ time will violate the $(a(k+1))-$OV hypothesis. This fulfills the statement in the theorem. 

\paragraph{Compression}
Now we will argue that these strings are compress-able with SLP.\@ We will mostly be using the same structure as~\cite{compressedLCSSETH}.
First we can build variables in our SLP for all of our base characters with $O(2^k)$ variables. Next we can build $@^{\ell}$ with $\lg(\ell) = O(k)$ variables. Next we can build all $S_{b,i}$ for all $i$ and $b$ with at most $O(k2^k)$ variables. Next, we want to build our longer strings.

Now we will use the recursive structure of $\mathbf{String_I}_{a}(L)$. Let
\[\mathbf{String_I}_{\ell}(L)^{[i]} =  \bigcirc_{j_1\in[1,n] \ldots, j_\ell \in [1,n]} L[j_1][i] \cdot L[j_2][i] \cdot \cdots \cdot L[j_\ell][i] .\]
Note that 
\[\mathbf{String_I}_{a}(L) = \bigcirc_{i=1}^d \left(\mathbf{String_I}_{a}(L)^{[i]} \right).\]
We are just pulling out the part related to the $i^{th}$ bit of every vector. 
Now note that 
\[\mathbf{String_I}_{a}(L)^{[i]} =  \bigcirc_{j\in[1,n]} \begin{cases}
\mathbf{String_I}_{a-1}(L)^{[i]} & \text{ if } L[j][i]=1\\
0^{(n^{a-1})} & \text{ if } L[j][i]=0
\end{cases}.\]
Where $0^{(n^{a-1})}$ is $n^{a-1}$ zeros in a row. 

Note that we can make SLP variables for all $0^{(n^{i})}$ strings for $i \in [1,a]$ with $a\lg(n^a) = a^2\lg(n)$ variables. Next note that given an SLP variable for $\mathbf{String_I}_{a-1}(L)^{[i]}$ we can add $n$ variables and form $\mathbf{String_I}_{a}(L)^{[i]}$. 
It takes $n$ variables to form $\mathbf{String_I}_{1}(L)^{[i]}$. So, with an SLP with $an+a^2\lg(n)$ variables we can represent $\mathbf{String_I}_{a}(L)^{[i]}$. 
So, with an SLP with $d(an+a^2\lg(n))$ variables we can represent $\mathbf{String_I}_{a}(L)$. To replace all zeros with $S_{0,i}$ and all ones with $S_{1,i}$ requires an additional $O(2^k)$ variables. 

So, we can compress all of our strings with $O(d(n+\lg(n)))$ variables. Given our restrictions on $d$ we can write this as $n^{1+o(1)}$. So our compression has length $m=n^{1+o(1)}$. Our input to our \CBKA~instance is $N=n^{a+o(1)}$. So $N^{1/a+o(1)} = n^{1+ao(1)} = n^{1+o(1)}$. Fulfilling the statement of the theorem.
\end{proof}

\subsection{Approximation Algorithm}

Let the \CBKA~distance be $k(N-\CBKA(X_1,\ldots,X_k))$. In other words, the \CBKA~distance is the total number of unmatched characters.

\begin{theorem}\label{thm:cycShiftApprox}
There is an $O\left(|\Sigma|N^{\lceil (k-1)/\ell  \rceil}\lg(|\Sigma|N)\right)$ time algorithm to get an $\ell$ approximation  of the \ComBkAlig~distance for any integer $\ell\geq 2$.
\end{theorem}
\begin{proof}
Partitions the $k-1$ of the strings into $\ell$ groups $G'_1, \ldots, G'_\ell$ which each contain as close to $(k-1)/\ell$ strings as possible, the maximum number of strings in each group is $\lceil (k-1)/ell \rceil$.
Now, take the final string, $S_k$ and add it to all the sets to make new sets $G_1, \ldots, G_\ell$, now the maximum number of strings in each group is $\lceil(k-1)/ell +1\rceil $.

On each of these partitions run the algorithm for \ComBkAlig. The time for this is\\
$O \left(\ell |\Sigma|N^{\lceil (k-1)/\ell +1 \rceil -1}\lg(|\Sigma|N) \right)$ and $\ell$ is a constant. Now, using the value of \ComBkAlig~we can compute the \ComBkAlig~distance. Let the distances of the sets of strings in $G_1, \ldots, G_\ell$ be $\Delta_1,\ldots, \Delta_\ell$. Now, note that these call be framed as distances to the last string $X_k$. So, the distance of all these strings together is at most $\Delta_1+ \cdots + \Delta_\ell$ and is at least $\max(\Delta_1,\ldots, \Delta_\ell)$. Finally, note that
\[1 \leq \frac{\Delta_1+ \cdots + \Delta_\ell}{\max(\Delta_1,\ldots, \Delta_\ell)} \leq \ell.\]
As a result there is an approximation factor of $\ell$ and a running time
$O\left(|\Sigma|N^{\lceil (k-1)/\ell \rceil}\lg(|\Sigma|N)\right)$.
\end{proof}

\section{On High-Dimensional Generalizations of DIST Matrices}\label{sec:disttensor}

Many of the crucial properties of DIST matrices derived in, e.g.,~\cite{Tiskin15} used for two-string algorithms rely on the Monge property. For LCS, the Monge property is that given two strings $X, Y$ and the alignment graph $\G_{X, y}$ then letting $d(u, v)$ be the longest path from $u$ to $v$, we have $d(v_{0, i}, v_{|X|, j}) + d(v_{0, i-1}, v_{|X|, j+1}) \leq d(v_{0, i-1}, v_{|X|, j}) + d(v_{0, i}, v_{|X|, j+1})$. For example, in this paper we used the ability to take min-plus products of unit Monge matrices efficiently, and our use of the SMAWK algorithm was enabled by the Monge property. 

However, it appears no analogous property holds for even DIST ``3-tensors'', the three-string generalization of DIST matrices. Intuitively, this is because it is not possible to enforce that any path from $v_1$ to $v_2$ intersects any path from $v_3$ to $v_4$ for four distinct vertices $v_1, v_2, v_3, v_4$, unlike in the two-dimensional alignment graph. We will use LCS as the metric for our examples here, but one can find similar examples for edit distance.

For example, let $A[i_1, i_2, j_1, j_2]$ be the longest path length from $v_{0, i_1, i_2}$ to $v_{n_1, j_1, j_2}$ in the three-dimensional alignment graph of three strings. An analog of the Monge property in three dimensions might be:

\[A(i_1, i_2, j_1, j_2) + A(i_1 - 1, i_2, j_1 + 1, j_2) \leq A(i_1 - 1, i_2, j_1, j_2) + A(i_1, i_2, j_1 + 1, j_2)\]

However, this does not seem true in general. Consider the following example, where there are two sets of length 1 edges. The first (in blue) has $\ell$ such edges, and is contained entirely between ``layer'' $i_1$ and $j_1$ of the DAG.\@ The second (in red) has $\ell+1$ edges, however two of these edges are outside the part of the DAG between $(0, i_1, i_2)$ and $(m, j_1, j_2)$. 

\begin{center}
    \begin{tikzpicture}[scale=3]
        \draw[thick] (0.65,0.65) rectangle (1.65,1.65);
        \draw[thick] (0,1) -- (0.65, 1.65) (0,0) -- (0.65,0.65) (1,0) -- (1.65, 0.65);

        \draw[very thick, red,-latex] (1.345, 0.6825) -- node[pos=0.6,left]{$1$} (1.45, 0.525);
        \filldraw[fill opacity = 0.85, black!30] (0, 0.15) -- (1, 0.15) -- (1.65, 0.8) -- (0.65, 0.8) -- cycle;
        \draw[very thick, blue,-latex] (0.3, 1.05) -- node[pos=0.39,left]{$\ell$}(0.7, 0.25);
        \draw[very thick, red,-latex] (0.855, 1.4175) -- node[pos=0.65,right]{$\ell-1$} (1.345, 0.6825);
        \filldraw[fill opacity = 0.85, black!30] (0, 0.85) -- (1, 0.85) -- (1.65, 1.5) -- (0.65, 1.5) -- cycle;
        \draw[very thick, red,-latex] (0.75, 1.575) -- node[pos=0.15,right]{$1$} (0.855, 1.4175);

        \draw[thick] (0,0) rectangle (1,1) (1,1) -- (1.65, 1.65);

        \draw[latex-] (1,0) -- (1.3,0) node[right] {$(m,j_1+1,j_2)$};
        \draw[latex-] (1,0.15) -- (1.3,0.15) node[right] {$(m,j_1,j_2)$};

        \draw[latex-] (0.65,1.65) -- (0.35, 1.65) node[left] {$(0,i_1-1,i_2)$};
        \draw[latex-] (0.65,1.5) -- (0.35, 1.5) node[left] {$(0,i_1,i_2)$};

    \end{tikzpicture}
\end{center}

The two sets of length 1 edges are positioned such that one cannot use a ``blue'' and ``red'' edge in the same path. Now, we have that $A(i_1-1, i_2, j_1+1, j_2) = \ell+1$ and all other terms in the above inequality are $\ell$. So the above inequality would say $2\ell+1 \leq 2\ell$, which is false. This can be generated by, e.g., the strings $X_1 = aabbb, X_2 = bbbaa, X_3 = baabb$; the LCS of the first two strings with $X_3[2\dd4], X_3[1\dd4], X_3[2\dd5]$ is $aa$, but the LCS of the first two strings with $X_3[1\dd5]$ is $bbb$.

While one can find other generalizations and even weakened versions of the Monge property which this example satisfies, for all the ones that we have considered there are three-string counterexamples that show they do not hold in general. 

For example, the unit Monge property also says that given a DIST matrix, if we subtract every row from the next row and every column from the next column, we get a permutation matrix. In other words, each row and column only differs in behavior from the previous row/column by 1 entry. However, for DIST 3-tensors, consider the two-dimensional ``slice'' $A$ for which $A[i,j]$ gives the path length between e.g. $(0, 0, i)$ and $(|X_1|, |X_2|, j)$. By looking at the DIST 3-tensors of even just three random strings of length roughly 100, we found that, e.g., for some sampled strings, $A$ had a row that could be expressed as a linear function, but the next row of $A$ was a piecewise linear function with six different pieces. 

As another example, consider the following weaker ``monotone'' property: $A$ is monotone if for any vector $b$, letting $m(i) = \argmin_j A[i, j] + b[i]$ and choosing the lowest value of $j$ to break ties, $m(i)$ is a monotonic function of $i$. This admits a divide and conquer algorithm for computing $\min_j A[i, j] + b[i]$ for all $i$ in accesses to $A$ near-linear in the number of $i$ (as opposed to the SMAWK algorithm using linear accesses), a primitive that is useful in dynamic programming algorithms for two-string similarity. Informally, knowing $\argmin_j A[i, j]$ lets us rule out a constant fraction of the possibilities for $\argmin_j A[i', j]$ for $i' \neq i$. The 3-dimensional generalization of this primitive would be to compute $\argmin_{i_1, i_2} A[i_1, i_2, j_1, j_2] + B[i_1, i_2]$ given access to entries of the DIST 3-tensor $A[i_1, i_2, j_1, j_2] = d(v_{0, i_1, i_2}, v_{|X_1|, j_1, j_2})$, and a matrix $B$. Put more simply, the rows of this slice have far less structural similarity to each other than the rows of a DIST matrix.

A weak generalization of the monotone property that would admit a similar divide and conquer algorithm for this problem is: knowing $i* = \argmin_{i_1, i_2} A[i_1, i_2, j_1, j_2] + B[i_1, i_2]$ lets us eliminate possibilities for $\argmin_{i_1, i_2} A[i_1, i_2, j_1', j_2'] + B[i_1, i_2]$ for $(j_1', j_2')$ that are in a given ``direction'' from $i^*$ if $(j_1', j_2')$ is in a given ``direction'' from $(j_1, j_2)$. Here, by in a given direction, we mean e.g. $j_1' \leq j_1$ and $j_2' \leq j_2$, or any of the four possibilities given by reversing neither, one, or both of these inequalities. Unfortunately, even considering random strings of length 10, we found counterexamples to each of the variants of this property given by choosing any pair of directions to slot in to the definition.

\section{Open Questions}
We find many novel lower bounds and upper bounds in this paper. However, some of these are not tight. We give some open problems below whose resolution we think would be particularly interesting. 

\begin{itemize}
    \item For solving $k$-edit distance or $k$-LCS on strings where $k\geq 3$,  we have a lower bound of $N^{k-1}n$ where $N$ is the length of the strings and $n$ is the size of the SLP. However, the best exact algorithms require $O(N^k)$ time. Can this gap be closed for any $k \geq 3$? Can this gap be closed for all constant~$k$?
    \item There are no tight lower bounds for approximating $k$-LCS and $k$-edit distance. Can we give a tight lower bound?
    \item The lower bounds for $k$-\emph{center} edit distance and the upper bounds do not match. Our lower bounds for $k$-\emph{center} edit distance are the same as those for $k$-median edit distance. However, $k$-center edit distance has slower algorithms. For example in the uncompressed and exact case the $k$-center edit distance lower bounds are $\Omega(N^k)$~\cite{editDistLBSeth}, but the best algorithm requires $\Tilde{O}(N^{2k})$ time~\cite{NicolasRivals05}. 
    %\item LZ78 - DP algorithms that don't rely on DIST matrix implementations?
    %\item Is computational hardness for LZ78 possible?
    %\item Can the LCS lower bound be extended to gap hardness? [Also take a look at the closest bichromatic pair problem from Abboud and Rubinstein's Distributed PCP paper].
    %\item Shift finding average case analysis? For the shift finding algorithm over two strings,
    %how the previous paper gets rid off the extra "X" factor that we have?
   % \item Approximation instead of exact computation even with simple LZ78 compression. Can we shed more dimensions?
\end{itemize}

In general, the space of multiple string comparison seems under-explored. We hope more work will happen in the space of algorithms and lower bounds for multiple string comparison. Specifically if there are efficient algorithms for the problem of comparing multiple strings with approximation for example, it will have significant impacts for multiple sequence alignment in biology.

\bibliographystyle{alphaurl}
\bibliography{references}

\end{document}